\documentclass[11pt,a4paper]{article}

\usepackage[a4paper,top=1in,bottom=1in,left=1in,right=1in]{geometry} 
\usepackage{setspace}
\usepackage[authoryear]{natbib}  
\usepackage{amsmath,amssymb,amsfonts,amsthm,bm}
\usepackage{graphicx,color,xcolor}
\usepackage{multirow,booktabs,array}
\usepackage{float,subfigure,rotating}
\usepackage{algorithm,algorithmicx,algpseudocode}
\usepackage{url,hyperref}
\usepackage{lipsum}
\usepackage{indentfirst}
\usepackage{authblk}
\usepackage{textcomp}
\usepackage{threeparttable}
\usepackage{pdflscape}
\usepackage{enumerate}
\usepackage{amsmath,amssymb,amsthm}
\usepackage{mathrsfs}
\usepackage{bm}
\usepackage{array}
\usepackage{booktabs}
\usepackage{tabularx}
\usepackage{caption}
\usepackage{geometry}
\usepackage{lastpage}
\usepackage{placeins}
\usepackage{adjustbox}
\usepackage{caption}

\usepackage{booktabs}
\usepackage{caption}
\usepackage{amsfonts}
\usepackage{amsmath}
\usepackage{tabularx}
\usepackage{bm}

\newtheorem{assumption}{Assumption}
\newtheorem*{Remark}{Remark} 
\newcommand{\trans}{^{\top}}

\def\sign{\ensuremath{\mathrm{sign}}}
\def\Err{\ensuremath{\mathrm{Err}}}
\def\tr{\ensuremath{\mathrm{tr}}}

\def\diag{\ensuremath{\mathrm{diag}}}
\def\var{\ensuremath{\mathrm{Var}}}

\def\rank{\ensuremath{\mathrm{rank}}}
\newcommand{\ind}{\mathbf{1}}

\newtheorem{theorem}{Theorem}
\newtheorem{lemma}{Lemma}

\newtheorem{corollary}{Corollary}
\newtheorem{proposition}{Proposition}

\date{\today}  

\hypersetup{
	colorlinks=true,
	linkcolor=blue,
	anchorcolor=blue,
	citecolor=blue,
	urlcolor=blue,
	CJKbookmarks=true
}

\begin{document}
	
	\title{Transfer Learning for Linear Discriminant Analysis with a Shared Classification Signal}
	
	\author{Yonghan Zhang$^1$}
	\author{Yimeng Fan$^1$}
	\author{Wenya Luo$^2$}
	\author{Jiang Hu$^1$}
	\affil{$^1$Northeast Normal University and $^2$Zhejiang University of Finance and Economics}
	\date{\today}
	\maketitle
	
	\begin{abstract}
		This paper studies transfer learning for linear discriminant analysis in high-dimensional two-class classification. We consider one target domain and several source domains, where the mean difference in each domain is decomposed into a deterministic common component and a domain-specific random deviation. The common component represents a shared classification signal across domains, while the random deviation captures domain-specific heterogeneity. Under spiked covariance models, we derive deterministic limits for the target-domain Gaussian-calibrated error of weighted transfer classifiers under both homogeneous and heterogeneous covariance settings. These limits quantify the effects of the shared signal, domain-specific variation, dimension-to-sample-size ratios, and spike structures on transfer performance. They further lead to oracle transfer weights and consistent data-driven plug-in estimators. We also characterize the intercept bias induced by unbalanced target-domain class sample sizes and provide an asymptotically optimal correction.
	\end{abstract}
	
	
	\section{Introduction}\label{sec:introduction}
	
	Linear discriminant analysis (LDA) is a classical method for classification. 
	Since the work of \cite{fisher1936USE}, it has been widely used because of its simple form and clear interpretation. 
	In the two-class case, the Bayes rule depends on the mean difference between the two classes and the inverse covariance matrix. 
	This structure makes LDA useful in many statistical and scientific applications. 
	For example, discriminant analysis and related dimension reduction methods have been used in image retrieval, gene-expression-based cancer classification, clinical outcome prediction, and text data analysis; see, for example, \cite{swets1996Usinga,golub1999Molecular,pomeroy2002Predictiona,park2003Lowera}.
	
	Classical LDA works well when the dimension is small compared with the sample size. 
	This condition is often not satisfied in modern data analysis. 
	When the number of variables is comparable to, or larger than, the sample size, the sample covariance matrix becomes unstable. 
	The resulting empirical LDA rule may then have poor classification performance. 
	This phenomenon has been studied in several high-dimensional settings; see, for example, \cite{bickel2004Theorya} and \cite{shao2011Sparsea}. 
	A common way to improve stability is to regularize the covariance estimator. 
	Regularized discriminant analysis was introduced by \cite{friedman1989Regularizeda}, and its high-dimensional behavior has been further studied using random matrix theory. 
	In particular, \cite{dobriban2018Highdimensionala} and \cite{wang2018Dimensiona} derived asymptotic approximations for regularized LDA under different assumptions.
	
	These studies show that the classification error is affected by the dimension-to-sample-size ratios and the regularization parameter. 
	They also suggest that the covariance structure may play an important role in high-dimensional classification. 
	In particular, the mean difference may align with some eigendirections of the covariance matrix, and such alignment can affect the classification error. 
	Spiked covariance models provide a useful way to describe this type of structure; see \cite{johnstone2001Distributiona}.
	Under this model, \cite{li2025SpectrallyCorrected} proposed a spectrally corrected regularized LDA method. 
	Further, \cite{zhang2025Structurala} studied the structural effect of covariance eigenvectors and proposed a spectral enhancement method for high-dimensional discriminant analysis.
	
	Another challenge is that the target domain may contain only a limited number of labeled samples. 
	In many applications, data from related source domains are available. 
	This situation is common in biomedical and genomic studies, where data may be collected from different cohorts, sites, or tissues. 
	Transfer learning aims to use such source data to improve learning in the target domain; see \citep{pan2010Survey,weiss2016Survey}.
	In high-dimensional statistics, transfer learning has been studied for linear regression by \cite{li2022Transfer} and for random-coefficient ridge regression by \cite{zhang2028Transfer}. 
	For classification, \cite{zhang2025Transfera} recently proposed a transfer learning framework for regularized LDA. Their method combines target and source discriminant directions through suitable weights and derives the corresponding high-dimensional classification error.

	The closest work to ours is \cite{zhang2025Transfera}, which studies transfer learning for regularized LDA under a random-effects model. In their framework, the class-separation vector in each domain is a zero-mean random vector, and the relatedness between the target and source domains is described by the correlations among these random effects. This model is useful for quantifying how similar different discriminant directions are on average. However, since the transferable information is placed entirely in random mean differences, this formulation does not explicitly capture the stable discriminative information that is genuinely shared across domains.
	
	Such a shared-signal structure naturally arises in multi-site biomedical studies. For example, when different hospitals or imaging centers collect data for the same disease classification task, the disease-related biological or imaging pattern may be stable across sites and can serve as a common discriminative signal. At the same time, differences in acquisition protocols, measurement platforms, patient populations, and preprocessing procedures may introduce additional site-specific deviations. Thus, the discriminative structure across domains is not purely random; it may consist of a transferable common component together with domain-specific perturbations.

	This paper is built on the premise that transferability should be modeled directly at the level of the discriminative signal. We therefore decompose the mean difference in each domain as
	\[
	\bm\mu_k=\bar{\bm\mu}+\bm\delta_k,
	\]
	where the deterministic component \(\bar{\bm\mu}\) represents the shared classification signal and the random component \(\bm\delta_k\) represents domain-specific heterogeneity. This formulation differs from existing random-domain models in a crucial way: it makes the transferable signal an explicit object of inference and allows its contribution to the target-domain classification error to be separated from the variability introduced by individual source domains. As a result, the model provides a theoretical framework for understanding when a source domain is beneficial, when it is harmful, and how different sources should be weighted.
	
	The main contribution of this paper is to develop a high-dimensional transfer learning theory for LDA that explicitly separates the transferable classification signal from domain-specific heterogeneity. Under spiked covariance structures, we derive deterministic equivalents for the target-domain Gaussian-calibrated error of weighted transfer classifiers in both homogeneous and heterogeneous covariance settings. These formulas identify the roles of the shared signal, domain-specific deviations, classwise dimension-to-sample-size ratios, sample spike-eigenvector bias, and cross-domain spike alignment. They also yield closed-form oracle transfer weights and consistent plug-in estimators. In addition, we characterize the intercept bias caused by unbalanced target-domain class sample sizes and provide an asymptotically optimal correction.

	The rest of this paper is organized as follows. Section~\ref{sec2} introduces the model and the proposed classifiers. Subsection~\ref{sec2.1} studies the homogeneous covariance case, and Subsection~\ref{sec2.2} studies the heterogeneous covariance case. Section~\ref{sec:simulation} presents simulation studies. Section~\ref{sec:experiments} presents empirical studies on two real multi-site biomedical datasets. Section~\ref{sec:conclusion} concludes the paper. All technical proofs are given in the Appendix.
	\paragraph{Notation.}
	Throughout the paper, \(\|\cdot\|\) denotes the Euclidean norm for vectors and the spectral norm for matrices. 
	For a square matrix \(\bm A\), \(\lambda_{\min}(\bm A)\), \(\tr(\bm A)\), and \(\rank(\bm A)\) denote its smallest eigenvalue, trace, and rank, respectively. 
	We write \(\bm A^\top\) for the transpose of \(\bm A\), \(\bm I_p\) for the \(p\times p\) identity matrix, and \(\diag(\cdot)\) for a diagonal matrix. 
	For an event \(E\), \(\ind_E\) denotes its indicator. 
	The notation \(\bm a\propto \bm b\) means that \(\bm a\) is equal to a positive scalar multiple of \(\bm b\). 
	We use \(o_p(1)\), \(O_p(1)\), \(\xrightarrow{p}\), and \(\xrightarrow{\mathrm{a.s.}}\) in their usual asymptotic senses. Unless otherwise stated, all model parameters other than the dimension and the sample sizes are treated as fixed constants independent of the asymptotic index.
	
	\section{Transfer learning for linear discriminant analysis}\label{sec2}
	We study transfer learning for linear discriminant analysis (TLDA) in a high-dimensional setting with one target domain and several source domains. Our goal is to understand when information from the source domains can improve classification in the target domain, and how this gain depends on the mean structure and the covariance structure across domains.
	
	We consider a two-class transfer learning setting. 
	The target domain is indexed by \(K\), and the \(K-1\) source domains are indexed by \(k=1,\dots,K-1\). The number of domains \(K\) is fixed and \(K\ge2\).
	For domain \(k\), let \(\bm\mu_{+1,k}\) and \(\bm\mu_{-1,k}\) denote the two class mean vectors, and let \(\bm\Sigma_k\) denote the within-domain common covariance matrix.
	For each domain \(k=1,\dots,K\), let
	$
	n_k=n_{k,+}+n_{k,-},
	$
	where \(n_{k,+}\) and \(n_{k,-}\) denote the numbers of observations from classes \(+1\) and \(-1\), respectively. 
	
	For \(i=1,\dots,n_k\), \((\bm X_k)_i\) is a \(p\)-dimensional observation with label \((y_k)_i\in\{-1,+1\}\). We write
	$
	\bm X_k=\left((\bm X_k)_1,\dots,(\bm X_k)_{n_k}\right)^\top
	$
	as the \(n_k\times p\) sample matrix, and
	$
	\bm y_k=((y_k)_1,\dots,(y_k)_{n_k})^\top
	$
	as the corresponding label vector.
	
	For the target domain, the Bayes optimal prediction direction takes the form $\bm d_{\mathrm{Bayes}}=\bm\Sigma_K^{-1}(\bm\mu_{+1,K}-\bm\mu_{-1,K})$. The Bayes prediction for a test observation $\bm x$ from the target domain is
	\begin{align*}
		y(\bm{x})
		=
		\sign\left\{
		\bm d_{\mathrm{Bayes}}\trans\left(\bm x-\frac{\bm\mu_{+1,K}+\bm\mu_{-1,K}}{2}\right)\right\}.
	\end{align*}
	The empirical classifier with $\bm{\widehat d}$ and $\widehat b$ is written as
	\begin{align*}
		\hat y(\bm{x};\bm{\widehat d},\widehat b)
		=
		\sign\left\{
		\bm{\widehat d}\trans\bm x+\widehat b\right\}.
	\end{align*}
	
	Throughout the theoretical analysis, the target-domain test distribution is taken to have equal class priors. Under the Gaussian mixture model, the exact misclassification error of the empirical classifier is
	\begin{align}\label{Err}
		\Err(\widehat{\bm d},\widehat b)
		=
		\frac12 \Phi\left(
		\frac{\widehat{\bm d}^{\top}\bm\mu_{-1,K}+\widehat b}
		{\sqrt{\widehat{\bm d}^{\top}\bm\Sigma_K\widehat{\bm d}}}
		\right)
		+
		\frac12 \Phi\left(
		-\frac{\widehat{\bm d}^{\top}\bm\mu_{+1,K}+\widehat b}
		{\sqrt{\widehat{\bm d}^{\top}\bm\Sigma_K\widehat{\bm d}}}
		\right),
	\end{align}
	where $\Phi$ denotes the cumulative distribution function of the standard normal distribution.

	Under non-Gaussian models, \eqref{Err} need not coincide with the exact misclassification error. It may nevertheless be interpreted as the Gaussian calibration of the corresponding standardized margins. In particular, the arguments of \(\Phi\) in \eqref{Err} are the signal-to-noise ratios of the linear scores, and related \(\Phi\)-type margin criteria have been used in high-dimensional analyses of transfer and robust discriminant rules \citep{zhang2025Transfera,auguin2021Large}. Throughout the paper, we analyze \eqref{Err} as a Gaussian-calibrated error criterion.

	\begin{assumption}\label{ass1}
		For each domain $k=1,\dots,K$, the design matrix $\bm X_k\in\mathbb{R}^{n_k\times p}$ is generated as
		\[
		\bm X_k=
		\left(\bm\mu_{(\bm y_k)_1,k},\dots,\bm\mu_{(\bm y_k)_{n_k},k}\right)\trans
		+\bm Z_k\bm \Sigma_k^{1/2},
		\]
		where $\bm Z_1,\dots,\bm Z_K$ are mutually independent random matrices. For each \(k\), the entries $(Z_k)_{ij}$ are i.i.d. with $\mathbb E (Z_k)_{ij}=0, \mathbb E (Z_k)_{ij}^2=1$, and $\mathbb E (Z_k)_{ij}^4<\infty$.
	\end{assumption}
	\begin{Remark}
		Assumption~\ref{ass1} gives a simple common data model for all domains. The domains may differ through their mean vectors and covariance matrices, while the noise part remains independent across samples and domains.
	\end{Remark}
	To describe the relatedness across domains, we define the mean-difference vectors $\bm{\mu}_k=\bm{\mu}_{+1,k}-\bm{\mu}_{-1,k}$ and propose the following assumption.
	\begin{assumption}\label{ass2}
		For each domain \(k=1,\dots,K\), the mean-difference vector \(\bm\mu_k\in\mathbb R^p\) is modeled as
		\begin{align}\label{eq:mudiff_model}
			\bm{\mu}_k=\bar{\bm{\mu}}+\bm{\delta}_k,
			\qquad k=1,\dots,K,
		\end{align}		
		where $\bar{\bm{\mu}}$ is a deterministic vector with bounded Euclidean norm, and $\bm{\delta}_1,\dots,\bm{\delta}_K$ are mutually independent and independent of $\bm Z_1,\dots,\bm Z_K$. Moreover, each $\bm{\delta}_k$ has i.i.d. coordinates with
		\[
		\mathbb{E}((\delta_k)_i)=0,\qquad
		\mathbb{E}((\delta_k)_i^2)=\alpha_k^2p^{-1},\qquad
		\mathbb{E}(|(\delta_k)_i|^{4})\leq Cp^{-2}
		\]
		for some $C>0$ and \(\alpha_1,\dots,\alpha_K\) independent of \(p\).
	\end{assumption}
	\begin{Remark}
			Assumption~\ref{ass2} separates the shared part and the domain-specific part of the mean difference. The deterministic vector \(\bar{\bm\mu}\) represents the stable classification signal shared by all domains, while the random vector \(\bm\delta_k\) represents the domain-specific deviation from this common direction. Randomizing the unknown signal or coefficient vector is a common way to obtain explicit high-dimensional limits. For example, \cite{dobriban2018Highdimensionala} use a random-effects formulation to describe the population signal in high-dimensional regression and classification; \cite{zhang2025Transfera} randomizes the class-separation vector in each domain; and \cite{zhang2028Transfer} use random regression coefficients in transfer regression models. In our model, the randomization is used only for the domain-specific deviations \(\bm\delta_k\), while the common component \(\bar{\bm\mu}\) is kept deterministic. This distinction allows us to retain the tractability of random-effects modeling while explicitly preserving a stable transferable classification signal.
	\end{Remark}
	
	\subsection{Homogeneous covariance case}\label{sec2.1}
	
	We first consider the homogeneous covariance case, where all domains share the same covariance matrix. This setting allows us to isolate the effect of the shared classification signal on the target-domain Gaussian-calibrated error, as the common covariance structure simplifies the expressions for the signal and variance components.
	\begin{assumption}\label{ass3}
		The covariance matrix is common across all domains and follows a spiked covariance model:
		\begin{align}\label{eq:homoSigma}
			\bm{\Sigma}_1=\cdots=\bm{\Sigma}_K=\bm{\Sigma}
			=
			\sigma^2
			\left(
			\bm{I}_p+\sum_{j\in\mathcal I}\lambda_j\bm{v}_j\bm{v}_j\trans
			\right),\qquad
			\mathcal I:=\mathcal I_+\cup\mathcal I_-,
		\end{align}
		where $1/c\le \sigma^2\le c$ for some constant $c>0$, $\{\bm{v}_j:j\in\mathcal I\}$ is an orthonormal family in $\mathbb{R}^p$, and $\mathcal I_+=\{1,\dots,r_+\},\mathcal I_-=\{-r_-,\dots,-1\}$ with fixed $r=r_++r_-$.
	\end{assumption}
	\begin{Remark}
		This is a standard spiked covariance model. For simplicity, in the theoretical analysis, we assume that $\sigma^2$, $r_+$, $r_-$, and $\lambda_j$, $j\in\mathcal I$, are known. In experiments, these population spectral parameters can be estimated consistently using existing methods for spiked covariance models; see, for example, \cite{bai2012Estimation,jiang2021Generalized}.
	\end{Remark}
	
	\begin{assumption}\label{ass4}
		For each $k=1,\ldots,K$ and $s\in\{+,-\}$, assume that $p,n_{k,s}\to\infty$ and $\widehat\gamma_{k,s}:=p/n_{k,s}\to\gamma_{k,s}\in(0,\infty)$.
	\end{assumption}
	Let $n=\sum_{k=1}^K n_k$. Under Assumption~\ref{ass4}, we have $\widehat\gamma_k:=p/n_k\to\gamma_k$ and $\widehat\gamma:=p/n\to\gamma$, where $1/\gamma_k=1/\gamma_{k,+}+1/\gamma_{k,-}$ and $1/\gamma=\sum_{k=1}^K1/\gamma_k$. Moreover, define $\widehat\tau_k:=\widehat\gamma_{k,+}+\widehat\gamma_{k,-}$ and $\widehat\Delta_k:=\widehat\gamma_{k,+}-\widehat\gamma_{k,-}$. Then $\widehat\tau_k\to\tau_k$ and $\widehat\Delta_k\to\Delta_k$, where $\tau_k:=\gamma_{k,+}+\gamma_{k,-}$ and $\Delta_k:=\gamma_{k,+}-\gamma_{k,-}$.

	\begin{assumption}\label{ass5}
		For each $j\in\mathcal I$, $\lambda_j$ is independent of $p$, and
		\[
		\lambda_1>\cdots>\lambda_{r_+}>\sqrt{\gamma}>-\sqrt{\gamma}>
		\lambda_{-r_-}>\cdots>\lambda_{-1}>-1,
		\]
		Moreover, if $\gamma\geq1$, $\mathcal I_-=\varnothing$.
	\end{assumption}
	Under Assumptions~\ref{ass3}--\ref{ass5}, all domains share the same spiked covariance structure, so it is natural to pool the covariance information across domains. We therefore use the pooled sample covariance matrix to estimate the common eigenspace, and then combine the domainwise sample mean differences through a weighted transfer direction.
	
	The sample mean vectors and the pooled sample covariance matrix are written as
	\begin{align*}
		\widehat{\bm{\mu}}_{\pm1,k}
		=
		\frac{1}{n_{k,\pm1}}
		\sum_{i:(y_k)_i=\pm1}(\bm X_k)_i,
		\quad
		\bm S_n
		=
		\frac{1}{n-2K}
		\sum_{k=1}^K\sum_{i=1}^{n_k}
		\Bigl[
		(\bm{X}_k)_i-\widehat{\bm{\mu}}_{(y_k)_i,k}
		\Bigr]
		\Bigl[
		(\bm{X}_k)_i-\widehat{\bm{\mu}}_{(y_k)_i,k}
		\Bigr]\trans.
	\end{align*}
	Write the spectral decomposition
	\[
	\bm S_n=\sum_{j=1}^p \ell_j\widehat{\bm v}_j\widehat{\bm v}_j^\top,
	\]
	with the eigenvalues arranged in decreasing order. For later use, we relabel the isolated sample spike eigenvectors by the corresponding population spike indices. That is, for \(j\in\mathcal I_+\), \(\widehat{\bm v}_j\) denotes the sample eigenvector associated with the outlier above the upper edge, while for \(j\in\mathcal I_-\), \(\widehat{\bm v}_j\) denotes the sample eigenvector associated with the outlier below the lower edge. Using the known spike structure, we define the spectrally-corrected covariance estimator
	\begin{align}\label{eq:Sigma_tilde_homo}
		\widehat{\bm{\Sigma}}
		=
		\sigma^2
		\left(
		\bm{I}_p+\sum_{j\in\mathcal I}\lambda_j\widehat{\bm{v}}_j\widehat{\bm{v}}_j\trans
		\right),
	\end{align}
	and the local discriminant direction in domain $k$ by $\widehat{\bm{d}}_k=\widehat{\bm{\Sigma}}^{-1}\widehat{\bm{\mu}}_k$, where $\widehat{\bm{\mu}}_k=\widehat{\bm{\mu}}_{+1,k}-\widehat{\bm{\mu}}_{-1,k}$.
	
	The transfer discriminant direction is
	\begin{align}\label{eq:dhat_w_homo}
		\widehat{\bm{d}}(\bm{w})=\sum_{k=1}^K w_k\widehat{\bm{d}}_k,
	\end{align}
	where $\bm{w}=(w_1,\dots,w_K)\trans\in\mathbb{R}^K$. Then, the resulting target-domain classifier is
	\begin{align}\label{eq:classifier_homo}
		\widehat{y}(\bm{x})
		=
		\sign\left\{
		\widehat{\bm{d}}(\bm{w})\trans\bm{x}+\widehat{b}_{K}(\bm{w})
		\right\},
	\end{align}
	where
	\[
	\widehat b_K(\bm w)=-\widehat{\bm d}(\bm w)^\top
	\frac{\widehat{\bm\mu}_{+1,K}+\widehat{\bm\mu}_{-1,K}}{2}.
	\]
	We refer to \eqref{eq:classifier_homo} as the Homogeneous Transfer Learning Discriminant Analysis (TLDA-O) classifier.
	\begin{theorem}\label{thm:centered_overlap_prob}
		Suppose that Assumptions \ref{ass1}--\ref{ass5} hold. Then, for every deterministic unit vector $\bm\xi\in\mathbb R^p$,
		\[
		\bm\xi^\top \widehat{\bm v}_j\widehat{\bm v}_j^\top \bm\xi
		\;\xrightarrow{p}\;
		\frac{\lambda_j^2-\gamma}{\lambda_j(\lambda_j+\gamma)}
		\,\bm\xi^\top \bm v_j\bm v_j^\top \bm\xi,
		\qquad j\in\mathcal I.
		\]
	\end{theorem}
	\begin{Remark}
		Theorem~\ref{thm:centered_overlap_prob} extends the corresponding sample spike-eigenvector result in \cite{li2025SpectrallyCorrected} by allowing non-Gaussian observations with finite fourth moments and by applying to the centered pooled sample covariance matrix used in the TLDA-O classifier. 
		In addition to the usual upper-spike case that is the main focus of spiked eigenvector studies \citep{bao2022Statisticala,liu2025Asymptotic}, it also provides theoretical support for the eigenvector limits associated with lower isolated covariance spikes.
	\end{Remark}
	Theorem \ref{thm:centered_overlap_prob} shows that the sample spike directions still carry nontrivial information about the population spike directions, which allows us to replace the unknown population eigenspace by its sample version. To describe the limiting error, we next introduce two deterministic quantities. The vector \(\bm u_p=(u_{1,p},\dots,u_{K,p})^\top\) represents the limiting signal term, and the matrix \(\bm A_p=(A_{kk',p})_{1\le k,k'\le K}\) represents the limiting variance term. For $k,k'=1,\dots,K$, define $\Err(\bm{w})=\Err(\widehat{\bm{d}}(\bm{w}),\widehat{b}_{K}(\bm{w}))$ and
	\begin{align*}
		u_{k,p}
		&=
		\frac{\|\bar{\bm{\mu}}\|^2}{\sigma^2}
		-
		\sum_{j\in\mathcal I}b_ja_j
		\frac{(\bar{\bm{\mu}}\trans\bm{v}_j)^2}{\sigma^2}
		+
		\frac{\alpha_K^2}{\sigma^2}\ind_{\{k=K\}},\\
		A_{kk',p}
		&=
		\frac{\|\bar{\bm{\mu}}\|^2}{\sigma^2}
		+
		\sum_{j\in\mathcal I}c_j
		\frac{(\bar{\bm{\mu}}\trans\bm{v}_j)^2}{\sigma^2}
		+
		\left(
		\frac{\alpha_k^2}{\sigma^2}+\tau_k
		\right)\ind_{\{k=k'\}},
	\end{align*}
	where
	\begin{align*}
		a_j=\frac{\lambda_j^2-\gamma}{\lambda_j(\lambda_j+\gamma)},
		\quad b_j=\frac{\lambda_j}{1+\lambda_j},\quad
		c_j=(b_j^2-2b_j)a_j+\lambda_j-2\lambda_j b_j a_j+\lambda_j b_j^2 a_j^2.
	\end{align*}
	\begin{theorem}\label{thm:err_homo}
		Under Assumptions~\ref{ass1}--\ref{ass5}, the target-domain Gaussian-calibrated error of the TLDA-O classifier satisfies
		\begin{align}\label{eq:err_limit_uncorrected_homo}
			\Err(\bm w)
			&-
			\frac12\Phi\left(
			-\frac{\bm u_p^\top\bm w+w_K\Delta_K}{2\sqrt{\bm w^\top\bm A_p\bm w}}
			\right)
			-
			\frac12\Phi\left(
			-\frac{\bm u_p^\top\bm w-w_K\Delta_K}{2\sqrt{\bm w^\top\bm A_p\bm w}}
			\right)
			\to0
			\qquad\text{in probability}
		\end{align}
		for every fixed nonzero $\bm{w}\in\mathbb{R}^K$.
	\end{theorem}
	\begin{Remark}
		Theorem~\ref{thm:err_homo} gives a deterministic approximation for the target-domain Gaussian-calibrated error. 
		The vector \(\bm u_p\) represents the effective signal, while \(\bm A_p\) represents the variance cost of the weighted transfer direction. 
		The diagonal term \(\alpha_k^2/\sigma^2+\tau_k\) shows that stronger domain-specific variation and larger classwise dimension-to-sample-size ratios increase the variability of using domain \(k\). 
		When the target-domain class sample sizes are unbalanced, the additional term \(w_K\Delta_K\) appears in the two margins.
	\end{Remark}
	To remove this target-domain intercept bias, for a scalar adjustment \(t\in\mathbb R\), consider the adjusted score
	\[
	\widehat{\bm d}(\bm w)^\top\bm x+\widehat b_K(\bm w)+t
	\]
	and write
	\[
	\Err_t(\bm w):=\Err(\widehat{\bm d}(\bm w),\widehat b_K(\bm w)+t).
	\]
	\begin{proposition}\label{prop:optimal_intercept_homo}
		Under Assumptions~\ref{ass1}--\ref{ass5}, for every fixed nonzero \(\bm w\) satisfying \(\bm u_p^\top\bm w>0\), let
		\[
		\widehat t_p^\ast(\bm w)\in\arg\min_{t\in\mathbb R}\Err_t(\bm w).
		\]
		Then
		\[
		\widehat t_p^\ast(\bm w)-\frac12w_K\widehat\Delta_K\to0
		\qquad\text{in probability}.
		\]
		Moreover,
		\begin{align}\label{eq:err_limit_corrected_homo}
			\Err_{\widehat t_p^\ast(\bm w)}(\bm w)
			-
			\Phi\left(
			-\frac{\bm u_p^\top\bm w}{2\sqrt{\bm w^\top\bm A_p\bm w}}
			\right)
			\to0
			\qquad\text{in probability}.
		\end{align}
	\end{proposition}
	\begin{Remark}
		If \(\bm u_p^\top\bm w\le0\), the limiting Gaussian-calibrated error under this fixed weight is at least \(1/2\), and we do not consider such cases. Proposition~\ref{prop:optimal_intercept_homo} shows that the optimal intercept correction removes the target-domain imbalance term \(w_K\Delta_K\). Hence the corrected deterministic limit is no larger than the limit in Theorem~\ref{thm:err_homo}, and it is strictly smaller when \(w_K\Delta_K\neq0\). After correction, the limiting error has the same form as the balanced-target limit obtained by setting \(\Delta_K=0\); the remaining effect of class-size imbalance is only through the variance term \(\tau_k\) in \(\bm A_p\).
	\end{Remark}
	Once the corrected error limit of TLDA-O is available, we can optimize it over the transfer weights. This leads to an optimal choice that serves as a benchmark for the best possible linear combination of the domainwise discriminant directions.
	\begin{corollary}\label{cor:oracle_weight_homo}
		Under Assumptions~\ref{ass1}--\ref{ass5}, suppose \(\bm u_p\neq \bm0\). Then the deterministic equivalent in \eqref{eq:err_limit_corrected_homo} is minimized by
		\begin{align}\label{eq:oracle_weight_homo}
			\bm{w}_p^\ast \propto \bm{A}_p^{-1}\bm{u}_p.
		\end{align}
		Under the normalization $\bm{w}\trans\bm{A}_p\bm{w}=1$, the unique maximizer is
		\begin{align}
			\bm{w}_p^\ast
			=
			\frac{\bm{A}_p^{-1}\bm{u}_p}
			{\sqrt{\bm{u}_p\trans\bm{A}_p^{-1}\bm{u}_p}}.
		\end{align}
	\end{corollary}
	\begin{Remark}
		Corollary~\ref{cor:oracle_weight_homo} shows that the optimal weight has the form of a generalized signal-to-noise maximizer. 
		Domains with larger effective shared signal tend to receive larger weights, while domains with stronger domain-specific variation or larger dimension-to-sample-size ratios are penalized through \(\bm A_p\). 
	\end{Remark}
	The optimal weight for TLDA-O depends on unknown population quantities and therefore cannot be used directly. We next replace these quantities by consistent estimators and obtain a consistent plug-in weight vector. To construct this plug-in version of \eqref{eq:oracle_weight_homo}, define
	\begin{align}
		\widehat{\kappa}_p
		:=
		\frac{1}{\sigma^2K(K-1)}
		\sum_{m\neq \ell}\widehat{\bm{\mu}}_m\trans\widehat{\bm{\mu}}_{\ell},
	\end{align}
	and, for each \(k=1,\dots,K\), define
	\begin{align}
		\widehat{\beta}_k
		:=
		\frac{1}{\sigma^2}
		\left(
		\|\widehat{\bm{\mu}}_k\|^2
		-
		\frac{1}{K-1}\sum_{\ell\neq k}
		\widehat{\bm{\mu}}_k^\top\widehat{\bm{\mu}}_{\ell}
		\right)-
		\widehat\tau_k.
	\end{align}
	As in Assumption~\ref{ass3}, the spike strengths \(\lambda_j\), \(j\in\mathcal I\), and hence
	\(b_j=\lambda_j/(1+\lambda_j)\), are treated as known population parameters.
	For each \(j\in\mathcal I\), define
	\[
	\widehat a_j
	:=
	\frac{\lambda_j^2-\widehat\gamma}
	{\lambda_j(\lambda_j+\widehat\gamma)},
	\qquad
	\widehat c_j
	:=
	(b_j^2-2b_j)\widehat a_j
	+\lambda_j
	-2\lambda_j b_j\widehat a_j
	+\lambda_j b_j^2\widehat a_j^2 .
	\]
	Then, for $k,k'=1,\dots,K$, define
	\begin{gather*}
		\widehat{u}_{k,p}
		:=
		\widehat{\kappa}_p
		-
		\sum_{j\in\mathcal I}
		\frac{b_j}{\sigma^2K(K-1)}
		\sum_{m\neq \ell}
		\widehat{\bm{\mu}}_{\ell}\trans
		\widehat{\bm{v}}_{j}\widehat{\bm{v}}_{j}\trans
		\widehat{\bm{\mu}}_m
		+
		\widehat{\beta}_K\,\ind_{\{k=K\}},\\
		\widehat{A}_{kk',p}
		:=
		\widehat{\kappa}_p
		+
		\sum_{j\in\mathcal I}
		\frac{\widehat c_j}
		{\sigma^2K(K-1)}
		\sum_{m\neq \ell}
		\frac{1}{\widehat a_j}\,
		\widehat{\bm{\mu}}_{\ell}\trans
		\widehat{\bm{v}}_{j}\widehat{\bm{v}}_{j}\trans
		\widehat{\bm{\mu}}_m
		+
		\ind_{\{k=k'\}}
		\bigl(\widehat{\beta}_k+\widehat\tau_k\bigr).
	\end{gather*}
	Let
	\[
	\widehat{\bm{u}}_p=(\widehat{u}_{1,p},\dots,\widehat{u}_{K,p})\trans,
	\qquad
	\widehat{\bm{A}}_p=(\widehat{A}_{kk',p})_{1\le k,k'\le K}.
	\]
	
	\begin{theorem}\label{thm:emp_weight_homo}
		Suppose that Assumptions \ref{ass1}--\ref{ass5} hold, and that there exists a constant \(c_u>0\) such that
		$\|\bm u_p\|\ge c_u$,
		for all sufficiently large \(p\). Then
		\begin{align}
			\widehat{\bm{u}}_p-\bm{u}_p\to\bm 0,
			\qquad
			\widehat{\bm{A}}_p-\bm{A}_p\to\bm 0\qquad\text{in probability}.
		\end{align}
		Moreover, with probability tending to one, \(\widehat{\bm{A}}_p\) is positive definite, and the empirical optimal weight vector
		\begin{align}\label{eq:emp_weight_homo}
			\widehat{\bm{w}}_p^\ast
			=
			\frac{\widehat{\bm{A}}_p^{-1}\widehat{\bm{u}}_p}
			{\sqrt{\widehat{\bm{u}}_p\trans\widehat{\bm{A}}_p^{-1}\widehat{\bm{u}}_p}}
		\end{align}
		satisfies
		\begin{align}
			\widehat{\bm{w}}_p^\ast-\bm{w}_p^\ast\to\bm 0
			\qquad\text{in probability}.
		\end{align}
	\end{theorem}
	\begin{Remark}
		The consistency is stated for known spectral parameters in Assumption~\ref{ass3}. The same conclusion continues to hold if these parameters are replaced by consistent estimators such as those in \cite{bai2012Estimation,jiang2021Generalized}.
	\end{Remark}
	\subsection{Heterogeneous covariance case}\label{sec2.2}
	We now move to the more general case where the covariance matrices are allowed to vary across domains. This setting is closer to practice, but the interaction across domains becomes more involved because both the spike strengths and the spike directions may differ from one domain to another.
	
	\begin{assumption}\label{ass3_hetero}
		For each domain $k=1,\dots,K$, the covariance matrix follows a spiked covariance model:
		\begin{align}\label{eq:heteroSigma_ass}
			\bm{\Sigma}_k
			=
			\sigma_k^2
			\left(
			\bm{I}_p+\sum_{j\in\mathcal I_k}\lambda_{j,k}\bm{v}_{j,k}\bm{v}_{j,k}\trans
			\right),\qquad \mathcal I_k=\mathcal I_{k,+}\cup\mathcal I_{k,-}
		\end{align}
		where $1/c\le \sigma_k^2\le c$ for some constant $c>0$, \(\{\bm v_{j,k}:j\in\mathcal I_k\}\) is an orthonormal family in
		\(\mathbb R^p\), and
		$
		\mathcal I_{k,+}=\{1,\dots,r_{k,+}\}, \mathcal I_{k,-}=\{-r_{k,-},\dots,-1\},
		$
		with \(r_k:=r_{k,+}+r_{k,-}\) fixed for each \(k\).
	\end{assumption}
	\begin{Remark}
		As in the homogeneous covariance case, for simplicity, in the theoretical analysis, we assume that
		\(\sigma_k^2\), \(r_{k,+}\), \(r_{k,-}\), and
		\(\lambda_{j,k}\), \(j\in\mathcal I_k\), are known for each \(k\).
		In experiments, these population spectral parameters can be estimated consistently using existing methods for spiked covariance models; see, for example, \cite{bai2012Estimation,jiang2021Generalized}.
	\end{Remark}
	\begin{assumption}\label{ass5_hetero}
		For each domain \(k=1,\dots,K\) and each \(j\in\mathcal I_k\), \(\lambda_{j,k}\) is independent of
		\(p\), and
		\[
		\lambda_{1,k}>\cdots>\lambda_{r_{k,+},k}>\sqrt{\gamma_k}>
		-\sqrt{\gamma_k}>
		\lambda_{-r_{k,-},k}>\cdots>\lambda_{-1,k}>-1.
		\]
		Moreover, if \(\gamma_k\geq1\), then \(\mathcal I_{k,-}=\varnothing\).
	\end{assumption}
	In the heterogeneous setting, the covariance information can no longer be pooled directly. We therefore estimate the spiked covariance structure within each domain and then combine the resulting local discriminant directions through transfer weights.
	The within-domain sample covariance matrix is estimated separately by
	\begin{align}\label{eq:sample_cov_hetero}
		\bm S_{n,k}
		=
		\frac{1}{n_k-2}
		\sum_{i=1}^{n_k}
		\Bigl[
		(\bm{X}_k)_i-\widehat{\bm{\mu}}_{(y_k)_i,k}
		\Bigr]
		\Bigl[
		(\bm{X}_k)_i-\widehat{\bm{\mu}}_{(y_k)_i,k}
		\Bigr]\trans.
	\end{align}
	Write its spectral decomposition as
	\[
	\bm S_{n,k}
	=
	\sum_{j=1}^p\ell_{j,k}\widehat{\bm v}_{j,k}\widehat{\bm v}_{j,k}^\top,
	\]
	with the eigenvalues arranged in decreasing order. For \(j\in\mathcal I_{k,+}\), \(\widehat{\bm v}_{j,k}\) denotes the sample eigenvector associated with the upper outlier corresponding to \(\lambda_{j,k}\). For \(j\in\mathcal I_{k,-}\), \(\widehat{\bm v}_{j,k}\) denotes the sample eigenvector associated with the lower outlier corresponding to \(\lambda_{j,k}\). We then define the spectrally-corrected covariance estimator by
	\begin{align}\label{eq:Sigma_tilde_hetero}
		\widehat{\bm{\Sigma}}_k
		=
		\sigma_k^2
		\left(
		\bm{I}_p+\sum_{j\in\mathcal{I}_k}\lambda_{j,k}\widehat{\bm{v}}_{j,k}\widehat{\bm{v}}_{j,k}\trans
		\right).
	\end{align}
	By the same argument as Theorem~\ref{thm:centered_overlap_prob}, the domainwise analogue holds with $(\gamma,\lambda_j,\bm v_j)$ replaced by $(\gamma_k,\lambda_{j,k},\bm v_{j,k})$. The local discriminant direction is then defined by
	$
	\widehat{\bm{d}}^{E}_k=\widehat{\bm{\Sigma}}_k^{-1}\widehat{\bm{\mu}}_k,
	$
	and the transfer discriminant direction remains
	\[
	\widehat{\bm{d}}^{E}(\bm{w})=\sum_{k=1}^K w_k\widehat{\bm{d}}^{E}_k.
	\]
	The resulting target-domain classifier is
	\begin{align}\label{eq:classifier_hetero}
		\widehat{y}^{E}(\bm{x})
		=
		\sign\left\{
		\widehat{\bm{d}}^{E}(\bm{w})\trans\bm{x}+\widehat{b}_{K}^{E}(\bm{w})
		\right\},
	\end{align}
	where
	\[
	\widehat b_K^{E}(\bm w)
	=
	-\widehat{\bm d}^{E}(\bm w)^\top
	\frac{\widehat{\bm\mu}_{+1,K}+\widehat{\bm\mu}_{-1,K}}{2}.
	\]
	We refer to \eqref{eq:classifier_hetero} as the Heterogeneous Transfer Learning Discriminant Analysis (TLDA-E) classifier.
	
	To state the error formula in this setting, we need a few quantities that describe both within-domain spike effects and cross-domain spike alignment. Define
	\[
	a_{j,k}:=\frac{\lambda_{j,k}^2-\gamma_k}{\lambda_{j,k}(\lambda_{j,k}+\gamma_k)},
	\qquad
	b_{j,k}:=\frac{\lambda_{j,k}}{1+\lambda_{j,k}},
	\qquad
	m_{j,k}:=\bar{\bm\mu}^\top \bm v_{j,k},
	\qquad j\in\mathcal I_k.
	\]
	For \(k,k'=1,\dots,K\), define
	\[
	\bm m_k:=(m_{j,k})_{j\in\mathcal I_k},
	\quad
	\bm R_{kk'}:=
	\bigl(\rho_{j\ell}^{(k,k')}\bigr)_{j\in\mathcal I_k,\ \ell\in\mathcal I_{k'}},\quad
	\rho_{j\ell}^{(k,k')}:=\bm v_{j,k}^\top \bm v_{\ell,k'}.
	\]
	Also define
	\[
	\bm B_k:=\diag(b_{j,k})_{j\in\mathcal I_k},
	\qquad
	\bm D_k:=\diag(b_{j,k}a_{j,k})_{j\in\mathcal I_k},
	\qquad
	\bm\Lambda_k:=\diag(\lambda_{j,k})_{j\in\mathcal I_k},
	\]
	and
	\[
	\bm C_{kk'}:=
	\begin{cases}
		\diag(a_{j,k})_{j\in\mathcal I_k}, & k=k',\\[1mm]
		\diag(a_{j,k})_{j\in\mathcal I_k}\,\bm R_{kk'}\,
		\diag(a_{\ell,k'})_{\ell\in\mathcal I_{k'}}, & k\neq k'.
	\end{cases}
	\]
	For each $k,k'=1,\dots,K$, define
	$
	\bm q_k:=\bm m_K-\bm R_{Kk}\bm D_k\bm m_k
	$, and
	\begin{gather*}
		u_{k,p}^{E}
		=
		\frac{\|\bar{\bm\mu}\|^2-\bm m_k^\top \bm D_k\bm m_k}{\sigma_k^2}
		+
		\frac{\alpha_K^2}{\sigma_K^2}\ind_{\{k=K\}},\\
		A_{kk',p}^{E}
		=
		\frac{\sigma_K^2}{\sigma_k^2\sigma_{k'}^2}\,\Omega_{kk',p}^{E}
		+
		\left(
		\frac{\sigma_K^2\alpha_k^2}{\sigma_k^4}
		+
		\tau_k\frac{\sigma_K^2}{\sigma_k^2}
		\right)\ind_{\{k=k'\}},
	\end{gather*}
	where
	\[
	\Omega_{kk',p}^{E}
	=
	\|\bar{\bm\mu}\|^2
	-\bm m_k^\top \bm D_k\bm m_k
	-\bm m_{k'}^\top \bm D_{k'}\bm m_{k'}
	+\bm m_k^\top \bm B_k\bm C_{kk'}\bm B_{k'}\bm m_{k'}
	+\bm q_k^\top \bm\Lambda_K \bm q_{k'}.
	\]
	
	Let
	\[
	\bm u_p^{E}:=(u_{1,p}^{E},\dots,u_{K,p}^{E})^\top,
	\qquad
	\bm A_p^{E}:=(A_{kk',p}^{E})_{1\le k,k'\le K}.
	\]
	Also write
	\[
	\Err^{E}(\bm w):=\Err(\widehat{\bm d}^{E}(\bm w),\widehat b_K^{E}(\bm w)).
	\]
	
	\begin{theorem}\label{thm:err_hetero}
		Under Assumptions~\ref{ass1}, \ref{ass2}, \ref{ass4}, \ref{ass3_hetero} and \ref{ass5_hetero},
		the target-domain Gaussian-calibrated error of the TLDA-E classifier satisfies
		\begin{align}\label{eq:err_limit_uncorrected_hetero}
			\Err^{E}(\bm w)
			&-
			\frac12\Phi\left(
			-\frac{(\bm u_p^{E})^\top\bm w+w_K\Delta_K}{2\sqrt{\bm w^\top\bm A_p^{E}\bm w}}
			\right)
			-
			\frac12\Phi\left(
			-\frac{(\bm u_p^{E})^\top\bm w-w_K\Delta_K}{2\sqrt{\bm w^\top\bm A_p^{E}\bm w}}
			\right)
			\to 0
			\qquad\text{in probability}
		\end{align}
		for every fixed nonzero \(\bm w\in\mathbb R^K\).
	\end{theorem}
	\begin{Remark}
		Theorem~\ref{thm:err_hetero} extends the deterministic error approximation to domain-specific covariance matrices. 
		Here \(\bm u_p^{E}\) still describes the effective signal, and \(\bm A_p^{E}\) describes the variance cost. 
		Compared with the homogeneous case, the error now depends not only on the shared signal, domain-specific variation, dimension-to-sample-size ratios, and within-domain spike directions, but also on cross-domain spike alignment. 
		The matrices \(\bm R_{kk'}\) and the vectors \(\bm q_k\) describe how the spike directions from different domains interact with the target-domain covariance structure.
	\end{Remark}
	To remove this target-domain intercept bias, for a scalar adjustment \(t\in\mathbb R\), consider the adjusted score
	\[
	(\widehat{\bm d}^{E}(\bm w))^\top\bm x+\widehat b_K^{E}(\bm w)+t
	\]
	and write
	\[
	\Err_t^{E}(\bm w):=\Err(\widehat{\bm d}^{E}(\bm w),\widehat b_K^{E}(\bm w)+t).
	\]
	\begin{proposition}\label{prop:optimal_intercept_hetero}
		Under Assumptions~\ref{ass1}, \ref{ass2}, \ref{ass4}, \ref{ass3_hetero} and \ref{ass5_hetero}, for every fixed nonzero \(\bm w\) satisfying \((\bm u_p^{E})^\top\bm w>0\), let
		\[
		\widehat t_{p,E}^\ast(\bm w)\in\arg\min_{t\in\mathbb R}\Err_t^{E}(\bm w).
		\]
		Then
		\[
		\widehat t_{p,E}^\ast(\bm w)-\frac12w_K\widehat\Delta_K\to0
		\qquad\text{in probability}.
		\]
		Moreover,
		\begin{align}\label{eq:err_limit_corrected_hetero}
			\Err_{\widehat t_{p,E}^\ast(\bm w)}^{E}(\bm w)
			-
			\Phi\left(
			-\frac{(\bm u_p^{E})^\top\bm w}{2\sqrt{\bm w^\top\bm A_p^{E}\bm w}}
			\right)
			\to0
			\qquad\text{in probability}.
		\end{align}
	\end{proposition}
	
	The corrected deterministic error formula also yields the optimal TLDA-E transfer weights.
	\begin{corollary}\label{cor:oracle_weight_hetero}
		Under Assumptions~\ref{ass1}, \ref{ass2}, \ref{ass4}, \ref{ass3_hetero} and \ref{ass5_hetero}, suppose \(\bm u_p^{E}\neq \bm 0\). Then the deterministic equivalent in \eqref{eq:err_limit_corrected_hetero} is minimized by
		\begin{align}\label{eq:oracle_weight_hetero}
			\bm w_{p,E}^\ast \propto (\bm A_p^{E})^{-1}\bm u_p^{E}.
		\end{align}
		Under the normalization
		$
		\bm w^\top \bm A_p^{E}\bm w=1,
		$
		the unique maximizer is
		\begin{align}
			\bm w_{p,E}^\ast
			=
			\frac{(\bm A_p^{E})^{-1}\bm u_p^{E}}
			{\sqrt{(\bm u_p^{E})^\top(\bm A_p^{E})^{-1}\bm u_p^{E}}}.
		\end{align}
	\end{corollary}
	Compared with TLDA-O, the TLDA-E weight additionally reflects the alignment between source-domain spike directions and the target-domain covariance structure. To construct a plug-in version of \eqref{eq:oracle_weight_hetero}, define, for any
	\(p\times p\) matrix \(\bm M\),
	\[
	\widehat{\mathcal A}_p(\bm M)
	:=
	\frac{1}{K(K-1)}
	\sum_{a\neq b}\widehat{\bm\mu}_a^\top \bm M \widehat{\bm\mu}_b.
	\]
	Also, for each \(k=1,\dots,K\), define
	\[
	\widehat{\beta}_k^{E}
	:=
	\|\widehat{\bm\mu}_k\|^2
	-
	\frac{1}{K-1}\sum_{\ell\neq k}\widehat{\bm\mu}_k^\top\widehat{\bm\mu}_\ell
	-
	\widehat\tau_k\sigma_k^2.
	\]
	As in Assumption~\ref{ass3_hetero}, the spike strengths $\lambda_{j,k}$, $j\in\mathcal I_k$, and hence $b_{j,k}=\lambda_{j,k}/(1+\lambda_{j,k})$, are treated as known. For each \(j\in\mathcal I_k\), define
	\[
	\widehat a_{j,k}
	:=
	\frac{\lambda_{j,k}^2-\widehat\gamma_k}
	{\lambda_{j,k}(\lambda_{j,k}+\widehat\gamma_k)},\qquad
	\widehat{\bm P}_{j,k}=\widehat{\bm v}_{j,k}\widehat{\bm v}_{j,k}\trans,
	\]
	\[
	\widehat{\bm M}_k
	:=
	\bm I_p-\sum_{j\in\mathcal I_k} b_{j,k}\widehat{\bm P}_{j,k},\qquad
	\widehat u_{k,p}^{E}
	:=
	\frac{1}{\sigma_k^2}\widehat{\mathcal A}_p(\widehat{\bm M}_k)
	+
	\frac{\widehat{\beta}_K^{E}}{\sigma_K^2}\ind_{\{k=K\}}.
	\]
	For $k,k'=1,\dots,K,\ell\in\mathcal I_K$, define
	\[
	\widehat{\Theta}_{kk',\ell}
	:=
	\begin{cases}
		\dfrac{1}{\widehat a_{\ell,K}}
		\widehat{\mathcal A}_p(\widehat{\bm M}_k\widehat{\bm P}_{\ell,K}\widehat{\bm M}_{k'}),
		& k\neq K,\ k'\neq K,\\[3mm]
		\dfrac{1-b_{\ell,K}\widehat a_{\ell,K}}{\widehat a_{\ell,K}}
		\widehat{\mathcal A}_p(\widehat{\bm P}_{\ell,K}\widehat{\bm M}_{k'}),
		& k=K,\ k'\neq K,\\[3mm]
		\dfrac{1-b_{\ell,K}\widehat a_{\ell,K}}{\widehat a_{\ell,K}}
		\widehat{\mathcal A}_p(\widehat{\bm M}_k\widehat{\bm P}_{\ell,K}),
		& k\neq K,\ k'=K,\\[3mm]
		\dfrac{(1-b_{\ell,K}\widehat a_{\ell,K})^2}{\widehat a_{\ell,K}}
		\widehat{\mathcal A}_p(\widehat{\bm P}_{\ell,K}),
		& k=K,\ k'=K.
	\end{cases}
	\]
	Then define
	\[
	\widehat{\Omega}_{kk',p}^{E}
	:=
	\widehat{\mathcal A}_p(\widehat{\bm M}_k\widehat{\bm M}_{k'})
	+
	\sum_{\ell\in\mathcal I_K}\lambda_{\ell,K}\widehat{\Theta}_{kk',\ell},
	\]
	and
	\[
	\widehat A_{kk',p}^{E}
	:=
	\frac{\sigma_K^2}{\sigma_k^2\sigma_{k'}^2}\,\widehat{\Omega}_{kk',p}^{E}
	+
	\left(
	\frac{\sigma_K^2\widehat{\beta}_k^{E}}{\sigma_k^4}
	+
	\widehat\tau_k\frac{\sigma_K^2}{\sigma_k^2}
	\right)\ind_{\{k=k'\}}.
	\]
	
	Let
	\[
	\widehat{\bm u}_p^{E}
	:=
	(\widehat u_{1,p}^{E},\dots,\widehat u_{K,p}^{E})^\top,
	\qquad
	\widehat{\bm A}_p^{E}
	:=
	(\widehat A_{kk',p}^{E})_{1\le k,k'\le K}.
	\]
	\begin{theorem}\label{thm:emp_weight_hetero}
		Under Assumptions~\ref{ass1}, \ref{ass2}, \ref{ass4}, \ref{ass3_hetero} and \ref{ass5_hetero}, suppose there exists a constant \(c_u>0\) such that
		$\|\bm u_p^{E}\|\ge c_u
		$,
		for all sufficiently large \(p\). Then
		\[
		\widehat{\bm u}_p^{E}-\bm u_p^{E}\to \bm 0,
		\qquad
		\widehat{\bm A}_p^{E}-\bm A_p^{E}\to \bm 0
		\qquad\text{in probability}.
		\]
		Moreover, with probability tending to one, \(\widehat{\bm A}_p^{E}\) is positive definite, and the empirical optimal weight vector
		\begin{align}\label{eq:emp_weight_hetero}
			\widehat{\bm w}_{p,E}^\ast
			=
			\frac{(\widehat{\bm A}_p^{E})^{-1}\widehat{\bm u}_p^{E}}
			{\sqrt{(\widehat{\bm u}_p^{E})^\top(\widehat{\bm A}_p^{E})^{-1}\widehat{\bm u}_p^{E}}}
		\end{align}
		satisfies
		\[
		\widehat{\bm w}_{p,E}^\ast-\bm w_{p,E}^\ast\to \bm 0
		\qquad\text{in probability}.
		\]
	\end{theorem}
	\begin{Remark}
		This theorem is the heterogeneous analogue of Theorem~\ref{thm:emp_weight_homo}. Its consistency statement is likewise understood for known spectral parameters, or for consistent plug-in estimates of these parameters.
	\end{Remark}
	
	\section{Simulations}\label{sec:simulation}
	
	We conduct four sets of simulations to examine: the accuracy of the deterministic approximation, the performance of the proposed classifiers, robustness to spike-number misspecification, and the effect of intercept bias correction under class imbalance.
	
	Throughout the simulations, we set $K=3$, with domain $K$ being the target domain. Each reported empirical error rate is computed on an independently generated balanced target-domain test set with 1000 observations, consisting of 500 observations from each class. Each simulation result is averaged over 1000 replications.
	
	For each replication, the common signal is generated as
	$
	\bar{\bm\mu}\sim N\left(\bm 0,\bm I_p/p\right)
	$.
	The domain-specific deviations are generated according to Assumption~\ref{ass2} with
	$
	\alpha_1=0.2,
	\alpha_2=0.5,
	\alpha_3=1.
	$
	Then $\bm\mu_k=\bar{\bm\mu}+\bm\delta_k$, and the class mean vectors are set as
	$
	\bm\mu_{+1,k}=\bm\mu_k/2
	$
	and
	$
	\bm\mu_{-1,k}=-\bm\mu_k/2.
	$
	For the covariance matrices, the noise level, spike numbers, and spike strengths are set to be the same across domains. That is, for each domain $k$, we take
	$
	\sigma_k^2=\sigma^2=1,
	r_{k,+}=r_+=1,
	r_{k,-}=r_-=2,
	$
	and set
	$
	\lambda_{1,k}=\lambda_1=3,
	\lambda_{-2,k}=\lambda_{-2}=-0.8,
	\lambda_{-1,k}=\lambda_{-1}=-0.9 .
	$
	These values are used to generate the population covariance matrices. In implementing the classifiers, unless otherwise stated, the spike numbers, spike strengths, and noise variances are estimated from the data using the methods of \cite{bai2012Estimation,jiang2021Generalized}, rather than being set to their population values.
	\paragraph{Case 1: homogeneous covariance.}
	Generate an orthonormal family
	$
	\bm v_1,\bm v_{-2},\bm v_{-1}\in\mathbb R^p
	$
	by orthonormalizing three independent standard Gaussian vectors. All domains share
	$
	\bm\Sigma_1=\cdots=\bm\Sigma_K=\bm\Sigma
	$,
	where
	\[
	\bm\Sigma
	=
	\bm I_p
	+\lambda_1\bm v_1\bm v_1\trans
	+\lambda_{-2}\bm v_{-2}\bm v_{-2}\trans
	+\lambda_{-1}\bm v_{-1}\bm v_{-1}\trans .
	\]
	In this case, the proposed method uses the pooled covariance estimator and is denoted by TLDA-O.
	
	\paragraph{Case 2: heterogeneous covariance.}
	For each domain $k=1,\ldots,K$, generate an orthonormal family
	$
	\bm v_{1,k},\bm v_{-2,k},\bm v_{-1,k}\in\mathbb R^p
	$
	by orthonormalizing three independent standard Gaussian vectors. The domain-specific covariance matrix is
	\[
	\bm\Sigma_k
	=
	\bm I_p
	+\lambda_1\bm v_{1,k}\bm v_{1,k}\trans
	+\lambda_{-2}\bm v_{-2,k}\bm v_{-2,k}\trans
	+\lambda_{-1}\bm v_{-1,k}\bm v_{-1,k}\trans,
	\qquad
	k=1,\ldots,K.
	\]
	In this case, the proposed method uses domain-specific covariance estimators and is denoted by TLDA-E.
	
	\subsection{Deterministic approximation}
	We first examine the accuracy of the deterministic approximation to the Gaussian-calibrated error. Both Case 1 and Case 2 are considered. We take
	$
	p=20,40,\ldots,400
	$
	and set
	$
	p/n_k=0.5
	$
	for each domain, with balanced class sizes
	$
	n_{k,+}=n_{k,-}=n_k/2.
	$
	
	In this experiment, we use the fixed transfer weight
	$
	\bm\omega=(1,1,1)\trans .
	$
	For Case 1, the classifier is constructed by TLDA-O with this fixed weight. For Case 2, the classifier is constructed by TLDA-E with the same fixed weight. We compare the empirical target-domain test error with the deterministic approximation derived in the theory.
	
	\begin{figure}[!htbp]
		\centering
		\begin{minipage}{0.48\textwidth}
			\centering
			\includegraphics[page=1,width=\textwidth]{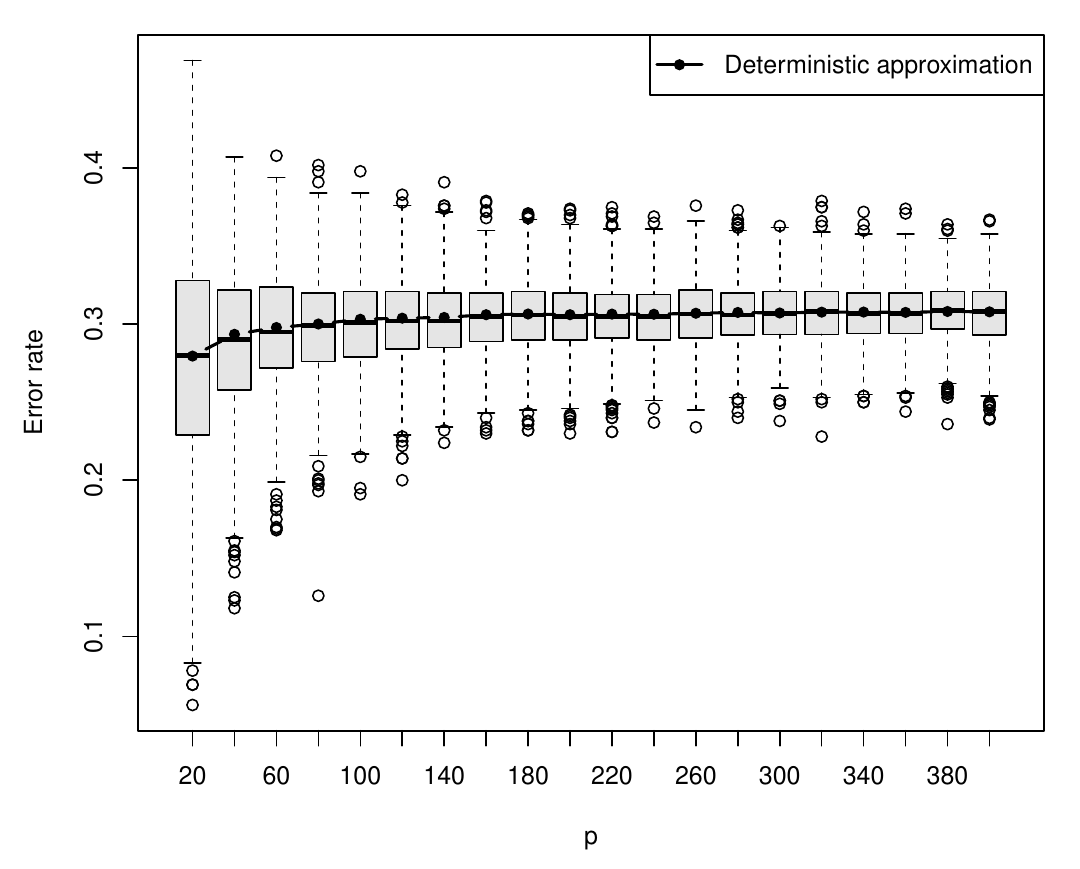}
			
			(a) Case 1: homogeneous covariance.
		\end{minipage}
		\hfill
		\begin{minipage}{0.48\textwidth}
			\centering
			\includegraphics[page=2,width=\textwidth]{boxplot_test_error_rate_by_p.pdf}
			
			(b) Case 2: heterogeneous covariance.
		\end{minipage}
		\caption{Empirical target-domain test error and its deterministic approximation.}
		\label{fig:deterministic_approximation}
	\end{figure}
	
	Figure~\ref{fig:deterministic_approximation} shows that the deterministic approximation closely tracks the empirical target-domain test error in both covariance settings. As $p$ increases, the empirical error becomes more concentrated around the theoretical curve, supporting the accuracy of the proposed deterministic equivalent.
	
	\subsection{Performance of classifiers}
	We next compare the proposed TLDA methods with baseline methods. In this experiment, we fix
	$
	p=100,
	$
	set the two source-domain sample sizes as
	$
	n_1=n_2=500,
	$
	and vary the target-domain sample size as
	$
	n_3=200,250,\ldots,500.
	$
	All domains have balanced class sizes.
	
	In Case 1, we compare LDA, pooled LDA, TLP-RDA, TLDA-O with equal weights, and TLDA-O with estimated optimal weights. In Case 2, we compare LDA, pooled LDA, TL-RDA, TLDA-E with equal weights, and TLDA-E with estimated optimal weights.
	For the optimal-weight TLDA methods, we use the plug-in estimators based on Theorem~\ref{thm:emp_weight_homo} and Theorem~\ref{thm:emp_weight_hetero}.
	
	\begin{figure}[!htbp]
		\centering
		\begin{minipage}{0.48\textwidth}
			\centering
			\includegraphics[page=1,width=\textwidth]{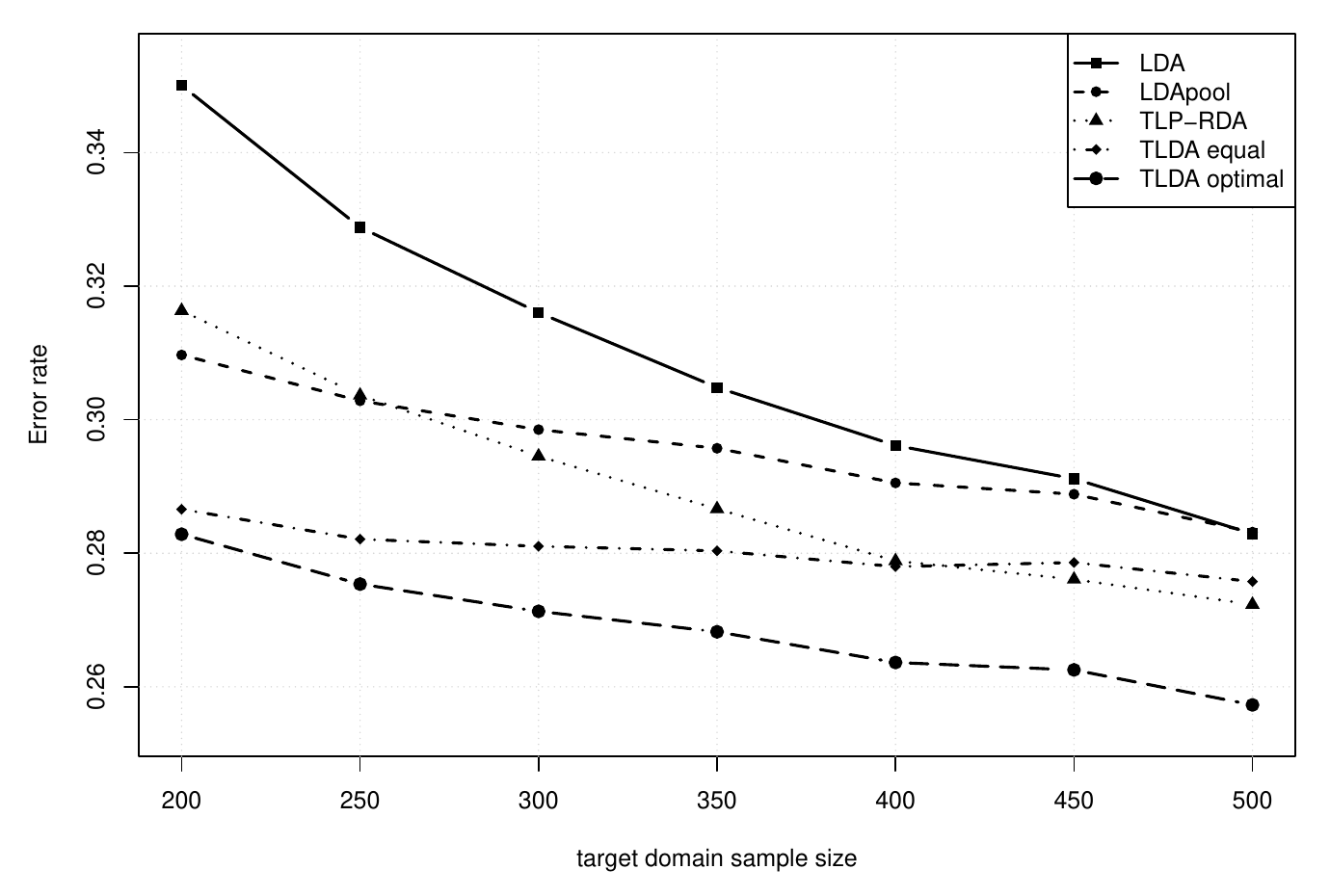}
			
			(a) Case 1: homogeneous covariance.
		\end{minipage}
		\hfill
		\begin{minipage}{0.48\textwidth}
			\centering
			\includegraphics[page=2,width=\textwidth]{lineplot_all_methods_mean_test_error_rate.pdf}
			
			(b) Case 2: heterogeneous covariance.
		\end{minipage}
		\caption{Mean target-domain test error of different classifiers as the target-domain sample size varies.}
		\label{fig:classifier_performance}
	\end{figure}
	
	Figure~\ref{fig:classifier_performance} reports the mean test error of all competing methods. In both covariance settings, the proposed TLDA method with estimated optimal weights gives the lowest error across the target sample sizes. Compared with the equal-weight version, the optimal-weight version achieves a clear improvement, showing the benefit of selecting transfer weights according to the proposed asymptotic criterion.
	
	\subsection{Robustness to spike-number misspecification}
	
	We further examine the stability of the proposed methods when the number of spikes is misspecified. The simulation setting follows that in the previous subsection. The data are generated from the true spike configuration
	$
	(r_+,r_-)=(1,2)
	$,
	while in constructing the TLDA classifiers, we deliberately use several prescribed spike numbers
	\[
	(r_+^\ast,r_-^\ast)
	\in
	\{(1,2),(0,2),(1,1),(2,2),(1,3)\}.
	\]
	Here $(1,2)$ corresponds to the correct specification, $(0,2)$ and $(1,1)$ represent under-estimation, and $(2,2)$ and $(1,3)$ represent over-estimation. For each prescribed spike number, the corresponding spike eigenvalues and eigenvectors are still estimated from the selected sample eigenpairs. In Case 1, we compare TLDA-O with TLP-RDA; in Case 2, we compare TLDA-E with TL-RDA.
	
	\begin{figure}[!htbp]
		\centering
		\begin{minipage}{0.48\textwidth}
			\centering
			\includegraphics[page=1,width=\textwidth]{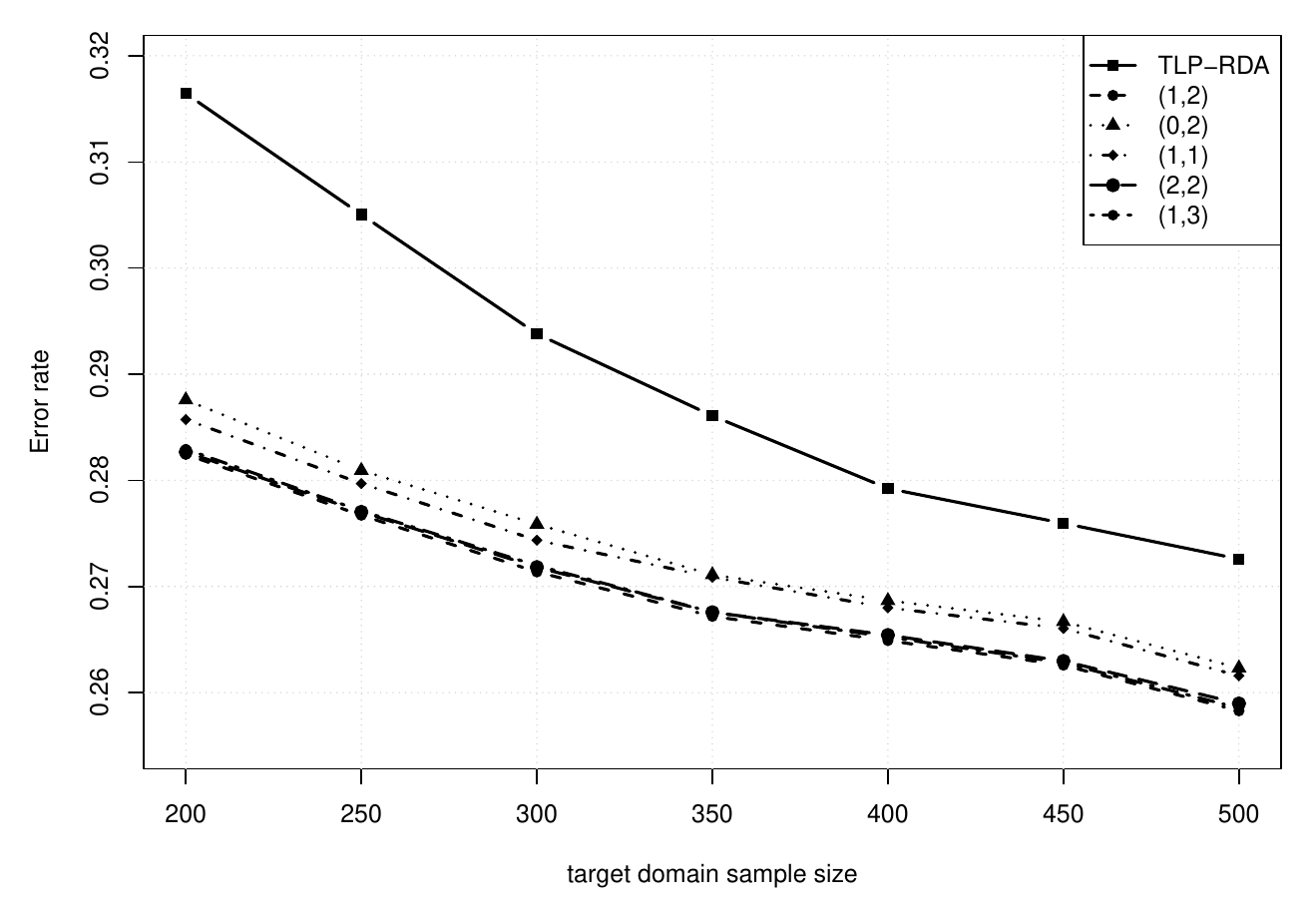}
			
			(a) Case 1: homogeneous covariance.
		\end{minipage}
		\hfill
		\begin{minipage}{0.48\textwidth}
			\centering
			\includegraphics[page=2,width=\textwidth]{lineplot_spike_misspec_stability.pdf}
			
			(b) Case 2: heterogeneous covariance.
		\end{minipage}
		\caption{Sensitivity of TLDA-O and TLDA-E to spike-number misspecification. The true spike numbers are $(r_+,r_-)=(1,2)$, while the classifiers use prescribed values $(r_+^\ast,r_-^\ast)$.}
		\label{fig:spike_misspec_stability}
	\end{figure}
	
	Figure~\ref{fig:spike_misspec_stability} shows that the proposed methods are stable under moderate spike-number misspecification. In both covariance settings, the curves corresponding to different prescribed spike numbers remain close to the correctly specified case $(1,2)$. Moreover, the proposed TLDA classifiers continue to outperform the corresponding transfer RDA baselines. This suggests that the classification performance is not overly sensitive to small errors in spike-number selection.
	
	\subsection{Bias correction}
	Finally, we study the intercept bias correction under class imbalance. Both Case 1 and Case 2 are considered. We fix
	$
	p=100
	$
	and set
	$
	n_k=200,
	k=1,\ldots,K.
	$
	For each domain, the positive-class sample size varies as
	$
	n_{k,+}=30,40,\ldots,170,
	$
	and the negative-class sample size is
	$
	n_{k,-}=200-n_{k,+}.
	$
	
	In this experiment, the discriminant direction is constructed using the estimated optimal TLDA weights. We compare three choices of intercept while keeping the discriminant direction fixed.
	The first one is the uncorrected intercept. The second one is the theoretically corrected intercept, obtained by Proposition~\ref{prop:optimal_intercept_homo} and Proposition~\ref{prop:optimal_intercept_hetero}. The third one is an oracle intercept, obtained by optimizing an intercept correction on the target-domain test set, and is used only as a benchmark.
	
	\begin{figure}[!htbp]
		\centering
		\begin{minipage}{0.48\textwidth}
			\centering
			\includegraphics[page=1,width=\textwidth]{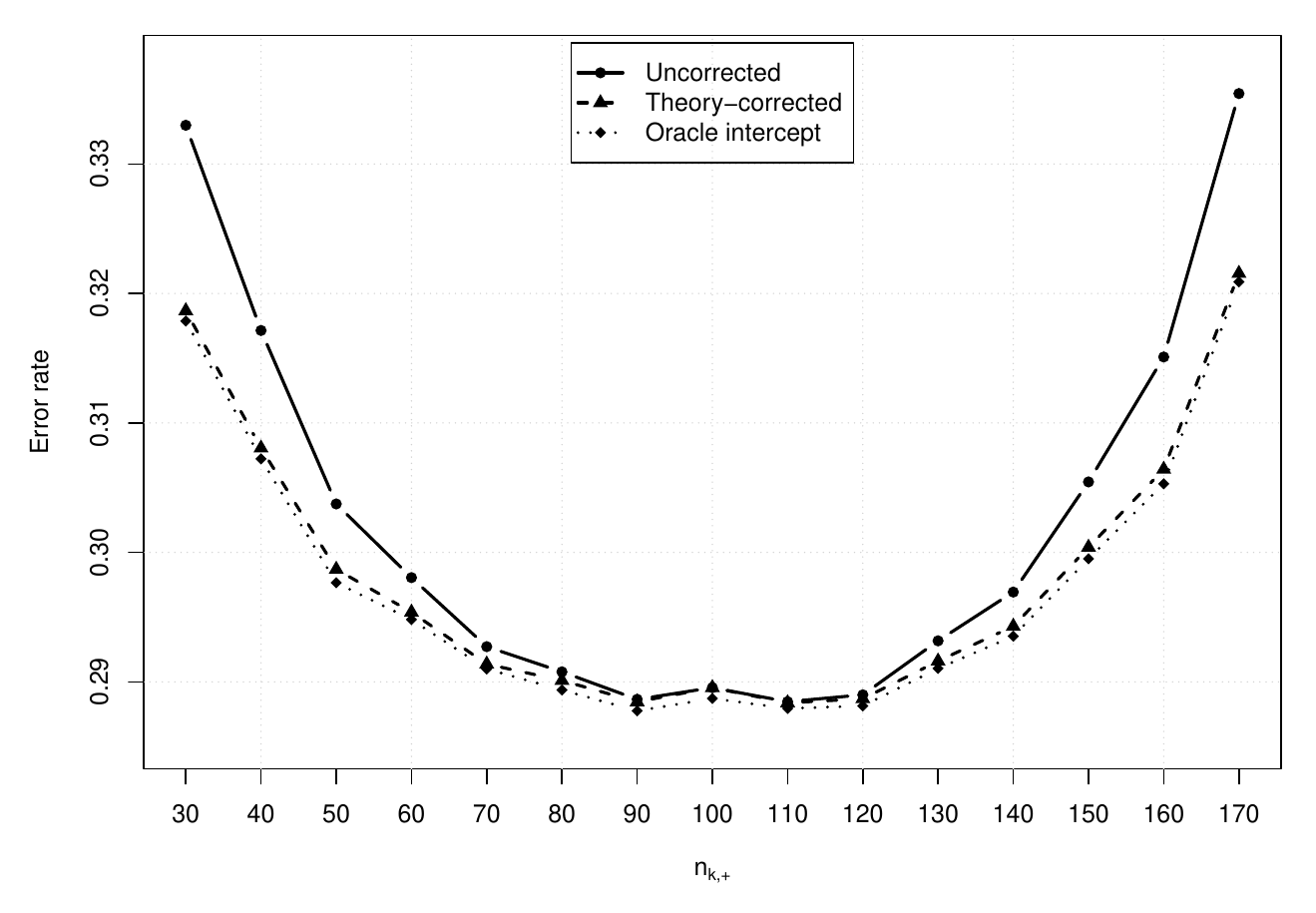}
			
			(a) Case 1: homogeneous covariance.
		\end{minipage}
		\hfill
		\begin{minipage}{0.48\textwidth}
			\centering
			\includegraphics[page=2,width=\textwidth]{lineplot_bias_correction_test1000_gd.pdf}
			
			(b) Case 2: heterogeneous covariance.
		\end{minipage}
		\caption{Effect of intercept bias correction under class imbalance.}
		\label{fig:bias_correction}
	\end{figure}
	
	Figure~\ref{fig:bias_correction} shows that the uncorrected intercept suffers from a clear increase in test error when the class sizes are highly unbalanced. The theoretically corrected intercept substantially reduces this error in highly unbalanced cases and is close to the oracle intercept selected on the test set. Around the balanced case $n_{k,+}=100$, all three intercepts perform similarly, as expected since the bias term is close to zero.
	
	\section{Experiments}\label{sec:experiments}
	We evaluate the proposed methods on two real multi-site biomedical datasets: ADHD-200 and the Parkinson's Progression Markers Initiative (PPMI). In both datasets, each site is treated as one domain. We select one site as the target domain in turn and use the remaining selected sites as source domains. The proposed TLDA-E and TLDA-O methods are compared with LDA, pooled LDA, and the transfer regularized discriminant analysis methods of \cite{zhang2025Transfera}. Specifically, LDA uses only the target-domain training data, while LDA-POOL pools the training data from all domains. When the sample covariance matrix in LDA is singular, a small ridge regularization term is added for numerical stability.
	
	\subsection{ADHD-200 dataset}
	We first evaluate the proposed methods on the ADHD-200 multi-site resting-state functional magnetic resonance imaging dataset. ADHD-200 is a publicly available neuroimaging dataset collected from 8 independent imaging sites, comprising 973 subjects in total, including 585 typically developing controls, 362 ADHD patients, and 26 subjects with missing diagnostic labels. Due to differences across sites in scanning equipment, acquisition parameters, and sample composition, ADHD-200 exhibits substantial multi-site heterogeneity. After removing duplicate and incomplete samples and aligning the phenotype information with the imaging features, we selected 4 sites comprising a total of 310 subjects for analysis. Each selected site was treated as the target domain in turn, while the remaining selected sites were used as source domains.
	
	For each target site, 20\% of the subjects are randomly assigned to the test set, and the remaining subjects are used as the training set. 
	Feature screening is conducted exclusively on the training data, using a two-sample \(t\)-test. 
	The classification accuracy of different methods is then compared under feature dimensions \(p=60,70,80,90,100\). All results are reported as averages across 50 random training-test splits, where $n$ represents the number of target-domain training samples.
	\begin{table}[!htbp]
		\centering
		\caption{Classification accuracy on the ADHD-200 dataset.}
		\label{tab:adhd_results}
		\small
		\setlength{\tabcolsep}{6pt}
		\renewcommand{\arraystretch}{1.12}
		\begin{tabular}{cccccccc}
			\toprule
			Target site & ~~~~\(p\)~~~~ & LDA & LDA-POOL & TL-RDA & TLP-RDA & TLDA-E & TLDA-O \\
			\midrule
			\multirow{5}{*}{\shortstack{site 1\\~\\(n=115)}}
			& 60 & 0.693 & 0.684 & 0.624 & 0.622 & 0.732 & \textbf{0.741} \\
			& 70 & 0.673 & 0.676 & 0.529 & 0.539 & 0.721 & \textbf{0.744} \\
			& 80 & 0.674 & 0.689 & 0.565 & 0.527 & 0.740 & \textbf{0.766} \\
			& 90 & 0.656 & 0.665 & 0.590 & 0.516 & 0.719 & \textbf{0.759} \\
			& 100 & 0.645 & 0.684 & 0.584 & 0.535 & 0.715 & \textbf{0.756} \\
			\midrule
			\multirow{5}{*}{\shortstack{site 2\\~\\(n=60)}}
			& 60 & 0.520 & 0.583 & 0.639 & 0.599 & \textbf{0.640} & 0.557 \\
			& 70 & 0.544 & 0.586 & 0.570 & 0.527 & \textbf{0.674} & 0.557 \\
			& 80 & 0.530 & 0.564 & 0.541 & 0.510 & \textbf{0.643} & 0.577 \\
			& 90 & 0.526 & 0.549 & 0.519 & 0.539 & \textbf{0.627} & 0.560 \\
			& 100 & 0.590 & 0.593 & 0.496 & 0.559 & \textbf{0.639} & 0.594 \\
			\midrule
			\multirow{5}{*}{\shortstack{site 3\\~\\(n=28)}}
			& 60 & 0.633 & 0.667 & 0.410 & 0.463 & 0.750 & \textbf{0.753} \\
			& 70 & 0.603 & 0.613 & 0.457 & 0.567 & \textbf{0.763} & 0.680 \\
			& 80 & 0.607 & 0.567 & 0.440 & 0.597 & 0.673 & \textbf{0.697} \\
			& 90 & 0.543 & 0.573 & 0.453 & 0.623 & \textbf{0.693} & 0.683 \\
			& 100 & 0.577 & 0.613 & 0.423 & 0.600 & \textbf{0.710} & 0.680 \\
			\midrule
			\multirow{5}{*}{\shortstack{site 4\\~\\(n=47)}}
			& 60 & 0.507 & 0.565 & 0.522 & 0.553 & \textbf{0.600} & 0.545 \\
			& 70 & 0.551 & 0.575 & 0.533 & 0.478 & \textbf{0.607} & 0.558 \\
			& 80 & 0.585 & 0.547 & 0.507 & 0.473 & \textbf{0.591} & 0.538 \\
			& 90 & 0.564 & 0.500 & 0.525 & 0.535 & \textbf{0.625} & 0.536 \\
			& 100 & 0.553 & 0.558 & 0.493 & 0.520 & \textbf{0.591} & 0.535 \\
			\bottomrule
		\end{tabular}
	\end{table}
	
	Table~\ref{tab:adhd_results} reports the classification accuracy on the ADHD-200 dataset. The proposed TLDA methods achieve the best performance in all reported settings. TLDA-O is consistently best for site 1, while TLDA-E performs better for site 2 and site 4. For site 3, both TLDA methods outperform the baseline methods, with the better method depending on the feature dimension. These results indicate that transfer learning is effective on this multi-site neuroimaging dataset, and that the relative advantage of TLDA-O and TLDA-E depends on the target site.
	
	\subsection{PPMI dataset}
	We next evaluate the proposed methods on the Parkinson’s Progression Markers Initiative (PPMI) dataset. PPMI is a publicly available multi-center dataset for Parkinson’s disease research, covering multimodal information such as clinical data, imaging data, genetic information, biospecimens, and biomarkers. It is mainly used to study the onset and progression of Parkinson’s disease. In the binary classification task, Parkinson’s disease patients were treated as the positive class, while healthy controls and prodromal subjects were combined as the negative class. After removing duplicate and incomplete samples and matching and aligning the features, we selected the 4 acquisition sites with the largest sample sizes, comprising a total of 612 subjects, for analysis. Each selected site was treated as the target domain in turn, while the remaining selected sites were used as source domains.
	
	For the target site, approximately 20\% of the subjects are held out as the test set, and the remaining subjects are used as the training set. 
	Feature screening is conducted exclusively on the training data using a two-sample \(t\)-test. 
	The classification accuracy of different methods is then compared under feature dimensions \(p=160,170,180,190,200\). All results are reported as averages across 50 random training-test splits, where $n$ represents the number of target-domain training samples.
	
	\begin{table}[!htbp]
		\centering
		\caption{Classification accuracy on the PPMI dataset.}
		\label{tab:ppmi_results}
		\small
		\setlength{\tabcolsep}{6pt}
		\renewcommand{\arraystretch}{1.12}
		\begin{tabular}{cccccccc}
			\toprule
			Target site & ~~~~\(p\)~~~~ & LDA & LDA-POOL & TL-RDA & TLP-RDA & TLDA-E & TLDA-O \\
			\midrule
			\multirow{5}{*}{\shortstack{site 1\\~\\(n=218)}}
			& 160 & 0.582 & 0.619 & 0.518 & 0.483 & 0.500 & \textbf{0.629} \\
			& 170 & 0.600 & 0.613 & 0.502 & 0.487 & 0.497 & \textbf{0.642} \\
			& 180 & 0.586 & 0.600 & 0.524 & 0.505 & 0.478 & \textbf{0.643} \\
			& 190 & 0.600 & 0.592 & 0.518 & 0.511 & 0.490 & \textbf{0.635} \\
			& 200 & 0.573 & 0.594 & 0.507 & 0.503 & 0.492 & \textbf{0.624} \\
			\midrule
			\multirow{5}{*}{\shortstack{site 2\\~\\(n=122)}}
			& 160 & 0.584 & 0.605 & 0.503 & 0.502 & \textbf{0.733} & 0.689 \\
			& 170 & 0.567 & 0.607 & 0.519 & 0.515 & \textbf{0.736} & 0.678 \\
			& 180 & 0.592 & 0.595 & 0.515 & 0.491 & \textbf{0.735} & 0.675 \\
			& 190 & 0.599 & 0.599 & 0.513 & 0.502 & \textbf{0.732} & 0.674 \\
			& 200 & 0.594 & 0.593 & 0.509 & 0.473 & \textbf{0.736} & 0.667 \\
			\midrule
			\multirow{5}{*}{\shortstack{site 3\\~\\(n=98)}}
			& 160 & 0.570 & 0.607 & 0.519 & 0.524 & \textbf{0.696} & 0.644 \\
			& 170 & 0.552 & 0.589 & 0.506 & 0.481 & \textbf{0.688} & 0.653 \\
			& 180 & 0.545 & 0.585 & 0.490 & 0.490 & \textbf{0.702} & 0.668 \\
			& 190 & 0.547 & 0.566 & 0.493 & 0.487 & \textbf{0.698} & 0.668 \\
			& 200 & 0.554 & 0.579 & 0.483 & 0.502 & \textbf{0.696} & 0.665 \\
			\midrule
			\multirow{5}{*}{\shortstack{site 4\\~\\(n=53)}}
			& 160 & 0.536 & 0.544 & 0.540 & 0.540 & 0.525 & \textbf{0.676} \\
			& 170 & 0.493 & 0.487 & 0.498 & 0.482 & 0.498 & \textbf{0.665} \\
			& 180 & 0.529 & 0.518 & 0.513 & 0.522 & 0.496 & \textbf{0.649} \\
			& 190 & 0.555 & 0.529 & 0.547 & 0.564 & 0.509 & \textbf{0.647} \\
			& 200 & 0.516 & 0.520 & 0.516 & 0.516 & 0.498 & \textbf{0.640} \\
			\bottomrule
		\end{tabular}
	\end{table}
	Table~\ref{tab:ppmi_results} reports the classification accuracy on the PPMI dataset. The proposed TLDA methods again achieve the best performance in all reported settings. TLDA-O performs best for site 1 and site 4, whereas TLDA-E gives the highest accuracy for site 2 and site 3. This site-dependent pattern suggests that the homogeneous and heterogeneous covariance versions capture different aspects of cross-site transfer.

	\section{Conclusion}\label{sec:conclusion}
	
	This paper develops a transfer learning framework for high-dimensional LDA with a shared classification signal. By decomposing the mean difference in each domain into a deterministic common component and a domain-specific random deviation, the proposed model separates transferable discriminative information from domain heterogeneity.
	
	Under spiked covariance models, we derive deterministic limits for the target-domain Gaussian-calibrated error under both homogeneous and heterogeneous covariance settings. The resulting formulas lead to oracle transfer weights and consistent plug-in estimators. We also identify the intercept bias caused by unbalanced target-domain class sample sizes and provide an asymptotically optimal correction. The empirical studies on the ADHD-200 and PPMI datasets show that the proposed TLDA methods achieve the best accuracy in all reported settings, with TLDA-O and TLDA-E showing complementary site-dependent advantages.
	
	There are several possible directions for future work. 
	First, the present model focuses on linear discriminant rules, while nonlinear classification methods may be more suitable for some complex data. Second, it would be interesting to study how to select useful source domains automatically when some source domains are weakly related to the target domain. Finally, future work may consider more general types of domain heterogeneity beyond the spiked covariance models studied in this paper.
	\newpage
	\bibliographystyle{apalike}
	\bibliography{TLDA}
	\newpage
	\appendix	
	\section{Appendix}
	\subsection{Proof of Theorem \ref{thm:centered_overlap_prob}}
	Throughout the proof of Theorem~\ref{thm:centered_overlap_prob}, we take \(\sigma^2=1\) without loss of generality, since multiplying the covariance matrix by a positive scalar does not change its eigenvectors or spectral projectors. Let
	$\bm Z=[\bm Z_1,\dots,\bm Z_K]^{\top}\in\mathbb R^{n\times p}$ denote the stacked whitened data matrix. We first present the following lemmas.
	\begin{lemma}
		\label{thm:uncentered_overlap_target_short}
		Suppose that Assumptions \ref{ass1}--\ref{ass5} hold. Define the uncentered pooled sample covariance matrix
		\[
		\widetilde{\bm S}_n^{\,e}
		:=
		\frac{1}{n}\bm\Sigma^{1/2}\bm Z^{\top}\bm Z\bm\Sigma^{1/2}.
		\]
		For each $j\in\mathcal I$, let $\widetilde\ell_j$ be the sample eigenvalue of
		$\widetilde{\bm S}_n^{\,e}$ associated with the population spike $\lambda_{j}$, and let
		$\widetilde{\bm v}_j$ be a corresponding eigenvector. Then, for every deterministic unit vector $\bm\xi\in\mathbb R^p$,
		\[
		\bm\xi^{\top}\bm{\widetilde v}_{j}\bm{\widetilde v}_j^{\top}\bm\xi
		\;\xrightarrow{\mathrm{a.s.}}\;
		\frac{\lambda_{j}^2-\gamma}{\lambda_{j}(\lambda_{j}+\gamma)}
		\,
		\bm\xi^{\top}\bm v_{j}\bm v_{j}^{\top}\bm\xi,
		\qquad j\in\mathcal I.
		\]
	\end{lemma}
	
	\begin{proof}
		Fix $j\in\mathcal I$. Let
		\[
		\bm B:=\frac{1}{n}\bm Z^{\top}\bm Z,
		\qquad
		\bm R(z):=(\bm B-z\bm I_p)^{-1},
		\qquad
		\widetilde{\bm G}(z):=(\widetilde{\bm S}_n^{\,e}-z\bm I_p)^{-1}.
		\]
		Set
		\[
		\gamma_-:=(1-\sqrt\gamma)^2,
		\qquad
		\gamma_+:=(1+\sqrt\gamma)^2,
		\]
		and define
		\[
		\mathfrak S_\gamma
		:=
		\begin{cases}
			[\gamma_-,\gamma_+], & 0<\gamma\le 1,\\[1mm]
			\{0\}\cup[\gamma_-,\gamma_+], & \gamma>1.
		\end{cases}
		\]
		Also write
		\[
		\bm V:=(\bm v_{l})_{l\in\mathcal I}\in\mathbb R^{p\times r},
		\qquad
		\bm D:=\operatorname{diag}\!\left(\frac{\lambda_l}{1+\lambda_l}\right)_{l\in\mathcal I}.
		\]
		Since
		\[
		\bm \Sigma^{-1}=\bm I_p-\bm V\bm D\bm V^{\top},
		\]
		we have
		\[
		\widetilde{\bm S}_n^{\,e}-z\bm I_p
		=
		\bm\Sigma^{1/2}(\bm B-z\bm \Sigma^{-1})\bm \Sigma^{1/2}
		=
		\bm\Sigma^{1/2}\bigl(\bm R(z)^{-1}+z\bm V\bm D\bm V^{\top}\bigr)\bm\Sigma^{1/2}.
		\]
		Hence
		\[
		\bm{\widetilde G}(z)
		=
		\bm \Sigma^{-1/2}
		\bigl(\bm R(z)^{-1}+z\bm V\bm D\bm V^{\top}\bigr)^{-1}
		\bm\Sigma^{-1/2}.
		\]
		
		Choose $\rho_j>0$ such that
		\[
		0<\rho_j<
		\frac12\min\Bigl\{
		\min_{l\in\mathcal I,l\neq j}|\lambda_{j}-\lambda_l|,
		|\lambda_{j}-\sqrt{\gamma}|,
		|\lambda_{j}+\sqrt{\gamma}|
		\Bigr\},
		\]
		and define
		\[
		\Gamma_j:=\{\zeta\in\mathbb C:|\zeta-\lambda_{j}|=\rho_j\},
		\qquad
		\vartheta(\zeta):=1+\zeta+\gamma+\frac{\gamma}{\zeta},
		\qquad
		\mathcal C_j:=\vartheta(\Gamma_j).
		\]
		By \cite{baik2006Eigenvalues},
		$\widetilde\ell_{j}\to \vartheta(\lambda_{j})$ almost surely, and $\mathcal C_j$
		encloses exactly the eigenvalue $\widetilde\ell_{j}$ of $\widetilde{\bm S}_n^{\,e}$ for all
		sufficiently large $p$. Therefore, define
		$
		\bm{\widetilde P}_{j}:=\widetilde{\bm v}_{j}\bm{\widetilde v}_{j}^{\top}
		$,
		then
		\[
		\bm{\widetilde P}_{j}
		=
		-\frac{1}{2\pi i}\oint_{\mathcal C_j}\bm{\widetilde G}(z)\,dz
		\]
		for all sufficiently large $p$, and hence
		\[
		\bm\xi^{\top}\bm{\widetilde P}_{j}\bm\xi
		=
		-\frac{1}{2\pi i}\oint_{\mathcal C_j}\bm\xi^{\top}\bm{\widetilde G}(z)\bm\xi\,dz.
		\]
		
		Define
		$
		\widetilde{\bm\xi}:=\bm\Sigma^{-1/2}\bm\xi.
		$
		By the matrix inversion lemma,
		\[
		\bm\xi^{\top}\bm{\widetilde G}(z)\bm\xi
		=
		\widetilde{\bm\xi}^{\top}\bm R(z)\widetilde{\bm\xi}
		-
		z\,\widetilde{\bm\xi}^{\top}\bm R(z)\bm V
		\bigl(\bm D^{-1}+z\bm V^{\top}\bm R(z)\bm V\bigr)^{-1}
		\bm V^{\top}\bm R(z)\widetilde{\bm\xi}.
		\]
		We next replace the random resolvent blocks by deterministic equivalents.
		
		Since the entries of $\bm Z$ are i.i.d. with mean zero, variance one, and finite fourth moment, Theorem~1 of \cite{bai2007Asymptotics} yields that for every fixed
		$
		z\in\mathbb C\setminus \mathfrak S_\gamma
		$
		and every deterministic unit vector $\bm\xi$,
		\begin{align}\label{eqR}
			\bm\xi^{\top}\bm R(z)\bm\xi-m_{\gamma}(z)\to 0
			\qquad\text{a.s.},
		\end{align}
		where $m_{\gamma}(z)$ is the Stieltjes transform of the Marchenko--Pastur law with
		ratio $\gamma$. To justify the uniform convergence on $\mathcal C_j$, fix any two deterministic vectors
		$\bm a,\bm b$ in the finite set
		\[
		\mathcal A:=\{\widetilde{\bm\xi}\}\cup\{\bm v_l:l\in\mathcal I\}.
		\]
		Define
		\[
		\Delta_p^{\bm a,\bm b}(z):=\bm a^\top \bm R(z)\bm b-m_\gamma(z)\bm a^\top \bm b.
		\]
		By \eqref{eqR} and polarization, for every fixed
		$z\in\mathbb C\setminus\mathfrak S_\gamma$,
		\[
		\Delta_p^{\bm a,\bm b}(z)\to 0
		\qquad\text{a.s.}
		\]
		
		Now let
		$
		\delta_j:=\operatorname{dist}(\mathcal C_j,\mathfrak S_\gamma)>0
		$.
		Since $\mathcal C_j$ is a fixed compact contour disjoint from $\mathfrak S_\gamma$, almost surely
		for all sufficiently large $p$,
		\[
		\operatorname{dist}(\mathcal C_j,\operatorname{spec}(\bm B))\ge \frac{\delta_j}{2}.
		\]
		Hence, for all $z\in\mathcal C_j$,
		\[
		\|\bm R(z)\|
		\le \frac{2}{\delta_j}
		\]
		almost surely for all sufficiently large $p$.
		
		Moreover, by the resolvent identity,
		\[
		\bm R(z_1)-\bm R(z_2)=(z_1-z_2)\bm R(z_1)\bm R(z_2),
		\]
		so that for $z_1,z_2\in\mathcal C_j$,
		\[
		\|\bm R(z_1)-\bm R(z_2)\|
		\le |z_1-z_2|\,\|\bm R(z_1)\|\,\|\bm R(z_2)\|
		\le \frac{4}{\delta_j^2}|z_1-z_2|.
		\]
		Since $m_\gamma$ is analytic on a neighborhood of $\mathcal C_j$, it is Lipschitz on
		$\mathcal C_j$, i.e., there exists a constant $L_j>0$ such that
		\[
		|m_\gamma(z_1)-m_\gamma(z_2)|\le L_j|z_1-z_2|,
		\qquad z_1,z_2\in\mathcal C_j.
		\]
		Therefore,
		\[
		|\Delta_p^{\bm a,\bm b}(z_1)-\Delta_p^{\bm a,\bm b}(z_2)|
		\le C_j |z_1-z_2|
		\]
		for all $z_1,z_2\in\mathcal C_j$, almost surely for all sufficiently large $p$, where $C_j>0$
		is deterministic.
		
		Fix $\varepsilon>0$, and choose a finite $\varepsilon$-net
		$\{z_1,\dots,z_{N_\varepsilon}\}\subset \mathcal C_j$. Then for any $z\in\mathcal C_j$, there
		exists $z_\ell$ such that $|z-z_\ell|\le \varepsilon$, and hence
		\[
		|\Delta_p^{\bm a,\bm b}(z)|
		\le
		|\Delta_p^{\bm a,\bm b}(z_\ell)|
		+
		C_j\varepsilon.
		\]
		Taking the supremum over $z\in\mathcal C_j$ gives
		\[
		\sup_{z\in\mathcal C_j}|\Delta_p^{\bm a,\bm b}(z)|
		\le
		\max_{1\le \ell\le N_\varepsilon}|\Delta_p^{\bm a,\bm b}(z_\ell)|
		+
		C_j\varepsilon.
		\]
		Since the net is finite and $\Delta_p^{\bm a,\bm b}(z_\ell)\to 0$ almost surely for each
		$\ell$, we obtain
		\[
		\limsup_{p\to\infty}\sup_{z\in\mathcal C_j}|\Delta_p^{\bm a,\bm b}(z)|
		\le C_j\varepsilon
		\qquad\text{a.s.}
		\]
		for every $\varepsilon>0$. Since $\varepsilon>0$ is arbitrary, it follows that
		\[
		\limsup_{p\to\infty}\sup_{z\in\mathcal C_j}|\Delta_p^{\bm a,\bm b}(z)|
		=0
		\qquad\text{a.s.}
		\]
		That is,
		\[
		\sup_{z\in\mathcal C_j}|\Delta_p^{\bm a,\bm b}(z)|\to 0
		\qquad\text{a.s.}
		\]
		Since $\mathcal A$ is finite, this uniform convergence holds simultaneously for all pairs
		$\bm a,\bm b\in\mathcal A$. In particular, uniformly for $z\in\mathcal C_j$,
		\[
		\widetilde{\bm\xi}^{\top}\bm R(z)\bm V
		=
		m_{\gamma}(z)\,\widetilde{\bm\xi}^{\top}\bm V+o(1),
		\qquad
		\bm V^{\top}\bm R(z)\bm V
		=
		m_{\gamma}(z)\bm I_r+o(1),
		\qquad\text{a.s.}
		\]
		
		Set
		$
		\bm L(z):=\bigl(\bm D^{-1}+zm_{\gamma}(z)\bm I_r\bigr)^{-1}.
		$
		For $z=\vartheta(\zeta)$ with $\zeta\in\Gamma_j$, the Marchenko--Pastur equation gives
		\[
		m_{\gamma}(\vartheta(\zeta))
		=
		-\frac{1}{\zeta+\gamma}.
		\]
		Hence, for every $l\in\mathcal I$,
		\[
		\frac{1+\lambda_l}{\lambda_l}
		+
		\vartheta(\zeta)m_{\gamma}(\vartheta(\zeta))
		=
		\frac{1}{\lambda_l}-\frac{1}{\zeta}.
		\]
		Since $\Gamma_j$ stays a positive distance away from every $\lambda_i$, the matrix
		$\bm D^{-1}+zm_{\gamma}(z)\bm I_r$ is uniformly invertible on $\mathcal C_j$. Thus,
		uniformly for $z\in\mathcal C_j$,
		\[
		\bigl(\bm D^{-1}+z\bm V^{\top}\bm R(z)\bm V\bigr)^{-1}
		=
		\bm L(z)+o(1),
		\qquad\text{a.s.}
		\]
		Consequently,
		\[
		\bm\xi^{\top}\bm{\widetilde G}(z)\bm\xi
		=
		\widetilde{\bm\xi}^{\top}\bm R(z)\widetilde{\bm\xi}
		-
		z\,m_{\gamma}(z)^2\,
		\widetilde{\bm\xi}^{\top}\bm V\bm L(z)\bm V^{\top}\widetilde{\bm\xi}
		+
		o(1)
		\]
		uniformly on $\mathcal C_j$, almost surely.
		
		Since $\mathcal C_j$ lies outside $\mathfrak S_\gamma$, it stays a positive distance away from
		the support of the $p$-dimensional Marchenko--Pastur law. If $0<\gamma\le 1$, the almost sure
		convergence of the extreme eigenvalues of $\bm B$ to $\gamma_-$ and $\gamma_+$
		(see \cite{yin1988Limita} for the upper edge and \cite{bai1993Limit} for the lower edge)
		implies that, almost surely for all sufficiently large $p$, the contour $\mathcal C_j$ and its
		interior are disjoint from $\operatorname{spec}(\bm B)$.
		
		If $\gamma>1$, then $\bm B=\frac{1}{n}\bm Z^\top\bm Z$ has exactly $p-n$ zero eigenvalues,
		and its nonzero extreme eigenvalues converge almost surely to $\gamma_-$ and $\gamma_+$.
		Since $\mathcal C_j$ surrounds the positive outlier location $\vartheta(\lambda_j)>\gamma_+$,
		it is disjoint from $\{0\}\cup[\gamma_-,\gamma_+]$, and hence also disjoint from
		$\operatorname{spec}(\bm B)$ for all sufficiently large $p$, almost surely.
		
		Therefore the function
		\[
		z\mapsto \widetilde{\bm\xi}^\top \bm R(z)\widetilde{\bm\xi}
		\]
		is analytic on and inside $\mathcal C_j$. Hence, by Cauchy's integral theorem,
		\[
		\oint_{\mathcal C_j}\widetilde{\bm\xi}^{\top}\bm R(z)\widetilde{\bm\xi}\,dz=0,
		\]
		and therefore
		\[
		\bm\xi^{\top}\bm{\widetilde P}_{j}\bm\xi
		=
		\frac{1}{2\pi i}\oint_{\mathcal C_j}
		z\,m_{\gamma}(z)^2\,
		\widetilde{\bm\xi}^{\top}\bm V\bm L(z)\bm V^{\top}\widetilde{\bm\xi}\,dz
		+
		o(1)
		\qquad\text{a.s.}
		\]
		
		Since $\bm L(z)$ is diagonal,
		\[
		\widetilde{\bm\xi}^{\top}\bm V\bm L(z)\bm V^{\top}\widetilde{\bm\xi}
		=
		\sum_{l\in\mathcal I}
		\frac{\lambda_l\langle \widetilde{\bm\xi},\bm v_l\rangle^2}
		{1+\lambda_l+\lambda_lzm_{\gamma}(z)}.
		\]
		Hence
		\[
		\bm\xi^{\top}\bm{\widetilde P}_j\bm\xi
		=
		\sum_{l\in\mathcal I}
		\frac{1}{2\pi i}\oint_{\mathcal C_j}
		zm_{\gamma}(z)^2
		\frac{\lambda_l\langle \widetilde{\bm\xi},\bm v_l\rangle^2}
		{1+\lambda_l+\lambda_lzm_{\gamma}(z)}\,dz
		+
		o(1).
		\]
		Now use the change of variables $z=\vartheta(\zeta)$ with $\zeta\in\Gamma_j$. Since
		\[
		m_{\gamma}(\vartheta(\zeta))
		=
		-\frac{1}{\zeta+\gamma},
		\qquad
		\frac{\lambda_l}{1+\lambda_l+\lambda_l\vartheta(\zeta)m_{\gamma}(\vartheta(\zeta))}
		=
		\frac{\zeta\lambda_l}{\zeta-\lambda_l},
		\]
		we obtain
		\[
		\bm\xi^{\top}\bm{\widetilde P}_j\bm\xi
		=
		\sum_{l\in\mathcal I}
		\frac{1}{2\pi i}\oint_{\Gamma_j}
		\vartheta(\zeta)\vartheta'(\zeta)m_{\gamma}(\vartheta(\zeta))^2
		\frac{\zeta\lambda_l}{\zeta-\lambda_l}
		\langle \widetilde{\bm\xi},\bm v_l\rangle^2\,d\zeta
		+
		o(1).
		\]
		Since $\Gamma_j$ encloses only $\lambda_j$, all terms with $l\neq j$ vanish by Cauchy's theorem. Thus
		\[
		\bm\xi^{\top}\bm{\widetilde P}_j\bm\xi
		=
		\frac{1}{2\pi i}\oint_{\Gamma_j}
		\vartheta(\zeta)\vartheta'(\zeta)m_{\gamma}(\vartheta(\zeta))^2
		\frac{\zeta\lambda_j}{\zeta-\lambda_j}
		\langle \widetilde{\bm\xi},\bm v_j\rangle^2\,d\zeta
		+
		o(1).
		\]
		
		Next, since
		$
		\bm\Sigma\bm v_j=(1+\lambda_j)\bm v_j
		$,
		we have
		\[
		\langle \widetilde{\bm\xi},\bm v_j\rangle^2
		=
		\frac{1}{1+\lambda_j}\langle \bm\xi,\bm v_j\rangle^2.
		\]
		Also,
		\[
		\vartheta(\zeta)=1+\zeta+\gamma+\frac{\gamma}{\zeta},\quad
		\vartheta'(\zeta)=1-\frac{\gamma}{\zeta^2},
		\quad
		m_{\gamma}(\vartheta(\zeta))
		=
		-\frac{1}{\zeta+\gamma}.
		\]
		Therefore the integrand equals
		\[
		\frac{\lambda_j}{1+\lambda_j}
		\cdot
		\frac{(1+\zeta)(\zeta^2-\gamma)}
		{\zeta^2(\zeta+\gamma)(\zeta-\lambda_j)}
		\,
		\langle \bm\xi,\bm v_j\rangle^2.
		\]
		Hence the residue at $\zeta=\lambda_j$ is
		\[
		\frac{\lambda_j}{1+\lambda_j}
		\cdot
		\frac{(1+\lambda_j)(\lambda_j^2-\gamma)}
		{\lambda_j^2(\lambda_j+\gamma)}
		\,
		\langle \bm\xi,\bm v_j\rangle^2
		=
		\frac{\lambda_j^2-\gamma}{\lambda_j(\lambda_j+\gamma)}
		\,
		\langle \bm\xi,\bm v_j\rangle^2.
		\]
		By the residue theorem,
		\[
		\bm\xi^{\top}\bm{\widetilde P}_j\bm\xi
		=
		\frac{\lambda_j^2-\gamma}{\lambda_j(\lambda_j+\gamma)}
		\,
		\langle \bm\xi,\bm v_j\rangle^2
		+
		o(1)
		\qquad\text{a.s.}
		\]
		Since
		\[
		\bm{\widetilde P}_j=\bm{\widetilde v}_j\bm{\widetilde v}_j^{\top},
		\qquad
		\langle \bm\xi,\bm v_j\rangle^2
		=
		\bm\xi^{\top}\bm v_j\bm v_j^{\top}\bm\xi,
		\]
		the desired conclusion follows.
	\end{proof}
	\begin{lemma}
		\label{lem:center_uncentered_fixed}
		Suppose that Assumptions \ref{ass1}--\ref{ass5} hold. Let
		\[
		\bm G_K
		:=
		\operatorname{diag}\!\bigl(
		\bm P_{n_{1,+}},\bm P_{n_{1,-}},\dots,\bm P_{n_{K,+}},\bm P_{n_{K,-}}
		\bigr),
		\qquad
		\bm P_m:=\bm I_m-\frac{1}{m}\bm 1_m\bm 1_m^\top.
		\]
		Define
		\[
		\widetilde{\bm S}_n^{\,e}
		:=
		\frac1n\,\bm\Sigma^{1/2}\bm Z^\top \bm Z\,\bm\Sigma^{1/2},
		\qquad
		\widetilde{\bm S}_n
		:=
		\frac1n\,\bm\Sigma^{1/2}\bm Z^\top \bm G_K \bm Z\,\bm\Sigma^{1/2}.
		\]
		For \(z\in\mathbb C^+\), let
		\[
		\bm Q_n^{\,e}(z):=(\widetilde{\bm S}_n^{\,e}-z\bm I_p)^{-1},
		\qquad
		\bm Q_n(z):=(\widetilde{\bm S}_n-z\bm I_p)^{-1}.
		\]
		Then for any deterministic unit vector $\bm\xi\in\mathbb R^p$ and any fixed
		$z\in\mathbb C^+$,
		\[
		\bm\xi^\top\bm Q_n(z)\bm\xi-\bm\xi^\top\bm Q_n^{\,e}(z)\bm\xi
		\;\xrightarrow{\mathrm{p}}\;0.
		\]
	\end{lemma}
	
	\begin{proof}
		The proof follows the same argument as Proposition~2 of
		\cite{liu2025Asymptotic}. In that proposition, the effect of the usual
		centering matrix
		\[
		\bm\Phi=\bm I_n-n^{-1}\bm 1_n\bm 1_n^\top
		\]
		is controlled by a rank-one resolvent perturbation. In the present setting,
		the block-centering matrix can be written as
		\[
		\bm G=\bm I_n-\bm U\bm U^\top ,
		\]
		where the columns of \(\bm U\) are the normalized block indicator vectors and
		\(\bm U^\top\bm U=\bm I_q\), with \(q=2K\) fixed. Thus the scalar correction
		terms in the proof of Proposition~2 are replaced by \(q\times q\) matrix
		corrections. Since \(q\) is fixed, the same resolvent bounds, the same
		Sherman--Morrison--Woodbury expansion, and the same moment estimates apply
		without changing the order of the error terms. Consequently, for any
		deterministic unit vector \(\bm\xi\) and any fixed \(z\in\mathbb C^+\),
		\[
		\bm\xi^\top\bm Q_n(z)\bm\xi
		-
		\bm\xi^\top\bm Q_n^{\,e}(z)\bm\xi
		\xrightarrow{p}0 .
		\]
	\end{proof}
	
	\begin{lemma}\label{lem:centered_uncentered_same_outlier_limit}
		Suppose that Assumptions \ref{ass1}--\ref{ass5} hold. Let $\widetilde\ell_j$ and $\widehat\ell_j$ denote the sample outlier eigenvalues
		of $\widetilde{\bm S}_n^{\,e}$ and $\widetilde{\bm S}_n$, respectively, associated with
		the population spike $\lambda_j,j\in\mathcal I$.
		
		Define
		\[
		\theta_j:=1+\lambda_j+\gamma+\frac{\gamma}{\lambda_j}.
		\]
		Then
		\[
		\widehat\ell_j \xrightarrow{p} \theta_j.
		\]
		In particular,
		\[
		\widehat\ell_j-\widetilde\ell_j \xrightarrow{p}0.
		\]
	\end{lemma}
	
	\begin{proof}
		For convenience, write
		\[
		\bm B_n:=\frac1n\bm Z^\top \bm G_K \bm Z,
		\qquad
		\bm B_n^{e}:=\frac1n\bm Z^\top \bm Z,
		\]
		so that
		\[
		\widetilde{\bm S}_n=\bm\Sigma^{1/2}\bm B_n\bm\Sigma^{1/2},
		\qquad
		\widetilde{\bm S}_n^{\,e}=\bm\Sigma^{1/2}\bm B_n^{e}\bm\Sigma^{1/2}.
		\]
		Also recall that
		\[
		\bm\Sigma^{-1}=\bm I_p-\bm V\bm D\bm V^\top,
		\qquad
		\bm D=\operatorname{diag}\!\left(\frac{\lambda_\ell}{1+\lambda_\ell}\right)_{\ell\in\mathcal I}.
		\]
		
		Set
		\[
		\gamma_-:=(1-\sqrt{\gamma})^2,\qquad
		\gamma_+:=(1+\sqrt{\gamma})^2,
		\]
		and
		\[
		\mathfrak S_\gamma
		:=
		\begin{cases}
			[\gamma_-,\gamma_+], & 0<\gamma\le 1,\\[1mm]
			\{0\}\cup[\gamma_-,\gamma_+], & \gamma>1.
		\end{cases}
		\]
		For each \(\ell\in\mathcal I\), write
		\[
		\theta_\ell:=1+\lambda_\ell+\gamma+\frac{\gamma}{\lambda_\ell}.
		\]
		Let
		\[
		\vartheta(\zeta):=1+\zeta+\gamma+\frac{\gamma}{\zeta}.
		\]
		Since
		\[
		\vartheta'(\lambda_j)=1-\frac{\gamma}{\lambda_j^2}\neq0,
		\]
		we may choose \(\rho_j>0\) small enough such that \(\vartheta\) is one-to-one
		on a neighborhood of \(\overline{B(\lambda_j,\rho_j)}\), the disk
		\(\overline{B(\lambda_j,\rho_j)}\) contains no \(\lambda_\ell\) with
		\(\ell\neq j\), and
		\[
		\vartheta\bigl(\overline{B(\lambda_j,\rho_j)}\bigr)
		\cap
		\mathfrak S_\gamma
		=
		\varnothing.
		\]
		Define
		\[
		\Gamma_j:=\{\zeta\in\mathbb C:|\zeta-\lambda_j|=\rho_j\},
		\qquad
		\mathcal C_j:=\vartheta(\Gamma_j),
		\qquad
		\mathcal D_j:=\vartheta\bigl(\overline{B(\lambda_j,\rho_j)}\bigr).
		\]
		Then \(\mathcal C_j\) is a simple closed contour enclosing
		\(\theta_j=\vartheta(\lambda_j)\), and \(\mathcal D_j\) is the closed domain
		enclosed by \(\mathcal C_j\).

		Since \(\bm G_K\) is an orthogonal projection with rank \(n-2K\), the nonzero
		eigenvalues of $\bm B_n$ are those of a centered empirical covariance matrix with dependent coordinates.
		By the extremal eigenvalue convergence result for such empirical covariance
		matrices in \cite{chafai2018Convergence}, the spectrum of \(\bm B_n\) is asymptotically confined to the
		Marchenko--Pastur support \(\mathfrak S_\gamma\). The same conclusion holds for
		\(\bm B_n^{e}\) by the classical Bai--Yin edge convergence; see
		\cite{yin1988Limita,bai1993Limit}. Since \(\mathcal D_j\) lies a positive
		distance away from \(\mathfrak S_\gamma\), there exists \(c_j>0\) such that
		\[
		\mathbb P(E_n)\to 1,
		\]
		where
		\[
		E_n:=
		\left\{
		\operatorname{dist}\bigl(\mathcal D_j,\operatorname{spec}(\bm B_n^{e})\bigr)\ge c_j,\ 
		\operatorname{dist}\bigl(\mathcal D_j,\operatorname{spec}(\bm B_n)\bigr)\ge c_j
		\right\}.
		\]
		Define
		\[
		\bm T_n(z):=\bm V^\top(\bm B_n-z\bm I_p)^{-1}\bm V,
		\qquad
		\bm T_n^{e}(z):=\bm V^\top(\bm B_n^{e}-z\bm I_p)^{-1}\bm V,
		\]
		and
		\[
		\bm M_n(z):=\bm D^{-1}+z\bm T_n(z).
		\]
		On the event $E_n$, both $\bm T_n^{e}(z)$ and $\bm T_n(z)$ are analytic on a neighborhood of
		\(\mathcal D_j\), and
		\[
		\sup_{z\in\mathcal C_j}\|(\bm B_n^{e}-z\bm I_p)^{-1}\|\le c_j^{-1},
		\qquad
		\sup_{z\in\mathcal C_j}\|(\bm B_n-z\bm I_p)^{-1}\|\le c_j^{-1}.
		\]
		
		We now prove that
		\[
		\sup_{z\in\mathcal C_j}
		\left\|
		\bm T_n(z)-m_\gamma(z)\bm I_r
		\right\|
		\xrightarrow{p}0.
		\]
		To this end, for fixed $z\in\mathbb C^+$, define
		\[
		\bm N_n(z):=\bm V^\top(\bm B_n-z\bm\Sigma^{-1})^{-1}\bm V,
		\qquad
		\bm N_n^{e}(z):=\bm V^\top(\bm B_n^{e}-z\bm\Sigma^{-1})^{-1}\bm V.
		\]
		Since
		\[
		\widetilde{\bm S}_n-z\bm I_p
		=
		\bm\Sigma^{1/2}(\bm B_n-z\bm\Sigma^{-1})\bm\Sigma^{1/2},
		\]
		we have
		\[
		(\bm B_n-z\bm\Sigma^{-1})^{-1}
		=
		\bm\Sigma^{1/2}(\widetilde{\bm S}_n-z\bm I_p)^{-1}\bm\Sigma^{1/2},
		\]
		and similarly for the uncentered matrix. Therefore, by polarization and
		Lemma~\ref{lem:center_uncentered_fixed}, for every fixed $z\in\mathbb C^+$,
		\[
		\bm N_n(z)-\bm N_n^{e}(z)\xrightarrow{p}0.
		\]
		
		On the other hand, for every fixed $z\in\mathbb C^+$, the deterministic-equivalent step in the
		uncentered proof gives
		\[
		\bm T_n^{e}(z)\xrightarrow{\mathrm{a.s.}} m_\gamma(z)\bm I_r.
		\]
		Using the Woodbury identity,
		\[
		(\bm B_n-z\bm\Sigma^{-1})^{-1}
		=
		(\bm B_n-z\bm I_p)^{-1}
		-
		(\bm B_n-z\bm I_p)^{-1}\bm V
		\bigl((z\bm D)^{-1}+\bm T_n(z)\bigr)^{-1}
		\bm V^\top(\bm B_n-z\bm I_p)^{-1},
		\]
		hence
		\[
		\bm N_n(z)
		=
		\bm T_n(z)-\bm T_n(z)\bigl((z\bm D)^{-1}+\bm T_n(z)\bigr)^{-1}\bm T_n(z)
		=
		\bm T_n(z)\bigl(\bm I_r+z\bm D\bm T_n(z)\bigr)^{-1}.
		\]
		Thus
		\[
		(\bm I_r-z\bm N_n(z)\bm D)\bm T_n(z)=\bm N_n(z),
		\]
		so that
		\[
		\bm T_n(z)=\Phi_z(\bm N_n(z)),
		\qquad
		\Phi_z(\bm X):=(\bm I_r-z\bm X\bm D)^{-1}\bm X.
		\]
		Similarly,
		\[
		\bm T_n^{e}(z)=\Phi_z(\bm N_n^{e}(z)).
		\]
		
		Now define
		\[
		\bm N(z):=m_\gamma(z)\bigl(\bm I_r+z\,m_\gamma(z)\bm D\bigr)^{-1}.
		\]
		Then $\bm N_n^{e}(z)\xrightarrow{\mathrm{a.s.}}\bm N(z)$, and therefore
		\[
		\bm N_n(z)\xrightarrow{p}\bm N(z)
		\qquad
		(z\in\mathbb C^+\ \text{fixed}).
		\]
		Since \(\bm I_r-z\bm N(z)\bm D\) is invertible for \(z\in\mathbb C^+\), the map \(\Phi_z\) is
		continuous at \(\bm N(z)\), and hence
		\[
		\bm T_n(z)=\Phi_z(\bm N_n(z))\xrightarrow{p}\Phi_z(\bm N(z))=m_\gamma(z)\bm I_r
		\qquad
		(z\in\mathbb C^+\ \text{fixed}).
		\]
		Because the matrices are real symmetric,
		\[
		\bm T_n(\bar z)=\overline{\bm T_n(z)},
		\qquad
		m_\gamma(\bar z)=\overline{m_\gamma(z)},
		\]
		so the same convergence holds for every fixed \(z\in \mathcal C_j\setminus\mathbb R\).
		
		We now upgrade pointwise convergence to uniform convergence on \(\mathcal C_j\).
		On the event \(E_n\), the resolvent identity implies that for \(z,w\in\mathcal C_j\),
		\begin{align*}
			\|\bm T_n(z)-\bm T_n(w)\|
			&=
			\left\|
			\bm V^\top\bigl[(\bm B_n-z\bm I_p)^{-1}-(\bm B_n-w\bm I_p)^{-1}\bigr]\bm V
			\right\|\\
			&\le
			|z-w|\,\|\bm V\|^2\,
			\|(\bm B_n-z\bm I_p)^{-1}\|\,
			\|(\bm B_n-w\bm I_p)^{-1}\|\\
			&\le
			c_j^{-2}|z-w|.
		\end{align*}
		Since \(\mathcal C_j\setminus\mathbb R\) is dense in \(\mathcal C_j\), we may choose the net
		inside \(\mathcal C_j\setminus\mathbb R\). Thus \(\{\bm T_n\}\) is stochastically equicontinuous on \(\mathcal C_j\).
		Since \(m_\gamma\) is continuous on \(\mathcal C_j\), a standard finite-net argument yields
		\[
		\sup_{z\in\mathcal C_j}
		\left\|
		\bm T_n(z)-m_\gamma(z)\bm I_r
		\right\|
		\xrightarrow{p}0.
		\]
		Consequently,
		\[
		\sup_{z\in\mathcal C_j}
		\left\|
		\bm M_n(z)-\bm M(z)
		\right\|
		\xrightarrow{p}0,
		\]
		where
		\[
		\bm M(z):=\bm D^{-1}+z\,m_\gamma(z)\bm I_r.
		\]
		
		Define
		\[
		f_n(z):=\det \bm M_n(z),
		\qquad
		f(z):=\det \bm M(z).
		\]
		Since the determinant is continuous on the finite-dimensional space of \(r\times r\) matrices,
		\[
		\sup_{z\in\mathcal C_j}|f_n(z)-f(z)|\xrightarrow{p}0.
		\]
		
		We now analyze \(f(z)\). By the identity already derived in the uncentered proof,
		for \(z=\vartheta(\zeta)=1+\zeta+\gamma+\gamma/\zeta\),
		\[
		\frac{1+\lambda_\ell}{\lambda_\ell}+\vartheta(\zeta)m_\gamma(\vartheta(\zeta))
		=
		\frac{1}{\lambda_\ell}-\frac{1}{\zeta},
		\qquad \ell\in\mathcal I.
		\]
		Hence
		\[
		f(\vartheta(\zeta))
		=
		\prod_{\ell\in\mathcal I}
		\left(
		\frac{1}{\lambda_\ell}-\frac{1}{\zeta}
		\right).
		\]
		Since \(\Gamma_j\) encloses only \(\lambda_j\), and \(\vartheta\) is one-to-one
		on a neighborhood of \(\overline{B(\lambda_j,\rho_j)}\), the function \(f\) has
		exactly one zero inside \(\mathcal C_j\), namely
		\[
		\theta_j=\vartheta(\lambda_j).
		\]
		Moreover this zero is simple because
		\[
		\vartheta'(\lambda_j)=1-\frac{\gamma}{\lambda_j^2}\neq0.
		\]
		In particular,
		\[
		\inf_{z\in\mathcal C_j}|f(z)|>0.
		\]
		
		Fix
		\[
		\eta_j:=\frac12\inf_{z\in\mathcal C_j}|f(z)|>0.
		\]
		Then
		\[
		\mathbb P\!\left(\sup_{z\in\mathcal C_j}|f_n(z)-f(z)|<\eta_j\right)\to 1.
		\]
		On the event
		\[
		E_n\cap\left\{\sup_{z\in\mathcal C_j}|f_n(z)-f(z)|<\eta_j\right\},
		\]
		the functions \(f_n\) and \(f\) are analytic on a neighborhood of \(\mathcal D_j\), and satisfy
		\[
		|f_n(z)-f(z)|<|f(z)|,\qquad z\in\mathcal C_j.
		\]
		Hence, by Rouch\'e's theorem, \(f_n\) and \(f\) have the same number of zeros inside
		\(\mathcal C_j\), namely exactly one.
		
		Since \(E_n\) also guarantees that \(\mathcal C_j\) is disjoint from \(\operatorname{spec}(\bm B_n)\),
		the zeros of \(f_n\) inside \(\mathcal C_j\) are precisely the eigenvalues of
		\(\widetilde{\bm S}_n\) inside \(\mathcal C_j\). Hence, with probability tending to one,
		\(\widetilde{\bm S}_n\) has exactly one eigenvalue in \(\mathcal C_j\). Denote it by
		\(\widehat\ell_j\).
		
		Since \(\operatorname{diam}(\mathcal D_j)\to0\) as \(\rho_j\downarrow0\), the above implies
		\[
		\widehat\ell_j\xrightarrow{p}\theta_j.
		\]
		Together with
		$
		\widetilde\ell_j\xrightarrow{\mathrm{a.s.}}\theta_j
		$,
		we obtain
		\[
		\widehat\ell_j-\widetilde\ell_j\xrightarrow{p}0.
		\]
		This completes the proof.
	\end{proof}
	Next, we formally begin the proof of Theorem~\ref{thm:centered_overlap_prob}. Recall that
	\[
	\widetilde{\bm S}_n
	:=
	\frac1n\,\bm\Sigma^{1/2}\bm Z^\top \bm G_K \bm Z\bm\Sigma^{1/2},
	\qquad
	\widetilde{\bm S}_n^{\,e}
	:=
	\frac1n\,\bm\Sigma^{1/2}\bm Z^\top \bm Z\bm\Sigma^{1/2}.
	\]
	By the definition of the pooled sample covariance matrix,
	\[
	\bm S_n=\frac{n}{n-2K}\,\widetilde{\bm S}_n.
	\]
	Hence \(\bm S_n\) and \(\widetilde{\bm S}_n\) have exactly the same eigenvectors. It is therefore enough to prove the claimed limit for the eigenvectors of \(\widetilde{\bm S}_n\).
	
	For $z\in\mathbb C\setminus\mathbb R$, define the resolvents
	\[
	\bm Q_n(z):=(\widetilde{\bm S}_n-z\bm I_p)^{-1},
	\qquad
	\bm Q_n^{\,e}(z):=(\widetilde{\bm S}_n^{\,e}-z\bm I_p)^{-1}.
	\]
	
	Fix $j\in\mathcal I$, and define
	\[
	\theta_j:=1+\lambda_j+\gamma+\frac{\gamma}{\lambda_j}.
	\]
	Take the same \(\rho_j\), \(\Gamma_j\), \(\mathcal C_j\), and \(\mathcal D_j\)
	as in the proof of Lemma~\ref{lem:centered_uncentered_same_outlier_limit}.
	Then \(\mathcal C_j\) is a simple closed contour enclosing \(\theta_j\), and it is
	disjoint from the limiting bulk \(\mathfrak S_\gamma\) and from all the other
	outlier limits \(\{\theta_\ell:\ell\in\mathcal I,\ell\neq j\}\).
	
	By Lemma~\ref{lem:centered_uncentered_same_outlier_limit},
	\[
	\widetilde\ell_j \xrightarrow{\mathrm{a.s.}} \theta_j,\quad
	\widehat\ell_j \xrightarrow{p} \theta_j.
	\]
	Moreover, by the proof of Lemma~\ref{lem:centered_uncentered_same_outlier_limit}, for
	sufficiently small \(\rho_j\), there exists an event \(E_n\) such that
	\[
	\mathbb P(E_n)\to 1,
	\]
	and on \(E_n\) the contour \(\mathcal C_j\) encloses exactly one eigenvalue of
	\(\widetilde{\bm S}_n^{\,e}\), namely \(\widetilde\ell_j\), and exactly one eigenvalue of
	\(\widetilde{\bm S}_n\), namely \(\widehat\ell_j\), and moreover, for some constant \(c_j>0\),
	\[
	\operatorname{dist}(\mathcal C_j,\operatorname{spec}(\widetilde{\bm S}_n))\ge c_j,
	\qquad
	\operatorname{dist}(\mathcal C_j,\operatorname{spec}(\widetilde{\bm S}_n^{\,e}))\ge c_j.
	\]
	
	On the event $E_n$, define the spectral projectors
	\[
	\widetilde{\bm P}_j
	:=
	-\frac{1}{2\pi i}\oint_{\mathcal C_j}\bm Q_n^{\,e}(z)\,dz,
	\qquad
	\widehat{\bm P}_j
	:=
	-\frac{1}{2\pi i}\oint_{\mathcal C_j}\bm Q_n(z)\,dz.
	\]
	Since each contour encloses exactly one eigenvalue, both projectors are rank one, and
	\[
	\widetilde{\bm P}_j=\widetilde{\bm v}_j\widetilde{\bm v}_j^\top,
	\qquad
	\widehat{\bm P}_j=\widehat{\bm v}_j\widehat{\bm v}_j^\top.
	\]
	
	We next show that
	\begin{equation}
		\bm\xi^\top(\widehat{\bm P}_j-\widetilde{\bm P}_j)\bm\xi\xrightarrow{p}0.
		\label{eq:projector_bridge_prob}
	\end{equation}
	
	To this end, define
	\[
	D_n(z):=\bm\xi^\top\bm Q_n(z)\bm\xi-\bm\xi^\top\bm Q_n^{\,e}(z)\bm\xi,
	\qquad z\in\mathcal C_j.
	\]
	By Lemma~\ref{lem:center_uncentered_fixed}, for every fixed
	$z\in\mathcal C_j\cap\mathbb C^+$,
	\[
	D_n(z)\xrightarrow{p}0.
	\]
	Since $\widetilde{\bm S}_n$ and $\widetilde{\bm S}_n^{\,e}$ are real symmetric, we have
	\[
	\bm Q_n(\bar z)=\overline{\bm Q_n(z)},
	\qquad
	\bm Q_n^{\,e}(\bar z)=\overline{\bm Q_n^{\,e}(z)},
	\]
	and hence
	\[
	D_n(\bar z)=\overline{D_n(z)}.
	\]
	Therefore,
	\[
	D_n(z)\xrightarrow{p}0
	\qquad\text{for every fixed } z\in\mathcal C_j\setminus\mathbb R.
	\]
	We will combine this pointwise convergence on the dense subset
	$\mathcal C_j\setminus\mathbb R$ with stochastic equicontinuity on $\mathcal C_j$
	to obtain uniform convergence on the whole contour.
	
	We now upgrade this pointwise convergence to uniform convergence on \(\mathcal C_j\).
	On the event \(E_n\), the contour \(\mathcal C_j\) encloses exactly one eigenvalue of
	\(\widetilde{\bm S}_n^{\,e}\) and exactly one eigenvalue of \(\widetilde{\bm S}_n\), and moreover,
	for some constant \(c_j>0\),
	\[
	\operatorname{dist}(\mathcal C_j,\operatorname{spec}(\widetilde{\bm S}_n))\ge c_j,
	\qquad
	\operatorname{dist}(\mathcal C_j,\operatorname{spec}(\widetilde{\bm S}_n^{\,e}))\ge c_j.
	\]
	Hence, on \(E_n\),
	\[
	\sup_{z\in\mathcal C_j}\|\bm Q_n(z)\|\le c_j^{-1},
	\qquad
	\sup_{z\in\mathcal C_j}\|\bm Q_n^{\,e}(z)\|\le c_j^{-1}.
	\]
	
	Let $z,w\in\mathcal C_j$. By the resolvent identity,
	\[
	\bm Q_n(z)-\bm Q_n(w)=(w-z)\bm Q_n(z)\bm Q_n(w),
	\]
	\[
	\bm Q_n^{\,e}(z)-\bm Q_n^{\,e}(w)=(w-z)\bm Q_n^{\,e}(z)\bm Q_n^{\,e}(w).
	\]
	Therefore, on \(E_n\),
	\begin{align*}
		|D_n(z)-D_n(w)|
		&\le
		\bigl|\bm\xi^\top(\bm Q_n(z)-\bm Q_n(w))\bm\xi\bigr|
		+
		\bigl|\bm\xi^\top(\bm Q_n^{\,e}(z)-\bm Q_n^{\,e}(w))\bm\xi\bigr|\\
		&\le
		|z-w|\,\|\bm Q_n(z)\|\|\bm Q_n(w)\|
		+
		|z-w|\,\|\bm Q_n^{\,e}(z)\|\|\bm Q_n^{\,e}(w)\|\\
		&\le C_j|z-w|,
	\end{align*}
	where \(C_j>0\) is deterministic.
	
	Fix \(\varepsilon>0\). Since \(\mathcal C_j\setminus\mathbb R\) is dense in \(\mathcal C_j\),
	choose a finite \(\delta\)-net
	\(\{z_1,\dots,z_{N_\delta}\}\subset\mathcal C_j\setminus\mathbb R\), where
	\(\delta>0\) is chosen so small that \(C_j\delta\le\varepsilon/2\).
	Then for every \(z\in\mathcal C_j\), there exists \(z_\ell\) such that \(|z-z_\ell|\le \delta\).
	On \(E_n\),
	\[
	|D_n(z)|\le |D_n(z_\ell)|+C_j\delta
	\le |D_n(z_\ell)|+\frac{\varepsilon}{2}.
	\]
	Hence, on $E_n$,
	\[
	\sup_{z\in\mathcal C_j}|D_n(z)|
	\le
	\max_{1\le \ell\le N_\delta}|D_n(z_\ell)|+\frac{\varepsilon}{2}.
	\]
	Therefore
	\[
	\mathbb P\!\left(\sup_{z\in\mathcal C_j}|D_n(z)|>\varepsilon\right)
	\le
	\mathbb P(E_n^c)
	+
	\sum_{\ell=1}^{N_\delta}
	\mathbb P\!\left(|D_n(z_\ell)|>\frac{\varepsilon}{2}\right).
	\]
	Since $\mathbb P(E_n^c)\to 0$ and $D_n(z_\ell)\xrightarrow{p}0$ for each fixed $\ell$, we obtain
	\begin{equation}
		\sup_{z\in\mathcal C_j}
		\left|
		\bm\xi^\top\bm Q_n(z)\bm\xi-\bm\xi^\top\bm Q_n^{\,e}(z)\bm\xi
		\right|
		\xrightarrow{p}0.
		\label{eq:uniform_bridge_prob}
	\end{equation}
	
	Using the contour representations of the projectors, we get
	\begin{align*}
		\left|\bm\xi^\top(\widehat{\bm P}_j-\widetilde{\bm P}_j)\bm\xi\right|
		&=
		\left|
		\frac{1}{2\pi i}\oint_{\mathcal C_j}
		\left(
		\bm\xi^\top\bm Q_n(z)\bm\xi-\bm\xi^\top\bm Q_n^{\,e}(z)\bm\xi
		\right)\,dz
		\right|\\
		&\le
		\frac{\operatorname{length}(\mathcal C_j)}{2\pi}
		\sup_{z\in\mathcal C_j}
		\left|
		\bm\xi^\top\bm Q_n(z)\bm\xi-\bm\xi^\top\bm Q_n^{\,e}(z)\bm\xi
		\right|.
	\end{align*}
	Together with \eqref{eq:uniform_bridge_prob}, this proves \eqref{eq:projector_bridge_prob}.
	
	Finally, by Lemma~\ref{thm:uncentered_overlap_target_short},
	\[
	\bm\xi^\top \widetilde{\bm P}_j \bm\xi
	=
	\bm\xi^\top \widetilde{\bm v}_j\widetilde{\bm v}_j^\top \bm\xi
	\;\xrightarrow{\mathrm{a.s.}}\;
	\frac{\lambda_j^2-\gamma}{\lambda_j(\lambda_j+\gamma)}
	\,\bm\xi^\top \bm v_j\bm v_j^\top \bm\xi.
	\]
	Hence also
	\[
	\bm\xi^\top \widetilde{\bm P}_j \bm\xi
	\;\xrightarrow{p}\;
	\frac{\lambda_j^2-\gamma}{\lambda_j(\lambda_j+\gamma)}
	\,\bm\xi^\top \bm v_j\bm v_j^\top \bm\xi.
	\]
	Combining this with \eqref{eq:projector_bridge_prob}, we obtain
	\[
	\bm\xi^\top \widehat{\bm v}_j\widehat{\bm v}_j^\top \bm\xi
	=
	\bm\xi^\top \widehat{\bm P}_j\bm\xi
	=
	\bm\xi^\top \widetilde{\bm P}_j\bm\xi
	+
	o_p(1),
	\]
	and therefore
	\[
	\bm\xi^\top \widehat{\bm v}_j\widehat{\bm v}_j^\top \bm\xi
	\;\xrightarrow{p}\;
	\frac{\lambda_j^2-\gamma}{\lambda_j(\lambda_j+\gamma)}
	\,\bm\xi^\top \bm v_j\bm v_j^\top \bm\xi.
	\]
	This completes the proof.
	\subsection{Proof of Theorem \ref{thm:err_homo}}
	\begin{lemma}\label{lem:Ap_positive}
		Under Assumptions~\ref{ass1}--\ref{ass5}, the matrix \(\bm A_p\) is positive
		definite. More precisely,
		\[
		\lambda_{\min}(\bm A_p)
		\ge
		\min_{1\le k\le K}\tau_k>0.
		\]
	\end{lemma}
	
	\begin{proof}
		Write
		\[
		A_{kk',p}=s_p+d_k\ind_{\{k=k'\}},
		\]
		where
		\[
		s_p=
		\frac{\|\bar{\bm\mu}\|^2}{\sigma^2}
		+
		\sum_{j\in\mathcal I}c_j
		\frac{(\bar{\bm\mu}^{\top}\bm v_j)^2}{\sigma^2},
		\qquad
		d_k=\frac{\alpha_k^2}{\sigma^2}+\tau_k .
		\]
		Under Assumption~\ref{ass5}, \(0<a_j<1\). Moreover,
		\[
		1+c_j
		=
		(1+\lambda_j)(1-b_ja_j)^2+b_j^2a_j(1-a_j)>0 .
		\]
		Let \(\bar{\bm\mu}_{\perp}\) be the component of \(\bar{\bm\mu}\) orthogonal to
		\(\operatorname{span}\{\bm v_j:j\in\mathcal I\}\). Then
		\[
		s_p
		=
		\frac{1}{\sigma^2}
		\left[
		\|\bar{\bm\mu}_{\perp}\|^2
		+
		\sum_{j\in\mathcal I}
		(1+c_j)(\bar{\bm\mu}^{\top}\bm v_j)^2
		\right]
		\ge0 .
		\]
		Hence, for any \(\bm w\neq\bm0\),
		\[
		\bm w^\top\bm A_p\bm w
		=
		s_p\left(\sum_{k=1}^K w_k\right)^2
		+
		\sum_{k=1}^K d_kw_k^2
		\ge
		\min_{1\le k\le K}\tau_k\,\|\bm w\|^2>0 .
		\]
		This proves the claim.
	\end{proof}
	
	\begin{lemma}\label{lem:ek_vhat_zero}
		Suppose that Assumptions~\ref{ass1}--\ref{ass5} hold. For each \(k=1,\dots,K\), define
		\[
		e_k:=\widehat{\bm\mu}_k-\bm\mu_k.
		\]
		Then, for every fixed \(k\in\{1,\dots,K\}\) and every fixed \(j\in\mathcal I\),
		\[
		e_k^\top \widehat{\bm v}_j \xrightarrow{p} 0.
		\]
	\end{lemma}
	
	\begin{proof}
		Without loss of generality, we take \(\sigma^2=1\), since multiplying the covariance matrix by a positive scalar does not change its eigenvectors or spectral projectors. Recall that
		\[
		\widehat{\bm\mu}_k
		=
		\widehat{\bm\mu}_{+1,k}-\widehat{\bm\mu}_{-1,k},
		\qquad
		\bm\mu_k=\bm\mu_{+1,k}-\bm\mu_{-1,k}.
		\]
		Under Assumption~\ref{ass1}, write
		\[
		(\bm X_k)_i
		=
		\bm\mu_{(y_k)_i,k}
		+
		\bm\Sigma^{1/2}(\bm z_k)_i,
		\]
		where \((\bm z_k)_i\in\mathbb R^p\) denotes the corresponding whitened sample vector. Define the classwise whitened sample means by
		\[
		\bar{\bm z}_{k,\pm1}
		:=
		\frac{1}{n_{k,\pm1}}\sum_{i:(y_k)_i=\pm1}(\bm z_k)_i.
		\]
		Then
		\[
		\widehat{\bm\mu}_{\pm1,k}
		=
		\bm\mu_{\pm1,k}
		+
		\bm\Sigma^{1/2}\bar{\bm z}_{k,\pm1},
		\]
		and hence
		\[
		e_k
		=
		\widehat{\bm\mu}_k-\bm\mu_k
		=
		\bm\Sigma^{1/2}\bigl(\bar{\bm z}_{k,+}-\bar{\bm z}_{k,-}\bigr).
		\]
		Therefore it suffices to prove that, for each fixed sign \(s\in\{+,-\}\),
		\[
		\bar{\bm z}_{k,s}^\top \bm\Sigma^{1/2}\widehat{\bm v}_j \xrightarrow{p} 0.
		\]
		Fix one such block mean and write
		\[
		\bar{\bm z}:=\bar{\bm z}_{k,s}.
		\]
		
		We first express \(\widehat{\bm v}_j\) through the spectral projector associated with the
		centered pooled sample covariance matrix. Recall from the proof of
		Theorem~\ref{thm:centered_overlap_prob} that
		\[
		\widetilde{\bm S}_n
		=
		\frac1n\,\bm\Sigma^{1/2}\bm Z^\top \bm G_K \bm Z\bm\Sigma^{1/2},
		\qquad
		\bm S_n=\frac{n}{n-2K}\,\widetilde{\bm S}_n,
		\]
		so that \(\bm S_n\) and \(\widetilde{\bm S}_n\) have exactly the same eigenvectors. In
		particular, \(\widehat{\bm v}_j\) is also the eigenvector of \(\widetilde{\bm S}_n\)
		associated with the centered outlier \(\widehat\ell_j\).
		
		Let
		\[
		\widehat{\bm P}_j:=\widehat{\bm v}_j\widehat{\bm v}_j^\top.
		\]
		By the proof of
		Theorem~\ref{thm:centered_overlap_prob}, we may take the same contour
		\(\mathcal C_j=\vartheta(\Gamma_j)\) as above. With probability tending to one,
		\(\mathcal C_j\) encloses exactly one eigenvalue of \(\widetilde{\bm S}_n\), namely
		\(\widehat\ell_j\), and
		\[
		\widehat{\bm P}_j
		=
		-\frac{1}{2\pi i}\oint_{\mathcal C_j}\bm Q_n(z)\,dz,
		\qquad
		\bm Q_n(z):=(\widetilde{\bm S}_n-z\bm I_p)^{-1}.
		\]
		Also,
		\[
		\widehat{\bm P}_j\bm v_j
		=
		(\widehat{\bm v}_j^\top \bm v_j)\widehat{\bm v}_j.
		\]
		Hence
		\[
		\bar{\bm z}^\top \bm\Sigma^{1/2}\widehat{\bm v}_j
		=
		\frac{\bar{\bm z}^\top \bm\Sigma^{1/2}\widehat{\bm P}_j\bm v_j}
		{\widehat{\bm v}_j^\top \bm v_j},
		\]
		whenever \(\widehat{\bm v}_j^\top \bm v_j\neq 0\).
		
		We now show that the denominator is bounded away from zero with probability tending to one.
		By Theorem~\ref{thm:centered_overlap_prob}, taking \(\bm\xi=\bm v_j\), we have
		\[
		\bm v_j^\top \widehat{\bm P}_j \bm v_j
		=
		\bm v_j^\top \widehat{\bm v}_j\widehat{\bm v}_j^\top \bm v_j
		=
		|\widehat{\bm v}_j^\top \bm v_j|^2
		\xrightarrow{p}
		\frac{\lambda_j^2-\gamma}{\lambda_j(\lambda_j+\gamma)}.
		\]
		Since Assumption~\ref{ass5} implies
		\[
		\frac{\lambda_j^2-\gamma}{\lambda_j(\lambda_j+\gamma)}>0,
		\]
		there exists a constant \(c_j>0\) such that
		\[
		\mathbb P\bigl(|\widehat{\bm v}_j^\top \bm v_j|\ge c_j\bigr)\to 1.
		\]
		Therefore, it remains to prove that
		\[
		\bar{\bm z}^\top \bm\Sigma^{1/2}\widehat{\bm P}_j\bm v_j \xrightarrow{p}0.
		\]
		
		Using the contour representation of \(\widehat{\bm P}_j\),
		\[
		\bar{\bm z}^\top \bm\Sigma^{1/2}\widehat{\bm P}_j\bm v_j
		=
		-\frac{1}{2\pi i}\oint_{\mathcal C_j}
		\bar{\bm z}^\top \bm\Sigma^{1/2}\bm Q_n(z)\bm v_j\,dz.
		\]
		Thus it suffices to show that
		\begin{equation}\label{eq:blockmean_resolvent_uniform}
			\sup_{z\in\mathcal C_j}
			\left|
			\bar{\bm z}^\top \bm\Sigma^{1/2}\bm Q_n(z)\bm v_j
			\right|
			\xrightarrow{p}0.
		\end{equation}
		
		Recall
		\[
		\bm Q_n^{\,e}(z):=(\widetilde{\bm S}_n^{\,e}-z\bm I_p)^{-1},
		\qquad
		\widetilde{\bm S}_n^{\,e}
		=
		\frac1n\,\bm\Sigma^{1/2}\bm Z^\top \bm Z\bm\Sigma^{1/2}.
		\]
		We first prove that, for any fixed block mean \(\bar{\bm z}_{a,t}\) and any
		deterministic unit vector \(\bm\xi\),
		\[
		\bar{\bm z}_{a,t}^{\top}\bm\Sigma^{1/2}\bm Q_n^{\,e}(z)\bm\xi
		\xrightarrow{p}0,\qquad z\in\mathbb C^+\ \text{fixed}.
		\]
		Fix the block \((a,t)\), write \(m=n_{a,t}\), and denote its whitened
		observations by \(\bm z_1,\ldots,\bm z_m\). Let
		\[
		\bm y_i=\frac1{\sqrt n}\bm\Sigma^{1/2}\bm z_i,
		\qquad
		\bm Q(z)=\bm Q_n^{\,e}(z).
		\]
		Then
		\[
		\bar{\bm z}_{a,t}^{\top}\bm\Sigma^{1/2}\bm Q_n^{\,e}(z)\bm\xi
		=
		\frac{\sqrt n}{m}\sum_{i=1}^m \bm y_i^\top \bm Q(z)\bm\xi .
		\]
		It is enough to prove
		\[
		L:=
		\frac{\sqrt n}{m}\sum_{i=1}^m \bm y_i^\top \bm Q(z)\bm\xi=o_p(1).
		\]
		
		Let \(\mathcal G\) be the \(\sigma\)-field generated by all observations outside
		this block. Define
		\[
		\mathcal F_0=\mathcal G,\qquad
		\mathcal F_i=\sigma(\mathcal G,\bm z_1,\ldots,\bm z_i),
		\quad i=1,\ldots,m.
		\]
		Write
		\[
		L-E(L\mid\mathcal G)=\sum_{i=1}^m \Delta_i,\qquad
		\Delta_i=E(L\mid\mathcal F_i)-E(L\mid\mathcal F_{i-1}).
		\]
		Let \(L^{(i)}\) be obtained from \(L\) by replacing \(\bm y_i\) with an
		independent copy \(\bm y_i'\). Then
		\[
		E(L^{(i)}\mid\mathcal F_{i-1})=E(L\mid\mathcal F_{i-1}),
		\]
		and hence
		\[
		\Delta_i=E(L-L^{(i)}\mid\mathcal F_i).
		\]
		By Jensen's inequality and the orthogonality of martingale differences,
		\[
		E|L-E(L\mid\mathcal G)|^2
		=
		\sum_{i=1}^m E|\Delta_i|^2
		\le
		\sum_{i=1}^m E|L-L^{(i)}|^2 .
		\]
		
		For each \(i\), let \(\bm Q_{(i)}(z)\) be the resolvent with \(\bm y_i\) removed, and set
		\[
		\beta_i(z)=\frac{1}{1+\bm y_i^\top\bm Q_{(i)}(z)\bm y_i},
		\qquad
		\beta_i'(z)=\frac{1}{1+(\bm y_i')^\top\bm Q_{(i)}(z)\bm y_i'} .
		\]
		Since \(z\in\mathbb C^+\) is fixed,
		\[
		|\beta_i(z)|+|\beta_i'(z)|\le C_z,\qquad
		\|\bm Q_{(i)}(z)\|\le C_z .
		\]
		By the Sherman--Morrison formula,
		\[
		\bm Q(z)=\bm Q_{(i)}(z)
		-\beta_i(z)\bm Q_{(i)}(z)\bm y_i\bm y_i^\top\bm Q_{(i)}(z),
		\]
		and the analogous identity holds after replacing \(\bm y_i\) by \(\bm y_i'\). Thus
		\[
		\begin{aligned}
			L-L^{(i)}
			=&
			\frac{\sqrt n}{m}
			\Big[
			\beta_i(z)\,\bm y_i^\top\bm Q_{(i)}(z)\bm\xi
			-\beta_i'(z)\,(\bm y_i')^\top\bm Q_{(i)}(z)\bm\xi  \\
			&
			-\beta_i(z)
			\Big(\sum_{\ell\ne i}\bm y_\ell^\top\bm Q_{(i)}(z)\bm y_i\Big)
			\bm y_i^\top\bm Q_{(i)}(z)\bm\xi \\
			&
			+\beta_i'(z)
			\Big(\sum_{\ell\ne i}\bm y_\ell^\top\bm Q_{(i)}(z)\bm y_i'\Big)
			(\bm y_i')^\top\bm Q_{(i)}(z)\bm\xi
			\Big].
		\end{aligned}
		\]
		
		Let
		\[
		\mathcal H_i
		=
		\sigma\bigl(\mathcal G,\{\bm y_\ell:\ell\ne i,\ 1\le \ell\le m\}\bigr).
		\]
		Conditional on \(\mathcal H_i\), \(\bm Q_{(i)}(z)\) and
		\(\sum_{\ell\ne i}\bm y_\ell\) are fixed, and \(\bm y_i\) is independent of them.
		Thus, for any \(\mathcal H_i\)-measurable matrices \(\bm A,\bm B\) with
		\(\|\bm A\|+\|\bm B\|\le C_z\) and any \(\mathcal H_i\)-measurable vectors
		\(\bm u,\bm v\), the finite fourth moment assumption gives
		\[
		E\left(
		|\bm y_i^\top\bm A\bm u|^2
		\mid \mathcal H_i
		\right)
		\le
		\frac{C_z}{n}\|\bm u\|^2 .
		\]
		Moreover, by the Cauchy--Schwarz inequality,
		\[
		\begin{aligned}
			&E\left(
			|\bm y_i^\top\bm A\bm u|^2
			|\bm y_i^\top\bm B\bm v|^2
			\mid \mathcal H_i
			\right) \\
			&\qquad\le
			\left\{
			E\left(
			|\bm y_i^\top\bm A\bm u|^4
			\mid \mathcal H_i
			\right)
			\right\}^{1/2}
			\left\{
			E\left(
			|\bm y_i^\top\bm B\bm v|^4
			\mid \mathcal H_i
			\right)
			\right\}^{1/2} \\
			&\qquad\le
			\frac{C_z}{n^2}\|\bm u\|^2\|\bm v\|^2 .
		\end{aligned}
		\]
		Therefore,
		\[
		E\left(
		|\beta_i(z)\bm y_i^\top\bm Q_{(i)}(z)\bm\xi|^2
		\mid\mathcal H_i
		\right)
		\le \frac{C_z}{n},
		\]
		and
		\[
		\begin{aligned}
			&E\left(
			\left|
			\beta_i(z)
			\Big(\sum_{\ell\ne i}\bm y_\ell^\top\bm Q_{(i)}(z)\bm y_i\Big)
			\bm y_i^\top\bm Q_{(i)}(z)\bm\xi
			\right|^2
			\mid\mathcal H_i
			\right)  \\
			&\qquad\le
			\frac{C_z}{n^2}
			\left\|\sum_{\ell\ne i}\bm y_\ell\right\|^2 .
		\end{aligned}
		\]
		The same bounds hold for the terms involving \(\bm y_i'\). Hence, by
		\(|x_1+\cdots+x_4|^2\le 4\sum_{r=1}^4|x_r|^2\),
		\[
		E|L-L^{(i)}|^2
		\le
		\frac{Cn}{m^2}
		\left\{
		\frac1n
		+
		\frac1{n^2}
		E\left\|\sum_{\ell\ne i}\bm y_\ell\right\|^2
		\right\}=O(n^{-2}).
		\]
		Consequently,
		\[
		E|L-E(L\mid\mathcal G)|^2
		\le
		\sum_{i=1}^m O(n^{-2})
		=O(n^{-1}),
		\]
		and hence
		\[
		L-E(L\mid\mathcal G)=o_p(1).
		\]
		
		It remains to prove \(E(L\mid\mathcal G)=o_p(1)\). By exchangeability,
		\[
		E(L\mid\mathcal G)
		=
		\sqrt n\,
		E\left(\bm y_1^\top \bm Q(z)\bm\xi\mid\mathcal G\right).
		\]
		Let
		\[
		\mathcal H_1=\sigma(\mathcal G,\bm y_2,\ldots,\bm y_m).
		\]
		Conditional on \(\mathcal H_1\), \(\bm Q_{(1)}(z)\) is fixed and
		\[
		\bm y_1^\top\bm Q(z)\bm\xi
		=
		\frac{\bm y_1^\top\bm Q_{(1)}(z)\bm\xi}
		{1+\bm y_1^\top\bm Q_{(1)}(z)\bm y_1}.
		\]
		Set
		\[
		c_1=1+\frac1n\operatorname{tr}\{\bm\Sigma\bm Q_{(1)}(z)\},
		\qquad
		\Delta_1=
		\bm y_1^\top\bm Q_{(1)}(z)\bm y_1
		-\frac1n\operatorname{tr}\{\bm\Sigma\bm Q_{(1)}(z)\}.
		\]
		Since
		\[
		E\left(\bm y_1^\top\bm Q_{(1)}(z)\bm\xi\mid\mathcal H_1\right)=0,
		\]
		we have
		\[
		\begin{aligned}
			&\left|
			E\left(
			\bm y_1^\top\bm Q(z)\bm\xi
			\mid\mathcal H_1
			\right)
			\right|=
			\left|
			E\left(
			\frac{
				\bm y_1^\top\bm Q_{(1)}(z)\bm\xi
			}{
				c_1+\Delta_1
			}
			\mid\mathcal H_1
			\right)
			\right|=
			\left|
			E\left(
			-\frac{
				\bm y_1^\top\bm Q_{(1)}(z)\bm\xi\,\Delta_1
			}{
				c_1(c_1+\Delta_1)
			}
			\mid\mathcal H_1
			\right)
			\right| .
		\end{aligned}
		\]
		By the standard resolvent denominator bounds, as in the bounds for
		\(\beta_i(z)\) and its deterministic counterpart,
		\[
		|c_1|^{-1}\le C_z,\qquad |c_1+\Delta_1|^{-1}\le C_z .
		\]
		Thus
		\[
		\begin{aligned}
			&\left|
			E\left(
			\bm y_1^\top\bm Q(z)\bm\xi
			\mid\mathcal H_1
			\right)
			\right|  \\
			&\qquad\le
			C_z
			\left\{
			E\left(
			|\bm y_1^\top\bm Q_{(1)}(z)\bm\xi|^2
			\mid\mathcal H_1
			\right)
			\right\}^{1/2}
			\left\{
			E\left(
			|\Delta_1|^2
			\mid\mathcal H_1
			\right)
			\right\}^{1/2}.
		\end{aligned}
		\]
		The finite fourth moment assumption gives
		\[
		E\left(
		|\bm y_1^\top\bm Q_{(1)}(z)\bm\xi|^2
		\mid\mathcal H_1
		\right)
		\le \frac{C_z}{n},
		\]
		and
		\[
		E\left(
		|\Delta_1|^2
		\mid\mathcal H_1
		\right)
		\le \frac{C_z}{n}.
		\]
		Therefore,
		\[
		\left|
		E\left(
		\bm y_1^\top\bm Q(z)\bm\xi
		\mid\mathcal H_1
		\right)
		\right|
		\le C_z n^{-1}.
		\]
		Taking conditional expectation with respect to \(\mathcal G\), we get
		\[
		\left|
		E\left(
		\bm y_1^\top\bm Q(z)\bm\xi
		\mid\mathcal G
		\right)
		\right|
		\le C_z n^{-1}.
		\]
		Consequently,
		\[
		E(L\mid\mathcal G)=O(n^{-1/2})=o(1).
		\]
		Combining the two parts gives \(L=o_p(1)\). Therefore
		\[
		\bar{\bm z}_{a,t}^{\top}\bm\Sigma^{1/2}\bm Q_n^{\,e}(z)\bm\xi
		\xrightarrow{p}0 .
		\]
		
		Taking \(\bm\xi=\bm v_j\), we obtain
		\[
		\bar{\bm z}^\top \bm\Sigma^{1/2}\bm Q_n^{\,e}(z)\bm v_j \xrightarrow{p}0
		\qquad
		(z\in\mathbb C^+\ \text{fixed}).
		\]
		
		Next, by the resolvent identity,
		\[
		\bm Q_n(z)-\bm Q_n^{\,e}(z)
		=
		\bm Q_n(z)\bm R_n \bm Q_n^{\,e}(z),
		\]
		where
		\[
		\bm R_n
		=
		\sum_{a=1}^K \sum_{t\in\{+,-\}} \frac{n_{a,t}}{n}\,
		\bm\Sigma^{1/2}\bar{\bm z}_{a,t}\bar{\bm z}_{a,t}^\top\bm\Sigma^{1/2}.
		\]
		Hence
		\begin{align*}
			\bar{\bm z}^\top \bm\Sigma^{1/2}\bigl(\bm Q_n(z)-\bm Q_n^{\,e}(z)\bigr)\bm v_j
			&=
			\sum_{a=1}^K\sum_{t\in\{+,-\}}
			\frac{n_{a,t}}{n}
			\Bigl(
			\bar{\bm z}^\top \bm\Sigma^{1/2}\bm Q_n(z)\bm\Sigma^{1/2}\bar{\bm z}_{a,t}
			\Bigr)
			\Bigl(
			\bar{\bm z}_{a,t}^\top \bm\Sigma^{1/2}\bm Q_n^{\,e}(z)\bm v_j
			\Bigr).
		\end{align*}
		For each fixed \((a,t)\), the second factor converges to zero in probability by the previous
		paragraph, while the first factor is \(O_p(1)\), because
		\[
		\left|
		\bar{\bm z}^\top \bm\Sigma^{1/2}\bm Q_n(z)\bm\Sigma^{1/2}\bar{\bm z}_{a,t}
		\right|
		\le
		\|\bm\Sigma\|\,\|\bm Q_n(z)\|\,\|\bar{\bm z}\|\,\|\bar{\bm z}_{a,t}\|,
		\]
		and \(\|\bm Q_n(z)\|\le (\Im z)^{-1}\), while \(\|\bar{\bm z}\|=O_p(1)\) and
		\(\|\bar{\bm z}_{a,t}\|=O_p(1)\). Since \(K\) is fixed, the sum has only finitely many terms.
		Therefore
		\[
		\bar{\bm z}^\top \bm\Sigma^{1/2}\bm Q_n(z)\bm v_j
		-
		\bar{\bm z}^\top \bm\Sigma^{1/2}\bm Q_n^{\,e}(z)\bm v_j
		\xrightarrow{p}0
		\qquad
		(z\in\mathbb C^+\ \text{fixed}),
		\]
		and hence
		\[
		\bar{\bm z}^\top \bm\Sigma^{1/2}\bm Q_n(z)\bm v_j \xrightarrow{p}0
		\qquad
		(z\in\mathbb C^+\ \text{fixed}).
		\]
		Because \(\widetilde{\bm S}_n\) is real symmetric, the same convergence holds for every fixed
		\(z\in\mathcal C_j\setminus\mathbb R\).
		
		We now upgrade the above pointwise convergence to uniform convergence on \(\mathcal C_j\).
		By the proof of Theorem~\ref{thm:centered_overlap_prob}, there exists an event \(E_n\) such that
		\[
		\mathbb P(E_n)\to 1
		\]
		and, on \(E_n\), for some constant \(c_j>0\),
		\[
		\operatorname{dist}(\mathcal C_j,\operatorname{spec}(\widetilde{\bm S}_n))\ge c_j.
		\]
		Hence, on \(E_n\),
		\[
		\sup_{z\in\mathcal C_j}\|\bm Q_n(z)\|\le c_j^{-1}.
		\]
		For \(z,w\in\mathcal C_j\), the resolvent identity gives
		\[
		\bm Q_n(z)-\bm Q_n(w)=(w-z)\bm Q_n(z)\bm Q_n(w),
		\]
		so that, on \(E_n\),
		\begin{align*}
			\left|
			\bar{\bm z}^\top \bm\Sigma^{1/2}\bm Q_n(z)\bm v_j
			-
			\bar{\bm z}^\top \bm\Sigma^{1/2}\bm Q_n(w)\bm v_j
			\right|
			&\le
			\|\bar{\bm z}\|\,\|\bm\Sigma^{1/2}\|\,\|\bm Q_n(z)-\bm Q_n(w)\|\\
			&\le
			\|\bar{\bm z}\|\,\|\bm\Sigma^{1/2}\|\,
			|z-w|\,\|\bm Q_n(z)\|\,\|\bm Q_n(w)\|\\
			&\le
			C_j\|\bar{\bm z}\|\,|z-w|,
		\end{align*}
		where \(C_j>0\) is deterministic.
		
		Fix \(\varepsilon>0\). Choose \(M>0\) such that
		\[
		\mathbb P(\|\bar{\bm z}\|>M)<\varepsilon.
		\]
		Since \(\mathcal C_j\setminus\mathbb R\) is dense in \(\mathcal C_j\), choose a finite
		\(\delta\)-net \(\{z_1,\dots,z_N\}\subset \mathcal C_j\setminus\mathbb R\). Then, on
		\(E_n\cap\{\|\bar{\bm z}\|\le M\}\), for every \(z\in\mathcal C_j\),
		\[
		\left|
		\bar{\bm z}^\top \bm\Sigma^{1/2}\bm Q_n(z)\bm v_j
		\right|
		\le
		\max_{1\le \ell\le N}
		\left|
		\bar{\bm z}^\top \bm\Sigma^{1/2}\bm Q_n(z_\ell)\bm v_j
		\right|
		+
		C_jM\delta.
		\]
		Therefore
		\begin{align*}
			\mathbb P\!\left(
			\sup_{z\in\mathcal C_j}
			\left|
			\bar{\bm z}^\top \bm\Sigma^{1/2}\bm Q_n(z)\bm v_j
			\right|>2\varepsilon
			\right)
			&\le
			\mathbb P(E_n^c)
			+
			\mathbb P(\|\bar{\bm z}\|>M)\\
			&\quad
			+
			\sum_{\ell=1}^N
			\mathbb P\!\left(
			\left|
			\bar{\bm z}^\top \bm\Sigma^{1/2}\bm Q_n(z_\ell)\bm v_j
			\right|
			>\varepsilon
			\right)
		\end{align*}
		provided \(\delta>0\) is chosen so small that \(C_jM\delta\le \varepsilon\). Since
		\(\mathbb P(E_n^c)\to 0\), \(\mathbb P(\|\bar{\bm z}\|>M)<\varepsilon\), and each fixed-net
		term tends to zero in probability, we conclude that
		\[
		\sup_{z\in\mathcal C_j}
		\left|
		\bar{\bm z}^\top \bm\Sigma^{1/2}\bm Q_n(z)\bm v_j
		\right|
		\xrightarrow{p}0.
		\]
		This proves \eqref{eq:blockmean_resolvent_uniform}.
		
		Consequently,
		\[
		\bar{\bm z}^\top \bm\Sigma^{1/2}\widehat{\bm P}_j\bm v_j
		=
		-\frac{1}{2\pi i}\oint_{\mathcal C_j}
		\bar{\bm z}^\top \bm\Sigma^{1/2}\bm Q_n(z)\bm v_j\,dz
		\xrightarrow{p}0.
		\]
		Together with the previously established lower bound on
		\(|\widehat{\bm v}_j^\top \bm v_j|\), this yields
		\[
		\bar{\bm z}^\top \bm\Sigma^{1/2}\widehat{\bm v}_j \xrightarrow{p}0.
		\]
		Since \(\bar{\bm z}\) was an arbitrary choice among \(\bar{\bm z}_{k,+}\) and
		\(\bar{\bm z}_{k,-}\), we finally obtain
		\[
		e_k^\top \widehat{\bm v}_j
		=
		\bigl(\bar{\bm z}_{k,+}-\bar{\bm z}_{k,-}\bigr)^\top \bm\Sigma^{1/2}\widehat{\bm v}_j
		\xrightarrow{p}0.
		\]
		This completes the proof.
	\end{proof}
	Next, we formally begin the proof of Theorem~\ref{thm:err_homo}. For \(k=1,\dots,K\), write
	\[
	\eta_{k,\pm1}:=\widehat{\bm\mu}_{\pm1,k}-\bm\mu_{\pm1,k},
	\qquad
	e_k:=\widehat{\bm\mu}_k-\bm\mu_k=\eta_{k,+}-\eta_{k,-},
	\qquad
	s_k:=\eta_{k,+}+\eta_{k,-}.
	\]
	Then
	\[
	\widehat{\bm\mu}_k=\bm\mu_k+e_k,
	\qquad
	2\bm\mu_{+1,K}-\widehat{\bm\mu}_{+1,K}-\widehat{\bm\mu}_{-1,K}
	=
	\bm\mu_K-s_K,
	\]
	and
	\[
	\widehat{\bm\mu}_{+1,K}+\widehat{\bm\mu}_{-1,K}-2\bm\mu_{-1,K}
	=
	\bm\mu_K+s_K.
	\]
	By construction,
	\[
	\widehat{\bm d}(\bm w)=\sum_{k=1}^K w_k\widehat{\bm d}_k,
	\qquad
	\widehat{\bm d}_k=\widehat{\bm\Sigma}^{-1}\widehat{\bm\mu}_k,
	\]
	and the empirical intercept is
	\[
	\widehat b_K
	=
	-\widehat{\bm d}(\bm w)^\top
	\frac{\widehat{\bm\mu}_{+1,K}+\widehat{\bm\mu}_{-1,K}}{2}.
	\]
	
	Define
	\[
	V_n(\bm w):=\widehat{\bm d}(\bm w)^\top \bm\Sigma \widehat{\bm d}(\bm w),
	\]
	and
	\[
	D_{+,n}(\bm w)
	:=
	\bigl(2\bm\mu_{+1,K}-\widehat{\bm\mu}_{+1,K}-\widehat{\bm\mu}_{-1,K}\bigr)^\top
	\widehat{\bm d}(\bm w),
	\]
	\[
	D_{-,n}(\bm w)
	:=
	\bigl(\widehat{\bm\mu}_{+1,K}+\widehat{\bm\mu}_{-1,K}-2\bm\mu_{-1,K}\bigr)^\top
	\widehat{\bm d}(\bm w).
	\]
	Then
	\[
	-\bigl(\widehat{\bm d}(\bm w)^\top\bm\mu_{+1,K}+\widehat b_K\bigr)
	=
	-\frac12 D_{+,n}(\bm w),
	\]
	and
	\[
	\widehat{\bm d}(\bm w)^\top\bm\mu_{-1,K}+\widehat b_K
	=
	-\frac12 D_{-,n}(\bm w).
	\]
	Hence
	\begin{align}\label{eq:err_margin_rep}
		\Err(\bm w)
		&=
		\frac12\Phi\!\left(
		-\frac{D_{-,n}(\bm w)}{2\sqrt{V_n(\bm w)}}
		\right)
		+
		\frac12\Phi\!\left(
		-\frac{D_{+,n}(\bm w)}{2\sqrt{V_n(\bm w)}}
		\right).
	\end{align}
	Therefore it suffices to prove that
	\begin{align}
		D_{+,n}(\bm w)&=\bm u_p^\top\bm w-w_K\Delta_K+o_p(1), \label{eq:Dplus_goal}\\
		D_{-,n}(\bm w)&=\bm u_p^\top\bm w+w_K\Delta_K+o_p(1), \label{eq:Dminus_goal}\\
		V_n(\bm w)&=\bm w^\top \bm A_p \bm w+o_p(1). \label{eq:V_goal}
	\end{align}
	
	We first record several elementary consequences of Assumptions~\ref{ass1}--\ref{ass4}.
	Since
	\[
	\bm\mu_k=\bar{\bm\mu}+\bm\delta_k,
	\qquad
	\mathbb E\|\bm\delta_k\|^2=\alpha_k^2,
	\]
	we have \(\|\bm\mu_k\|=O_p(1)\) for each fixed \(k\). Moreover, for fixed \(k,k'\),
	\begin{gather}
		\bar{\bm\mu}^\top\bm\delta_k=o_p(1),\qquad
		\bar{\bm\mu}^\top e_k=o_p(1),\qquad
		\bm\delta_k^\top e_{k'}=o_p(1),\label{eq:identity_mixed_small}\\
		\bm\mu_K^\top e_k=o_p(1),\qquad
		\bm\mu_K^\top s_K=o_p(1),\qquad
		s_K^\top \bm\mu_k=o_p(1),\label{eq:mu_e_small}\\
		\frac{1}{\sigma^2}s_K^\top e_k
		=
		\Delta_K\ind_{\{k=K\}}+o_p(1),\label{eq:s_e_limit}\\
		\frac{1}{\sigma^2}e_k^\top e_{k'}
		=
		\tau_k\,\ind_{\{k=k'\}}+o_p(1),\label{eq:e_e_limit}\\
		\bm\mu_K^\top\bm\mu_k
		=
		\|\bar{\bm\mu}\|^2+\alpha_K^2\ind_{\{k=K\}}+o_p(1).\label{eq:mu_mu_limit}
	\end{gather}
	Indeed, the above relations follow from standard second-moment calculations, using that each classwise sample-mean noise has covariance
	\[
	\var(\eta_{k,\pm1})=\frac{1}{n_{k,\pm1}}\bm\Sigma.
	\]
	Moreover,
	\[
	s_K^\top e_K
	=
	\|\eta_{K,+}\|^2-\|\eta_{K,-}\|^2,
	\]
	and concentration of each classwise squared norm gives
	\[
	\frac{1}{\sigma^2}s_K^\top e_K
	=
	\frac{p}{n_{K,+}}-\frac{p}{n_{K,-}}+o_p(1)
	=
	\Delta_K+o_p(1).
	\]
	Similarly,
	\[
	e_k=\eta_{k,+}-\eta_{k,-},
	\qquad
	\var(e_k)=\left(\frac{1}{n_{k,+}}+\frac{1}{n_{k,-}}\right)\bm\Sigma,
	\]
	which gives \eqref{eq:e_e_limit}. Finally, \eqref{eq:mu_mu_limit} follows from Assumption~\ref{ass2}.
	
	Next, by \eqref{eq:Sigma_tilde_homo},
	\[
	\widehat{\bm\Sigma}
	=
	\sigma^2\left(
	\bm I_p+\sum_{j\in\mathcal I}\lambda_j \widehat{\bm P}_j
	\right),
	\qquad
	\widehat{\bm P}_j:=\widehat{\bm v}_j\widehat{\bm v}_j^\top.
	\]
	Since the projectors \(\widehat{\bm P}_j\) are mutually orthogonal,
	\[
	\widehat{\bm\Sigma}^{-1}
	=
	\frac1{\sigma^2}
	\left(
	\bm I_p-\sum_{j\in\mathcal I}\frac{\lambda_j}{1+\lambda_j}\widehat{\bm P}_j
	\right).
	\]
	Then
	\begin{equation}\label{eq:dhat_expand}
		\widehat{\bm d}_k
		=
		\frac1{\sigma^2}
		\left(
		\widehat{\bm\mu}_k-\sum_{j\in\mathcal I} b_j \widehat{\bm P}_j \widehat{\bm\mu}_k
		\right).
	\end{equation}
	
	We now analyze the spike-projector terms. By Theorem~\ref{thm:centered_overlap_prob} and
	polarization, for any deterministic vectors \(\bm x,\bm y\) with bounded Euclidean norms,
	\begin{equation}\label{eq:polarization_overlap}
		\bm x^\top \widehat{\bm P}_j \bm y
		=
		a_j\,\bm x^\top \bm v_j\bm v_j^\top \bm y+o_p(1)
		=
		a_j(\bm x^\top \bm v_j)(\bm y^\top \bm v_j)+o_p(1).
	\end{equation}
	In particular,
	\begin{equation}\label{eq:mubar_P_mubar}
		\bar{\bm\mu}^\top \widehat{\bm P}_j \bar{\bm\mu}
		=
		a_j(\bar{\bm\mu}^\top \bm v_j)^2+o_p(1).
	\end{equation}
	Also, since \(\bm\delta_k\) is independent of the training data and
	\[
	\mathbb E\bigl((\bm\delta_k^\top \widehat{\bm v}_j)^2\mid \widehat{\bm v}_j\bigr)
	=
	\frac{\alpha_k^2}{p},
	\]
	we have
	\begin{equation}\label{eq:delta_vhat_small}
		\bm\delta_k^\top \widehat{\bm v}_j \xrightarrow{p}0,
		\qquad
		\bm\delta_k^\top \widehat{\bm P}_j \bm\delta_{k'}=o_p(1),
		\qquad
		\bar{\bm\mu}^\top \widehat{\bm P}_j \bm\delta_k=o_p(1).
	\end{equation}
	Furthermore, by Lemma~\ref{lem:ek_vhat_zero},
	\begin{equation}\label{eq:e_vhat_small}
		e_k^\top \widehat{\bm v}_j\xrightarrow{p}0,
		\qquad
		e_k^\top \widehat{\bm P}_j e_{k'}=o_p(1),
		\qquad
		\bm\mu_K^\top \widehat{\bm P}_j e_k=o_p(1).
	\end{equation}
	Exactly the same proof as in Lemma~\ref{lem:ek_vhat_zero}, replacing
	\[
	e_k=\bm\Sigma^{1/2}(\bar{\bm z}_{k,+}-\bar{\bm z}_{k,-})
	\]
	by
	\[
	s_k=\bm\Sigma^{1/2}(\bar{\bm z}_{k,+}+\bar{\bm z}_{k,-}),
	\]
	shows that
	\begin{equation}\label{eq:s_vhat_small}
		s_k^\top \widehat{\bm v}_j\xrightarrow{p}0,
		\qquad
		s_k^\top \widehat{\bm P}_j \bm\mu_k=o_p(1),
		\qquad
		s_k^\top \widehat{\bm P}_j e_k=o_p(1).
	\end{equation}
	Combining \eqref{eq:mubar_P_mubar}, \eqref{eq:delta_vhat_small}, \eqref{eq:e_vhat_small},
	and \eqref{eq:s_vhat_small}, we obtain, for each fixed \(k\) and \(j\),
	\begin{equation}\label{eq:muKminusS_P_mukhat}
		(\bm\mu_K-s_K)^\top \widehat{\bm P}_j \widehat{\bm\mu}_k
		=
		a_j(\bar{\bm\mu}^\top \bm v_j)^2+o_p(1),
	\end{equation}
	and likewise
	\begin{equation}\label{eq:muKplusS_P_mukhat}
		(\bm\mu_K+s_K)^\top \widehat{\bm P}_j \widehat{\bm\mu}_k
		=
		a_j(\bar{\bm\mu}^\top \bm v_j)^2+o_p(1).
	\end{equation}
	
	We now prove \eqref{eq:Dplus_goal}. By \eqref{eq:dhat_expand},
	\begin{align*}
		D_{+,n}(\bm w)
		&=
		(\bm\mu_K-s_K)^\top \widehat{\bm d}(\bm w)\\
		&=
		\frac1{\sigma^2}\sum_{k=1}^K w_k
		\left[
		(\bm\mu_K-s_K)^\top \widehat{\bm\mu}_k
		-
		\sum_{j\in\mathcal I}
		b_j(\bm\mu_K-s_K)^\top \widehat{\bm P}_j \widehat{\bm\mu}_k
		\right].
	\end{align*}
	By \eqref{eq:mu_mu_limit}, \eqref{eq:mu_e_small}, and \eqref{eq:s_e_limit},
	\[
	\frac1{\sigma^2}(\bm\mu_K-s_K)^\top \widehat{\bm\mu}_k
	=
	\frac{\|\bar{\bm\mu}\|^2}{\sigma^2}
	+
	\frac{\alpha_K^2}{\sigma^2}\ind_{\{k=K\}}
	-
	\Delta_K\ind_{\{k=K\}}
	+
	o_p(1).
	\]
	Together with \eqref{eq:muKminusS_P_mukhat}, this yields
	\begin{align*}
		D_{+,n}(\bm w)
		&=
		\sum_{k=1}^K w_k
		\left[
		\frac{\|\bar{\bm\mu}\|^2}{\sigma^2}
		-
		\sum_{j\in\mathcal I}\frac{\lambda_j}{1+\lambda_j}
		a_j\frac{(\bar{\bm\mu}^\top \bm v_j)^2}{\sigma^2}
		+
		\frac{\alpha_K^2}{\sigma^2}\ind_{\{k=K\}}
		-
		\Delta_K\ind_{\{k=K\}}
		\right]
		+
		o_p(1).
	\end{align*}
	By the definition of \(\bm u_p\),
	\[
	D_{+,n}(\bm w)=\bm u_p^\top \bm w-w_K\Delta_K+o_p(1).
	\]
	This proves \eqref{eq:Dplus_goal}.
	Similarly,
	\[
	D_{-,n}(\bm w)
	=
	(\bm\mu_K+s_K)^\top \widehat{\bm d}(\bm w)
	=
	\bm u_p^\top\bm w+w_K\Delta_K+o_p(1),
	\]
	which proves \eqref{eq:Dminus_goal}.
	
	It remains to prove \eqref{eq:V_goal}. Write
	\[
	V_n(\bm w)
	=
	\sum_{k,k'=1}^K w_k w_{k'}
	\widehat{\bm\mu}_k^\top
	\widehat{\bm\Sigma}^{-1}\bm\Sigma\widehat{\bm\Sigma}^{-1}
	\widehat{\bm\mu}_{k'}.
	\]
	Thus it suffices to show that, for each fixed \(k,k'\),
	\begin{equation}\label{eq:Tkk_goal}
		\widehat{\bm\mu}_k^\top
		\widehat{\bm\Sigma}^{-1}\bm\Sigma\widehat{\bm\Sigma}^{-1}
		\widehat{\bm\mu}_{k'}
		=
		A_{kk',p}+o_p(1).
	\end{equation}
	
	Using
	\[
	\bm\Sigma
	=
	\sigma^2\left(
	\bm I_p+\sum_{\ell\in\mathcal I}\lambda_\ell \bm P_\ell
	\right),
	\qquad
	\bm P_\ell:=\bm v_\ell\bm v_\ell^\top,
	\]
	and
	\[
	\widehat{\bm\Sigma}^{-1}
	=
	\frac1{\sigma^2}\left(
	\bm I_p-\sum_{j\in\mathcal I} b_j \widehat{\bm P}_j
	\right),
	\]
	we obtain
	\begin{align}
		\widehat{\bm\Sigma}^{-1}\bm\Sigma\widehat{\bm\Sigma}^{-1}
		=
		\frac1{\sigma^2}
		\left(
		\bm I_p-\sum_{j\in\mathcal I} b_j \widehat{\bm P}_j
		\right)
		\left(
		\bm I_p+\sum_{\ell\in\mathcal I}\lambda_\ell \bm P_\ell
		\right)
		\left(
		\bm I_p-\sum_{m\in\mathcal I} b_m \widehat{\bm P}_m
		\right).
		\label{eq:Sigma_inv_Sigma_Sigma_inv_expand}
	\end{align}
	Now decompose
	\[
	\widehat{\bm\mu}_k=\bar{\bm\mu}+\bm\delta_k+e_k.
	\]
	Since \(\bm P_\ell=\bm v_\ell\bm v_\ell^\top\) is a rank-one deterministic projector, for each fixed
	\(\ell\in\mathcal I\) we also have
	\begin{equation}\label{eq:det_proj_small}
		\bm\delta_k^\top \bm v_\ell \xrightarrow{p}0,
		\qquad
		e_k^\top \bm v_\ell \xrightarrow{p}0,
	\end{equation}
	and hence
	\begin{equation}\label{eq:det_proj_mixed_small}
		\bm\delta_k^\top \bm P_\ell \bm\delta_{k'}=o_p(1),\quad
		\bar{\bm\mu}^\top \bm P_\ell \bm\delta_k=o_p(1),\quad
		e_k^\top \bm P_\ell e_{k'}=o_p(1),\quad
		\bar{\bm\mu}^\top \bm P_\ell e_k=o_p(1),
	\end{equation}
	as well as the analogous mixed terms with one \(\bm\delta\)-factor and one \(e\)-factor. Indeed,
	each term factors through \(\bm v_\ell\), for example
	\[
	\bm\delta_k^\top \bm P_\ell \bm\delta_{k'}
	=
	(\bm\delta_k^\top \bm v_\ell)(\bm v_\ell^\top \bm\delta_{k'}),
	\qquad
	e_k^\top \bm P_\ell e_{k'}
	=
	(e_k^\top \bm v_\ell)(\bm v_\ell^\top e_{k'}).
	\]
	Therefore, by \eqref{eq:identity_mixed_small}, \eqref{eq:delta_vhat_small}, \eqref{eq:e_vhat_small},
	\eqref{eq:det_proj_small}, \eqref{eq:det_proj_mixed_small}, and the boundedness of
	\(\|\widehat{\bm\Sigma}^{-1}\bm\Sigma\widehat{\bm\Sigma}^{-1}\|\), all mixed terms involving at
	least one factor among the \(\bm\delta_k\)- or \(e_k\)-components are \(o_p(1)\), except for the
	pure identity contribution \(\sigma^{-2}\bm\delta_k^\top\bm\delta_{k'}\) and
	\(\sigma^{-2}e_k^\top e_{k'}\).
	Consequently,
	\begin{align*}
		\widehat{\bm\mu}_k^\top
		\widehat{\bm\Sigma}^{-1}\bm\Sigma\widehat{\bm\Sigma}^{-1}
		\widehat{\bm\mu}_{k'}
		&=
		\bar{\bm\mu}^\top
		\widehat{\bm\Sigma}^{-1}\bm\Sigma\widehat{\bm\Sigma}^{-1}
		\bar{\bm\mu}
		+
		\frac{1}{\sigma^2}\bm\delta_k^\top\bm\delta_{k'}
		+
		\frac{1}{\sigma^2}e_k^\top e_{k'}
		+
		o_p(1).
	\end{align*}
	By Assumption~\ref{ass2},
	\[
	\frac{1}{\sigma^2}\bm\delta_k^\top\bm\delta_{k'}
	=
	\frac{\alpha_k^2}{\sigma^2}\ind_{\{k=k'\}}+o_p(1),
	\]
	and by \eqref{eq:e_e_limit},
	\[
	\frac{1}{\sigma^2}e_k^\top e_{k'}
	=
	\tau_k\ind_{\{k=k'\}}+o_p(1).
	\]
	Thus it remains to evaluate
	\[
	\bar{\bm\mu}^\top
	\widehat{\bm\Sigma}^{-1}\bm\Sigma\widehat{\bm\Sigma}^{-1}
	\bar{\bm\mu}.
	\]
	
	Expanding \eqref{eq:Sigma_inv_Sigma_Sigma_inv_expand} between \(\bar{\bm\mu}\) and itself and
	using \(\widehat{\bm P}_j\widehat{\bm P}_m=\mathbf 1_{\{j=m\}}\widehat{\bm P}_j\), we obtain
	\begin{align*}
		\bar{\bm\mu}^\top
		\widehat{\bm\Sigma}^{-1}\bm\Sigma\widehat{\bm\Sigma}^{-1}
		\bar{\bm\mu}
		&=
		\frac{\|\bar{\bm\mu}\|^2}{\sigma^2}
		+
		\frac{1}{\sigma^2}
		\sum_{j\in\mathcal I}
		\Bigl(b_j^2-2b_j\Bigr)\,
		\bar{\bm\mu}^\top \widehat{\bm P}_j \bar{\bm\mu}\\
		&\quad
		+
		\frac{1}{\sigma^2}
		\sum_{\ell\in\mathcal I}\lambda_\ell\,
		\bar{\bm\mu}^\top \bm P_\ell \bar{\bm\mu}\\
		&\quad
		-
		\frac{1}{\sigma^2}
		\sum_{\ell,m\in\mathcal I}\lambda_\ell b_m\,
		\bar{\bm\mu}^\top \bm P_\ell \widehat{\bm P}_m \bar{\bm\mu}\\
		&\quad
		-
		\frac{1}{\sigma^2}
		\sum_{j,\ell\in\mathcal I}\lambda_\ell b_j\,
		\bar{\bm\mu}^\top \widehat{\bm P}_j \bm P_\ell \bar{\bm\mu}\\
		&\quad
		+
		\frac{1}{\sigma^2}
		\sum_{j,\ell,m\in\mathcal I}\lambda_\ell b_j b_m\,
		\bar{\bm\mu}^\top \widehat{\bm P}_j \bm P_\ell \widehat{\bm P}_m \bar{\bm\mu}.
	\end{align*}
	By Theorem~\ref{thm:centered_overlap_prob} and polarization,
	\[
	\bar{\bm\mu}^\top \widehat{\bm P}_j \bar{\bm\mu}
	=
	a_j(\bar{\bm\mu}^\top \bm v_j)^2+o_p(1),
	\]
	\[
	\bar{\bm\mu}^\top \bm P_\ell \widehat{\bm P}_m \bar{\bm\mu}
	=
	a_m\,\ind_{\{\ell=m\}}(\bar{\bm\mu}^\top \bm v_m)^2+o_p(1),
	\]
	\[
	\bar{\bm\mu}^\top \widehat{\bm P}_j \bm P_\ell \bar{\bm\mu}
	=
	a_j\,\ind_{\{j=\ell\}}(\bar{\bm\mu}^\top \bm v_j)^2+o_p(1),
	\]
	and
	\[
	\bar{\bm\mu}^\top \widehat{\bm P}_j \bm P_\ell \widehat{\bm P}_m \bar{\bm\mu}
	=
	a_j^2\,\ind_{\{j=\ell=m\}}(\bar{\bm\mu}^\top \bm v_j)^2+o_p(1).
	\]
	Therefore
	\[
	\bar{\bm\mu}^\top
	\widehat{\bm\Sigma}^{-1}\bm\Sigma\widehat{\bm\Sigma}^{-1}
	\bar{\bm\mu}
	=
	\frac{\|\bar{\bm\mu}\|^2}{\sigma^2}
	+
	\frac{1}{\sigma^2}
	\sum_{j\in\mathcal I}
	\left[
	(b_j^2-2b_j)a_j+\lambda_j-2\lambda_j b_j a_j+\lambda_j b_j^2 a_j^2
	\right]
	(\bar{\bm\mu}^\top \bm v_j)^2
	+
	o_p(1).
	\]
	By the definition of \(b_j\) and \(c_j\),
	\[
	c_j=(b_j^2-2b_j)a_j+\lambda_j-2\lambda_j b_j a_j+\lambda_j b_j^2 a_j^2.
	\]
	Hence
	\[
	\bar{\bm\mu}^\top
	\widehat{\bm\Sigma}^{-1}\bm\Sigma\widehat{\bm\Sigma}^{-1}
	\bar{\bm\mu}
	=
	\frac{\|\bar{\bm\mu}\|^2}{\sigma^2}
	+
	\sum_{j\in\mathcal I}
	c_j\frac{(\bar{\bm\mu}^\top \bm v_j)^2}{\sigma^2}
	+
	o_p(1).
	\]
	Combining the previous displays, we conclude that
	\[
	\widehat{\bm\mu}_k^\top
	\widehat{\bm\Sigma}^{-1}\bm\Sigma\widehat{\bm\Sigma}^{-1}
	\widehat{\bm\mu}_{k'}
	=
	\frac{\|\bar{\bm\mu}\|^2}{\sigma^2}
	+
	\sum_{j\in\mathcal I}
	c_j\frac{(\bar{\bm\mu}^\top \bm v_j)^2}{\sigma^2}
	+
	\left(
	\frac{\alpha_k^2}{\sigma^2}+\tau_k
	\right)\ind_{\{k=k'\}}
	+
	o_p(1).
	\]
	By the definition of \(A_{kk',p}\), this proves \eqref{eq:Tkk_goal}, and hence
	\eqref{eq:V_goal}.
	
	Finally, \eqref{eq:Dplus_goal}, \eqref{eq:Dminus_goal}, and \eqref{eq:V_goal} imply
	\[
	\frac{D_{+,n}(\bm w)}{2\sqrt{V_n(\bm w)}}
	=
	\frac{\bm u_p^\top \bm w-w_K\Delta_K}{2\sqrt{\bm w^\top \bm A_p \bm w}}+o_p(1),
	\qquad
	\frac{D_{-,n}(\bm w)}{2\sqrt{V_n(\bm w)}}
	=
	\frac{\bm u_p^\top \bm w+w_K\Delta_K}{2\sqrt{\bm w^\top \bm A_p \bm w}}+o_p(1),
	\]
	because \(\bm w\neq\bm0\) and \(\bm A_p\) is positive definite by
	Lemma~\ref{lem:Ap_positive}. Substituting the above into
	\eqref{eq:err_margin_rep} and using the continuity of \(\Phi\), we obtain
	\eqref{eq:err_limit_uncorrected_homo}.
	This completes the proof.
	\subsection{Proof of Proposition \ref{prop:optimal_intercept_homo}}
	For a fixed scalar adjustment \(t\), the preceding proof remains unchanged except that the two margin terms are replaced by
	\[
	D_{+,n}^{(t)}(\bm w)=D_{+,n}(\bm w)+2t,
	\qquad
	D_{-,n}^{(t)}(\bm w)=D_{-,n}(\bm w)-2t.
	\]
	Thus
	\[
	\Err_t(\bm w)
	=
	\frac12\Phi\left(
	-\frac{D_{-,n}(\bm w)-2t}{2\sqrt{V_n(\bm w)}}
	\right)
	+
	\frac12\Phi\left(
	-\frac{D_{+,n}(\bm w)+2t}{2\sqrt{V_n(\bm w)}}
	\right).
	\]
	On the event \(D_{+,n}(\bm w)+D_{-,n}(\bm w)>0\), differentiating the last display with respect to \(t\) shows that its unique minimizer is
	\[
	\widehat t_p^\ast(\bm w)=\frac14\{D_{-,n}(\bm w)-D_{+,n}(\bm w)\}.
	\]
	Since \(\bm u_p^\top\bm w>0\), \eqref{eq:Dplus_goal} and \eqref{eq:Dminus_goal} imply
	\[
	D_{+,n}(\bm w)+D_{-,n}(\bm w)
	=2\bm u_p^\top\bm w+o_p(1)>0
	\]
	with probability tending to one. Moreover,
	\[
	\widehat t_p^\ast(\bm w)
	=\frac12w_K\Delta_K+o_p(1)
	=\frac12w_K\widehat\Delta_K+o_p(1),
	\]
	because \(\widehat\Delta_K\to\Delta_K\). Substituting this minimizer into the margin representation and using \eqref{eq:Dplus_goal}, \eqref{eq:Dminus_goal}, and \eqref{eq:V_goal} yields
	\[
	\Err_{\widehat t_p^\ast(\bm w)}(\bm w)
	-
	\Phi\left(
	-\frac{\bm u_p^\top\bm w}{2\sqrt{\bm w^\top\bm A_p\bm w}}
	\right)
	\to0
	\]
	in probability. This completes the proof.
	
	\subsection{Proof of Corollary \ref{cor:oracle_weight_homo}}
	By Proposition~\ref{prop:optimal_intercept_homo}, the deterministic equivalent of the target-domain
	Gaussian-calibrated error under the optimal intercept is
	\[
	\Phi\left(
	-\frac{\bm u_p^\top \bm w}{2\sqrt{\bm w^\top \bm A_p\bm w}}
	\right).
	\]
	Since \(\Phi\) is strictly increasing, minimizing this quantity is equivalent to maximizing
	\[
	\frac{\bm u_p^\top \bm w}{\sqrt{\bm w^\top \bm A_p\bm w}}
	\]
	over all \(\bm w\neq \bm 0\). This criterion is invariant under positive rescaling of
	\(\bm w\). Therefore the oracle weight is identifiable only up to a positive scalar multiple.
	
	By Lemma~\ref{lem:Ap_positive}, \(\bm A_p\) is positive definite. Therefore,
	the Cauchy--Schwarz inequality under the \(\bm A_p\)-inner product gives
	\[
	\bm u_p^\top \bm w
	=
	(\bm A_p^{-1/2}\bm u_p)^\top(\bm A_p^{1/2}\bm w)
	\le
	\sqrt{\bm u_p^\top \bm A_p^{-1}\bm u_p}\,
	\sqrt{\bm w^\top \bm A_p\bm w}.
	\]
	Hence
	\[
	\frac{\bm u_p^\top \bm w}{\sqrt{\bm w^\top \bm A_p\bm w}}
	\le
	\sqrt{\bm u_p^\top \bm A_p^{-1}\bm u_p}.
	\]
	Equality holds if and only if
	\[
	\bm A_p^{1/2}\bm w
	\propto
	\bm A_p^{-1/2}\bm u_p,
	\]
	or equivalently,
	\[
	\bm w\propto \bm A_p^{-1}\bm u_p.
	\]
	Therefore the deterministic equivalent is minimized by
	\[
	\bm w_p^\ast\propto \bm A_p^{-1}\bm u_p.
	\]
	
	Under the normalization \(\bm w^\top\bm A_p\bm w=1\), the maximizer is uniquely determined as
	\[
	\bm w_p^\ast
	=
	\frac{\bm A_p^{-1}\bm u_p}
	{\sqrt{\bm u_p^\top\bm A_p^{-1}\bm u_p}}.
	\]
	This completes the proof.
	\subsection{Proof of Theorem \ref{thm:emp_weight_homo}}
	We first prove the consistency of the quantities used in
	\(\widehat{\bm u}_p\) and \(\widehat{\bm A}_p\). Recall that
	\[
	\widehat{\bm\mu}_k=\bm\mu_k+e_k
	=
	\bar{\bm\mu}+\bm\delta_k+e_k.
	\]
	From the calculations in the proof of Theorem~\ref{thm:err_homo}, for any fixed
	\(k\neq \ell\),
	\[
	\widehat{\bm\mu}_k^\top\widehat{\bm\mu}_\ell
	=
	\|\bar{\bm\mu}\|^2+o_p(1).
	\]
	Hence, since \(K\) is fixed,
	\[
	\widehat\kappa_p
	=
	\frac{1}{\sigma^2K(K-1)}
	\sum_{m\neq \ell}\widehat{\bm\mu}_m^\top\widehat{\bm\mu}_\ell
	=
	\frac{\|\bar{\bm\mu}\|^2}{\sigma^2}+o_p(1).
	\]
	
	Next, for each fixed \(k\),
	\[
	\|\widehat{\bm\mu}_k\|^2
	=
	\|\bar{\bm\mu}\|^2+\|\bm\delta_k\|^2+e_k^\top e_k+o_p(1).
	\]
	By Assumption~\ref{ass2},
	\[
	\|\bm\delta_k\|^2=\alpha_k^2+o_p(1),
	\]
	and by the same argument used for \eqref{eq:e_e_limit},
	\[
	e_k^\top e_k
	=
	\sigma^2\left(\frac{p}{n_{k,+}}+\frac{p}{n_{k,-}}\right)+o_p(1)
	=
	\sigma^2\tau_k+o_p(1).
	\]
	Moreover, for \(\ell\neq k\),
	\[
	\widehat{\bm\mu}_k^\top\widehat{\bm\mu}_\ell
	=
	\|\bar{\bm\mu}\|^2+o_p(1).
	\]
	Therefore
	\[
	\widehat\beta_k
	=
	\frac1{\sigma^2}
	\left(
	\|\widehat{\bm\mu}_k\|^2
	-
	\frac1{K-1}\sum_{\ell\neq k}
	\widehat{\bm\mu}_k^\top\widehat{\bm\mu}_\ell
	\right)-
	\widehat\tau_k
	=
	\frac{\alpha_k^2}{\sigma^2}+o_p(1).
	\]
	
	We now consider the spike-projection terms. For any fixed \(m\neq \ell\) and
	\(j\in\mathcal I\), write
	\[
	\widehat{\bm\mu}_\ell^\top
	\widehat{\bm P}_j
	\widehat{\bm\mu}_m
	=
	(\bar{\bm\mu}+\bm\delta_\ell+e_\ell)^\top
	\widehat{\bm P}_j
	(\bar{\bm\mu}+\bm\delta_m+e_m).
	\]
	By the same argument as in \eqref{eq:delta_vhat_small} and Lemma~\ref{lem:ek_vhat_zero},
	all terms involving at least one factor among
	\(\bm\delta_\ell,\bm\delta_m,e_\ell,e_m\) are \(o_p(1)\). Hence
	\[
	\widehat{\bm\mu}_\ell^\top
	\widehat{\bm P}_j
	\widehat{\bm\mu}_m
	=
	\bar{\bm\mu}^\top
	\widehat{\bm P}_j
	\bar{\bm\mu}
	+o_p(1).
	\]
	By Theorem~\ref{thm:centered_overlap_prob} and polarization,
	\[
	\bar{\bm\mu}^\top
	\widehat{\bm P}_j
	\bar{\bm\mu}
	=
	a_j(\bar{\bm\mu}^\top\bm v_j)^2+o_p(1).
	\]
	Therefore
	\[
	\widehat{\bm\mu}_\ell^\top
	\widehat{\bm v}_j\widehat{\bm v}_j^\top
	\widehat{\bm\mu}_m
	=
	a_j(\bar{\bm\mu}^\top\bm v_j)^2+o_p(1).
	\]
	
	Since \(K\) and \(|\mathcal I|\) are fixed, this convergence holds after finite summation.
	Consequently,
	\[
	\widehat u_{k,p}
	=
	\frac{\|\bar{\bm\mu}\|^2}{\sigma^2}
	-
	\sum_{j\in\mathcal I}
	b_ja_j\frac{(\bar{\bm\mu}^\top\bm v_j)^2}{\sigma^2}
	+
	\frac{\alpha_K^2}{\sigma^2}\ind_{\{k=K\}}
	+o_p(1)
	=
	u_{k,p}+o_p(1).
	\]
	Thus
	\[
	\widehat{\bm u}_p-\bm u_p\to\bm0
	\qquad\text{in probability}.
	\]
	
	Similarly, since \(\widehat a_j\to a_j\), using the factor \(1/\widehat a_j\) in the definition of \(\widehat A_{kk',p}\),
	\[
	\frac1{K(K-1)}
	\sum_{m\neq \ell}
	\frac1{\widehat a_j}
	\widehat{\bm\mu}_\ell^\top
	\widehat{\bm v}_j\widehat{\bm v}_j^\top
	\widehat{\bm\mu}_m
	=
	(\bar{\bm\mu}^\top\bm v_j)^2+o_p(1).
	\]
	Since \(\widehat c_j\to c_j\), \(\widehat\tau_k\to\tau_k\), together with the consistency of \(\widehat\kappa_p\) and \(\widehat\beta_k\), this gives
	\[
	\widehat A_{kk',p}
	=
	\frac{\|\bar{\bm\mu}\|^2}{\sigma^2}
	+
	\sum_{j\in\mathcal I}
	c_j\frac{(\bar{\bm\mu}^\top\bm v_j)^2}{\sigma^2}
	+
	\left(
	\frac{\alpha_k^2}{\sigma^2}+\tau_k
	\right)\ind_{\{k=k'\}}
	+o_p(1)
	=
	A_{kk',p}+o_p(1).
	\]
	Because \(K\) is fixed, entrywise convergence implies
	\[
	\widehat{\bm A}_p-\bm A_p\to \bm0
	\qquad\text{in probability}.
	\]
	
	It remains to prove the consistency of the empirical weight. By Lemma~\ref{lem:Ap_positive}, there exists a constant \(c_A>0\) such that
	\[
	\lambda_{\min}(\bm A_p)\ge c_A
	\]
	for all sufficiently large \(p\). Since
	\(\|\widehat{\bm A}_p-\bm A_p\|\to0\) in probability, Weyl's inequality gives
	\[
	\mathbb P\bigl(\lambda_{\min}(\widehat{\bm A}_p)\ge c_A/2\bigr)\to1.
	\]
	Thus \(\widehat{\bm A}_p\) is positive definite with probability tending to one. Moreover,
	on this event,
	\[
	\widehat{\bm A}_p^{-1}-\bm A_p^{-1}
	=
	\widehat{\bm A}_p^{-1}(\bm A_p-\widehat{\bm A}_p)\bm A_p^{-1},
	\]
	and therefore
	\[
	\widehat{\bm A}_p^{-1}-\bm A_p^{-1}\to \bm0
	\qquad\text{in probability}.
	\]
	Hence
	\[
	\widehat{\bm A}_p^{-1}\widehat{\bm u}_p
	-
	\bm A_p^{-1}\bm u_p
	\to \bm0
	\qquad\text{in probability}.
	\]
	Also,
	\[
	\widehat{\bm u}_p^\top
	\widehat{\bm A}_p^{-1}
	\widehat{\bm u}_p
	-
	\bm u_p^\top\bm A_p^{-1}\bm u_p
	\to0
	\qquad\text{in probability}.
	\]
	Since \(\|\bm u_p\|\ge c_u\) and \(\lambda_{\max}(\bm A_p)=O(1)\), there exists
	a constant \(c'>0\) such that
	\[
	\bm u_p^\top\bm A_p^{-1}\bm u_p
	\ge c' .
	\]
	Thus the denominator in \(\widehat{\bm w}_p^\ast\) is bounded away from zero with probability
	tending to one. Therefore, by the continuous mapping theorem,
	\[
	\widehat{\bm w}_p^\ast
	=
	\frac{\widehat{\bm A}_p^{-1}\widehat{\bm u}_p}
	{\sqrt{\widehat{\bm u}_p^\top\widehat{\bm A}_p^{-1}\widehat{\bm u}_p}}
	\to
	\frac{\bm A_p^{-1}\bm u_p}
	{\sqrt{\bm u_p^\top\bm A_p^{-1}\bm u_p}}
	=
	\bm w_p^\ast
	\]
	in probability. This completes the proof.
	\subsection{Proof of Theorem \ref{thm:err_hetero}}
	\begin{lemma}\label{lem:hetero_polarization_overlap}
		For each fixed domain \(k\in\{1,\dots,K\}\), let
		\[
		\widehat{\bm P}_{j,k}:=\widehat{\bm v}_{j,k}\widehat{\bm v}_{j,k}^\top,
		\qquad
		\bm P_{j,k}:=\bm v_{j,k}\bm v_{j,k}^\top,
		\qquad j\in\mathcal I_k.
		\]
		Under Assumptions~\ref{ass1}, \ref{ass2}, \ref{ass4}, \ref{ass3_hetero} and \ref{ass5_hetero}, for any deterministic vectors \(\bm x,\bm y\in\mathbb R^p\) with bounded Euclidean norms,
		\[
		\bm x^\top \widehat{\bm P}_{j,k}\bm y
		=
		a_{j,k}\,\bm x^\top \bm P_{j,k}\bm y+o_p(1)
		=
		a_{j,k}(\bm x^\top \bm v_{j,k})(\bm y^\top \bm v_{j,k})+o_p(1),
		\qquad j\in\mathcal I_k.
		\]
	\end{lemma}
	
	\begin{proof}
		For each fixed domain \(k\), the proof of Theorem~\ref{thm:centered_overlap_prob} can be repeated verbatim
		for the within-domain sample covariance matrix $\bm S_{n,k}$, with
		$(\bm\Sigma,\gamma,\lambda_j,\bm v_j,\widehat{\bm v}_j)$ replaced by
		$(\bm\Sigma_k,\gamma_k,\lambda_{j,k},\bm v_{j,k},\widehat{\bm v}_{j,k})$.
		Hence, for every deterministic unit vector $\bm\xi\in\mathbb R^p$,
		\[
		\bm\xi^\top \widehat{\bm P}_{j,k}\bm\xi
		=
		\bm\xi^\top \widehat{\bm v}_{j,k}\widehat{\bm v}_{j,k}^\top \bm\xi
		\xrightarrow{p}
		a_{j,k}\,\bm\xi^\top \bm v_{j,k}\bm v_{j,k}^\top \bm\xi.
		\]
		Now apply polarization:
		\[
		\bm x^\top \widehat{\bm P}_{j,k}\bm y
		=
		\frac14\Bigl[
		(\bm x+\bm y)^\top \widehat{\bm P}_{j,k}(\bm x+\bm y)
		-
		(\bm x-\bm y)^\top \widehat{\bm P}_{j,k}(\bm x-\bm y)
		\Bigr].
		\]
		The same identity holds with $\widehat{\bm P}_{j,k}$ replaced by $\bm P_{j,k}$. Since
		$\|\bm x\pm1\bm y\|=O(1)$, the desired conclusion follows.
	\end{proof}
	\begin{lemma}\label{lem:hetero_mean_noise_vhat}
		For each domain \(k=1,\dots,K\), define
		\[
		e_k:=\widehat{\bm\mu}_k-\bm\mu_k,
		\qquad
		s_k:=
		(\widehat{\bm\mu}_{+1,k}-\bm\mu_{+1,k})
		+
		(\widehat{\bm\mu}_{-1,k}-\bm\mu_{-1,k}).
		\]
		Under Assumptions~\ref{ass1}, \ref{ass2}, \ref{ass4}, \ref{ass3_hetero} and \ref{ass5_hetero}, for every fixed \(k\in\{1,\dots,K\}\) and every fixed \(j\in\mathcal I_k\),
		\[
		e_k^\top \widehat{\bm v}_{j,k}\xrightarrow{p}0,
		\qquad
		s_k^\top \widehat{\bm v}_{j,k}\xrightarrow{p}0.
		\]
		Consequently,
		\[
		e_k^\top \widehat{\bm P}_{j,k}e_k=o_p(1),
		\qquad
		s_k^\top \widehat{\bm P}_{j,k}s_k=o_p(1),
		\qquad
		s_k^\top \widehat{\bm P}_{j,k}e_k=o_p(1).
		\]
	\end{lemma}
	
	\begin{proof}
		The proof is exactly the same as that of Lemma~\ref{lem:ek_vhat_zero}, applied within
		domain \(k\). Indeed, under Assumption~\ref{ass1},
		\[
		e_k=\bm\Sigma_k^{1/2}(\bar{\bm z}_{k,+}-\bar{\bm z}_{k,-}),
		\qquad
		s_k=\bm\Sigma_k^{1/2}(\bar{\bm z}_{k,+}+\bar{\bm z}_{k,-}),
		\]
		where \(\bar{\bm z}_{k,\pm1}\) are the whitened sample means in domain \(k\). Repeating the contour-projector argument in Lemma~\ref{lem:ek_vhat_zero}
		for the within-domain sample covariance matrix \(\bm S_{n,k}\)
		we obtain
		\[
		e_k^\top \widehat{\bm v}_{j,k}\xrightarrow{p}0,
		\qquad
		s_k^\top \widehat{\bm v}_{j,k}\xrightarrow{p}0.
		\]
		The projector conclusions follow immediately from
		\[
		e_k^\top \widehat{\bm P}_{j,k}e_k=(e_k^\top \widehat{\bm v}_{j,k})^2,
		\quad
		s_k^\top \widehat{\bm P}_{j,k}s_k=(s_k^\top \widehat{\bm v}_{j,k})^2,
		\quad s_k^\top \widehat{\bm P}_{j,k}e_k
		=
		(s_k^\top \widehat{\bm v}_{j,k})(e_k^\top \widehat{\bm v}_{j,k}).
		\]
	\end{proof}
	\begin{lemma}\label{lem:hetero_mixed_projectors}
		Under Assumptions~\ref{ass1}, \ref{ass2}, \ref{ass4}, \ref{ass3_hetero} and \ref{ass5_hetero}, let
		\[
		\widehat{\bm P}_{j,k}:=\widehat{\bm v}_{j,k}\widehat{\bm v}_{j,k}^\top,
		\quad
		\bm P_{j,k}:=\bm v_{j,k}\bm v_{j,k}^\top,
		\quad
		m_{j,k}:=\bar{\bm\mu}^\top \bm v_{j,k},
		\quad
		\rho_{j\ell}^{(k,k')}:=\bm v_{j,k}^\top \bm v_{\ell,k'}.
		\]
		Then, for any fixed
		\(k,k'\in\{1,\dots,K\}\), \(j\in\mathcal I_k\), \(m\in\mathcal I_{k'}\), and
		\(\ell\in\mathcal I_K\),
		\begin{align*}
			\bar{\bm\mu}^\top \widehat{\bm P}_{j,k}\bar{\bm\mu}
			&=
			a_{j,k}\,m_{j,k}^2+o_p(1),\\
			\bar{\bm\mu}^\top \widehat{\bm P}_{j,k}\bm P_{\ell,K}\bar{\bm\mu}
			&=
			a_{j,k}\,m_{j,k}\rho_{j\ell}^{(k,K)}m_{\ell,K}+o_p(1),\\
			\bar{\bm\mu}^\top \bm P_{\ell,K}\widehat{\bm P}_{m,k'}\bar{\bm\mu}
			&=
			a_{m,k'}\,m_{\ell,K}\rho_{\ell m}^{(K,k')}m_{m,k'}+o_p(1),\\
			\bar{\bm\mu}^\top \widehat{\bm P}_{j,k}\bm P_{\ell,K}\widehat{\bm P}_{m,k'}\bar{\bm\mu}
			&=
			a_{j,k}a_{m,k'}\,
			m_{j,k}\rho_{j\ell}^{(k,K)}\rho_{\ell m}^{(K,k')}m_{m,k'}+o_p(1).
		\end{align*}
		Moreover,
		\[
		\bar{\bm\mu}^\top \widehat{\bm P}_{j,k}\widehat{\bm P}_{m,k'}\bar{\bm\mu}
		=
		\begin{cases}
			a_{j,k}\,\ind_{\{j=m\}}\,m_{j,k}^2+o_p(1), & k=k',\\[1mm]
			a_{j,k}a_{m,k'}\,m_{j,k}\rho_{jm}^{(k,k')}m_{m,k'}+o_p(1), & k\neq k'.
		\end{cases}
		\]
	\end{lemma}
	
	\begin{proof}
		The first display is Lemma~\ref{lem:hetero_polarization_overlap} with
		\(\bm x=\bm y=\bar{\bm\mu}\). For the second display, apply Lemma~\ref{lem:hetero_polarization_overlap} with
		\[
		\bm x=\bar{\bm\mu},
		\qquad
		\bm y=\bm P_{\ell,K}\bar{\bm\mu}
		=
		(\bar{\bm\mu}^\top \bm v_{\ell,K})\bm v_{\ell,K}
		=
		m_{\ell,K}\bm v_{\ell,K}.
		\]
		Then
		\[
		\bar{\bm\mu}^\top \widehat{\bm P}_{j,k}\bm P_{\ell,K}\bar{\bm\mu}
		=
		a_{j,k}(\bar{\bm\mu}^\top \bm v_{j,k})
		(\bm v_{j,k}^\top \bm P_{\ell,K}\bar{\bm\mu})+o_p(1).
		\]
		Since
		\[
		\bm v_{j,k}^\top \bm P_{\ell,K}\bar{\bm\mu}
		=
		(\bm v_{j,k}^\top \bm v_{\ell,K})(\bm v_{\ell,K}^\top \bar{\bm\mu})
		=
		\rho_{j\ell}^{(k,K)}m_{\ell,K},
		\]
		the second display follows. The third display is analogous.
		
		For the fourth display, use the factorization
		\[
		\bar{\bm\mu}^\top
		\widehat{\bm P}_{j,k}\bm P_{\ell,K}\widehat{\bm P}_{m,k'}\bar{\bm\mu}
		=
		\bigl(\bar{\bm\mu}^\top\widehat{\bm P}_{j,k}\bm v_{\ell,K}\bigr)
		\bigl(\bm v_{\ell,K}^\top\widehat{\bm P}_{m,k'}\bar{\bm\mu}\bigr).
		\]
		By Lemma~\ref{lem:hetero_polarization_overlap},
		\[
		\bar{\bm\mu}^\top\widehat{\bm P}_{j,k}\bm v_{\ell,K}
		=
		a_{j,k}m_{j,k}\rho_{j\ell}^{(k,K)}+o_p(1),
		\]
		and
		\[
		\bm v_{\ell,K}^\top\widehat{\bm P}_{m,k'}\bar{\bm\mu}
		=
		a_{m,k'}\rho_{\ell m}^{(K,k')}m_{m,k'}+o_p(1).
		\]
		Multiplying the two displays yields
		\[
		\bar{\bm\mu}^\top
		\widehat{\bm P}_{j,k}\bm P_{\ell,K}\widehat{\bm P}_{m,k'}\bar{\bm\mu}
		=
		a_{j,k}a_{m,k'}\,
		m_{j,k}\rho_{j\ell}^{(k,K)}\rho_{\ell m}^{(K,k')}m_{m,k'}
		+o_p(1).
		\]
		
		It remains to prove the last display. If \(k=k'\), then the sample eigenvectors
		\(\widehat{\bm v}_{j,k}\) and \(\widehat{\bm v}_{m,k}\) are orthonormal, so
		\[
		\widehat{\bm P}_{j,k}\widehat{\bm P}_{m,k}
		=
		\ind_{\{j=m\}}\widehat{\bm P}_{j,k}.
		\]
		Hence
		\[
		\bar{\bm\mu}^\top \widehat{\bm P}_{j,k}\widehat{\bm P}_{m,k}\bar{\bm\mu}
		=
		\ind_{\{j=m\}}\,
		\bar{\bm\mu}^\top \widehat{\bm P}_{j,k}\bar{\bm\mu}
		=
		a_{j,k}\,\ind_{\{j=m\}}\,m_{j,k}^2+o_p(1),
		\]
		by the first display.
		
		If \(k\neq k'\), condition on the data from domain \(k'\). Then
		\(\widehat{\bm P}_{m,k'}\bar{\bm\mu}\) is fixed relative to the randomness from
		domain \(k\), is independent of \(\widehat{\bm P}_{j,k}\), and has norm bounded
		by \(\|\bar{\bm\mu}\|\). Hence the same polarization argument as in
		Lemma~\ref{lem:hetero_polarization_overlap}, applied conditionally, gives
		\[
		\bar{\bm\mu}^\top \widehat{\bm P}_{j,k}\widehat{\bm P}_{m,k'}\bar{\bm\mu}
		=
		a_{j,k}(\bar{\bm\mu}^\top \bm v_{j,k})
		\bigl(\bm v_{j,k}^\top \widehat{\bm P}_{m,k'}\bar{\bm\mu}\bigr)
		+o_p(1).
		\]
		Applying Lemma~\ref{lem:hetero_polarization_overlap} again, now in domain \(k'\), with
		\(\bm x=\bm v_{j,k}\) and \(\bm y=\bar{\bm\mu}\), yields
		\[
		\bm v_{j,k}^\top \widehat{\bm P}_{m,k'}\bar{\bm\mu}
		=
		a_{m,k'}(\bm v_{j,k}^\top \bm v_{m,k'})(\bm v_{m,k'}^\top \bar{\bm\mu})+o_p(1)
		=
		a_{m,k'}\rho_{jm}^{(k,k')}m_{m,k'}+o_p(1).
		\]
		Substituting this into the previous display gives
		\[
		\bar{\bm\mu}^\top \widehat{\bm P}_{j,k}\widehat{\bm P}_{m,k'}\bar{\bm\mu}
		=
		a_{j,k}a_{m,k'}\,m_{j,k}\rho_{jm}^{(k,k')}m_{m,k'}+o_p(1).
		\]
		This completes the proof.
	\end{proof}
	\begin{lemma}\label{lem:ApH_positive}
		Under Assumptions~\ref{ass1}, \ref{ass2}, \ref{ass4},
		\ref{ass3_hetero} and \ref{ass5_hetero}, the matrix \(\bm A_p^{E}\) is positive
		definite for all sufficiently large \(p\). More precisely, there exists a
		constant \(c_E>0\) such that
		\[
		\lambda_{\min}(\bm A_p^{E})\ge c_E
		\]
		for all sufficiently large \(p\).
	\end{lemma}
	
	\begin{proof}
		Since \(0<a_{j,k}<1\), enlarge the Euclidean space if necessary and choose
		orthonormal vectors
		\(\{\bm f_{j,k}:j\in\mathcal I_k,\ 1\le k\le K\}\) orthogonal to
		\(\mathbb R^p\). Define
		\[
		\bm g_{j,k}
		=
		a_{j,k}\bm v_{j,k}
		+
		\sqrt{a_{j,k}(1-a_{j,k})}\,\bm f_{j,k}.
		\]
		Then
		\[
		\langle \bm g_{j,k},\bm g_{\ell,k'}\rangle
		=
		\begin{cases}
			a_{j,k}\ind_{\{j=\ell\}}, & k=k',\\[1mm]
			a_{j,k}a_{\ell,k'}\rho_{j\ell}^{(k,k')}, & k\ne k',
		\end{cases}
		\qquad
		\langle \bar{\bm\mu},\bm g_{j,k}\rangle
		=
		a_{j,k}m_{j,k}.
		\]
		Set
		\[
		\bm h_k
		=
		\bar{\bm\mu}
		-
		\sum_{j\in\mathcal I_k}
		b_{j,k}m_{j,k}\bm g_{j,k}.
		\]
		Then
		\[
		\begin{aligned}
			\langle \bm h_k,\bm h_{k'}\rangle
			&=
			\|\bar{\bm\mu}\|^2
			-\bm m_k^\top\bm D_k\bm m_k
			-\bm m_{k'}^\top\bm D_{k'}\bm m_{k'}  \\
			&\quad
			+\bm m_k^\top\bm B_k\bm C_{kk'}\bm B_{k'}\bm m_{k'} .
		\end{aligned}
		\]
		For \(\ell\in\mathcal I_K\),
		\[
		\langle \bm v_{\ell,K},\bm h_k\rangle
		=
		m_{\ell,K}
		-
		\bigl(\bm R_{Kk}\bm D_k\bm m_k\bigr)_\ell
		=
		(q_k)_\ell .
		\]
		Let
		\[
		\bm N_K
		=
		\bm I+
		\sum_{\ell\in\mathcal I_K}
		\lambda_{\ell,K}\bm v_{\ell,K}\bm v_{\ell,K}^\top .
		\]
		Since \(\lambda_{\ell,K}>-1\), \(\bm N_K\) is positive definite. Moreover,
		\[
		\begin{aligned}
			\langle \bm h_k,\bm N_K\bm h_{k'}\rangle
			&=
			\langle \bm h_k,\bm h_{k'}\rangle
			+
			\sum_{\ell\in\mathcal I_K}
			\lambda_{\ell,K}(q_k)_\ell(q_{k'})_\ell  \\
			&=
			\Omega_{kk',p}^{E}.
		\end{aligned}
		\]
		Thus \(\bm\Omega_p^{E}=(\Omega_{kk',p}^{E})_{1\le k,k'\le K}\) is positive
		semidefinite.
		
		For any \(\bm w\in\mathbb R^K\), let
		\[
		\widetilde{\bm w}
		=
		(\widetilde w_1,\ldots,\widetilde w_K)^\top,
		\qquad
		\widetilde w_k=\frac{w_k}{\sigma_k^2}.
		\]
		Then
		\[
		\begin{aligned}
			\bm w^\top\bm A_p^{E}\bm w
			&=
			\sigma_K^2\widetilde{\bm w}^{\top}\bm\Omega_p^{E}\widetilde{\bm w}
			+
			\sum_{k=1}^K
			\left(
			\frac{\sigma_K^2\alpha_k^2}{\sigma_k^4}
			+
			\tau_k\frac{\sigma_K^2}{\sigma_k^2}
			\right)w_k^2  \\
			&\ge
			\sum_{k=1}^K
			\tau_k\frac{\sigma_K^2}{\sigma_k^2}w_k^2 .
		\end{aligned}
		\]
		Since \(K\) is fixed, \(\tau_k>0\), and the \(\sigma_k^2\)'s are uniformly
		bounded away from zero and infinity, there exists a constant \(c_E>0\) such that
		\[
		\bm w^\top\bm A_p^{E}\bm w
		\ge c_E\|\bm w\|^2 .
		\]
		Therefore,
		\[
		\lambda_{\min}(\bm A_p^{E})\ge c_E
		\]
		for all sufficiently large \(p\).
	\end{proof}
	
	Next, we formally begin the proof of Theorem~\ref{thm:err_hetero}. For \(k=1,\dots,K\), write
	\[
	\eta_{k,\pm1}:=\widehat{\bm\mu}_{\pm1,k}-\bm\mu_{\pm1,k},
	\qquad
	e_k:=\widehat{\bm\mu}_k-\bm\mu_k=\eta_{k,+}-\eta_{k,-},
	\qquad
	s_k:=\eta_{k,+}+\eta_{k,-}.
	\]
	Then
	\[
	\widehat{\bm\mu}_k=\bm\mu_k+e_k,
	\qquad
	2\bm\mu_{+1,K}-\widehat{\bm\mu}_{+1,K}-\widehat{\bm\mu}_{-1,K}
	=
	\bm\mu_K-s_K,
	\]
	and
	\[
	\widehat{\bm\mu}_{+1,K}+\widehat{\bm\mu}_{-1,K}-2\bm\mu_{-1,K}
	=
	\bm\mu_K+s_K.
	\]
	
	By construction,
	\[
	\widehat{\bm d}^{E}(\bm w)=\sum_{k=1}^K w_k\widehat{\bm d}_k^{E},
	\qquad
	\widehat{\bm d}_k^{E}=\widehat{\bm\Sigma}_k^{-1}\widehat{\bm\mu}_k,
	\]
	and the empirical intercept is
	\[
	\widehat b_K^{E}(\bm w)
	=
	-\widehat{\bm d}^{E}(\bm w)^\top
	\frac{\widehat{\bm\mu}_{+1,K}+\widehat{\bm\mu}_{-1,K}}{2}.
	\]
	
	Define
	\[
	V_n^{E}(\bm w):=
	\widehat{\bm d}^{E}(\bm w)^\top \bm\Sigma_K \widehat{\bm d}^{E}(\bm w),
	\]
	and
	\[
	D_{+,n}^{E}(\bm w)
	:=
	\bigl(2\bm\mu_{+1,K}-\widehat{\bm\mu}_{+1,K}-\widehat{\bm\mu}_{-1,K}\bigr)^\top
	\widehat{\bm d}^{E}(\bm w),
	\]
	\[
	D_{-,n}^{E}(\bm w)
	:=
	\bigl(\widehat{\bm\mu}_{+1,K}+\widehat{\bm\mu}_{-1,K}-2\bm\mu_{-1,K}\bigr)^\top
	\widehat{\bm d}^{E}(\bm w).
	\]
	Then
	\[
	-\bigl(\widehat{\bm d}^{E}(\bm w)^\top\bm\mu_{+1,K}+\widehat b_K^{E}(\bm w)\bigr)
	=
	-\frac12 D_{+,n}^{E}(\bm w),
	\]
	and
	\[
	\widehat{\bm d}^{E}(\bm w)^\top\bm\mu_{-1,K}+\widehat b_K^{E}(\bm w)
	=
	-\frac12 D_{-,n}^{E}(\bm w).
	\]
	Hence
	\begin{align}\label{eq:err_margin_rep_hetero}
		\Err^{E}(\bm w)
		&=
		\frac12\Phi\!\left(
		-\frac{D_{-,n}^{E}(\bm w)}{2\sqrt{V_n^{E}(\bm w)}}
		\right)
		+
		\frac12\Phi\!\left(
		-\frac{D_{+,n}^{E}(\bm w)}{2\sqrt{V_n^{E}(\bm w)}}
		\right).
	\end{align}
	Therefore it suffices to prove that
	\begin{align}
		D_{+,n}^{E}(\bm w)&=(\bm u_p^{E})^\top\bm w-w_K\Delta_K+o_p(1), \label{eq:Dplus_goal_hetero}\\
		D_{-,n}^{E}(\bm w)&=(\bm u_p^{E})^\top\bm w+w_K\Delta_K+o_p(1), \label{eq:Dminus_goal_hetero}\\
		V_n^{E}(\bm w)&=\bm w^\top \bm A_p^{E}\bm w+o_p(1). \label{eq:V_goal_hetero}
	\end{align}
	
	We first record several elementary consequences of Assumptions~\ref{ass1},
	\ref{ass2}, \ref{ass4}, and \ref{ass3_hetero}. Since
	\[
	\bm\mu_k=\bar{\bm\mu}+\bm\delta_k,
	\qquad
	\mathbb E\|\bm\delta_k\|^2=\alpha_k^2,
	\]
	we have \(\|\bm\mu_k\|=O_p(1)\) and \(\|\widehat{\bm\mu}_k\|=O_p(1)\) for each fixed \(k\).
	Moreover, for fixed \(k,k'\),
	\begin{gather}
		\bm\mu_K^\top e_k=o_p(1),\qquad
		\bm\mu_K^\top s_K=o_p(1),\qquad
		s_K^\top \bm\mu_k=o_p(1), \label{eq:hetero_mu_e_small}\\
		\frac{1}{\sigma_k^2}s_K^\top e_k
		=
		\Delta_K\ind_{\{k=K\}}+o_p(1), \label{eq:hetero_s_e_limit}\\
		e_k^\top e_{k'}
		=
		\tau_k\sigma_k^2\,\ind_{\{k=k'\}}+o_p(1), \label{eq:hetero_e_e_limit}\\
		\bm\mu_K^\top\bm\mu_k
		=
		\|\bar{\bm\mu}\|^2+\alpha_K^2\ind_{\{k=K\}}+o_p(1). \label{eq:hetero_mu_mu_limit}
	\end{gather}
	Indeed, the above relations follow from standard second-moment calculations,
	using that each classwise sample-mean noise has covariance
	\[
	\var(\eta_{k,\pm1})=\frac{1}{n_{k,\pm1}}\bm\Sigma_k.
	\]
	Also,
	\[
	e_k=\eta_{k,+}-\eta_{k,-},
	\qquad
	\var(e_k)=\left(\frac{1}{n_{k,+}}+\frac{1}{n_{k,-}}\right)\bm\Sigma_k,
	\]
	which gives \eqref{eq:hetero_e_e_limit}. Moreover,
	\[
	s_K^\top e_K=\|\eta_{K,+}\|^2-\|\eta_{K,-}\|^2,
	\]
	and hence \eqref{eq:hetero_s_e_limit} follows by concentration of the two classwise squared norms. Finally,
	\eqref{eq:hetero_mu_mu_limit} follows from Assumption~\ref{ass2}.
	
	Next, by \eqref{eq:Sigma_tilde_hetero},
	\[
	\widehat{\bm\Sigma}_k
	=
	\sigma_k^2
	\left(
	\bm I_p+\sum_{j\in\mathcal I_k}\lambda_{j,k}\widehat{\bm P}_{j,k}
	\right),
	\qquad
	\widehat{\bm P}_{j,k}:=\widehat{\bm v}_{j,k}\widehat{\bm v}_{j,k}^\top.
	\]
	Since the projectors \(\widehat{\bm P}_{j,k}\) are mutually orthogonal,
	\[
	\widehat{\bm\Sigma}_k^{-1}
	=
	\frac1{\sigma_k^2}
	\left(
	\bm I_p-\sum_{j\in\mathcal I_k} b_{j,k}\widehat{\bm P}_{j,k}
	\right),
	\]
	and hence
	\begin{equation}\label{eq:dhat_expand_hetero}
		\widehat{\bm d}_k^{E}
		=
		\frac1{\sigma_k^2}
		\left(
		\widehat{\bm\mu}_k-\sum_{j\in\mathcal I_k} b_{j,k}\widehat{\bm P}_{j,k}\widehat{\bm\mu}_k
		\right).
	\end{equation}
	
	We now analyze the spike-projector terms. By
	Lemma~\ref{lem:hetero_polarization_overlap},
	\begin{equation}\label{eq:hetero_polarization_main}
		\bm x^\top \widehat{\bm P}_{j,k}\bm y
		=
		a_{j,k}(\bm x^\top \bm v_{j,k})(\bm y^\top \bm v_{j,k})+o_p(1)
	\end{equation}
	for any deterministic vectors \(\bm x,\bm y\) with bounded Euclidean norms.
	In particular,
	\begin{equation}\label{eq:hetero_mubar_P_mubar}
		\bar{\bm\mu}^\top \widehat{\bm P}_{j,k}\bar{\bm\mu}
		=
		a_{j,k}m_{j,k}^2+o_p(1).
	\end{equation}
	
	Also, for any fixed \(a\in\{1,\dots,K\}\), since \(\bm\delta_a\) is independent of the
	training data and
	\[
	\mathbb E\bigl((\bm\delta_a^\top \widehat{\bm v}_{j,k})^2\mid \widehat{\bm v}_{j,k}\bigr)
	=
	\frac{\alpha_a^2}{p},
	\]
	we have
	\begin{equation}\label{eq:hetero_delta_vhat_small}
		\bm\delta_a^\top \widehat{\bm v}_{j,k}\xrightarrow{p}0,
		\qquad
		\bm\delta_a^\top \widehat{\bm P}_{j,k}\bm\delta_b=o_p(1),
		\qquad
		\bar{\bm\mu}^\top \widehat{\bm P}_{j,k}\bm\delta_a=o_p(1).
	\end{equation}
	Furthermore, by Lemma~\ref{lem:hetero_mean_noise_vhat},
	\begin{equation}\label{eq:hetero_e_vhat_small}
		e_k^\top \widehat{\bm v}_{j,k}\xrightarrow{p}0,
		\qquad
		e_k^\top \widehat{\bm P}_{j,k}e_k=o_p(1),
		\qquad
		\bar{\bm\mu}^\top \widehat{\bm P}_{j,k}e_k=o_p(1).
	\end{equation}
	Moreover, for any fixed \(a,k\in\{1,\dots,K\}\) and \(j\in\mathcal I_k\),
	\[
	e_a^\top \widehat{\bm v}_{j,k}\xrightarrow{p}0.
	\]
	If \(a=k\), this is Lemma~\ref{lem:hetero_mean_noise_vhat}; if \(a\neq k\), then \(e_a\) is independent of
	\(\widehat{\bm v}_{j,k}\), and
	\[
	\mathbb E\!\left[(e_a^\top \widehat{\bm v}_{j,k})^2\mid \widehat{\bm v}_{j,k}\right]
	=
	\left(\frac{1}{n_{a,+}}+\frac{1}{n_{a,-}}\right)\widehat{\bm v}_{j,k}^\top \bm\Sigma_a \widehat{\bm v}_{j,k}
	=
	O(n_{a,+}^{-1}+n_{a,-}^{-1}),
	\]
	which implies the claim.
	If \(k=K\), Lemma~\ref{lem:hetero_mean_noise_vhat} also gives
	\[
	s_K^\top \widehat{\bm v}_{j,K}\xrightarrow{p}0.
	\]
	If \(k\neq K\), then \(s_K\) is independent of \(\widehat{\bm v}_{j,k}\), and
	\[
	\mathbb E\bigl((s_K^\top \widehat{\bm v}_{j,k})^2\mid \widehat{\bm v}_{j,k}\bigr)
	=
	\left(\frac{1}{n_{K,+}}+\frac{1}{n_{K,-}}\right)\widehat{\bm v}_{j,k}^\top \bm\Sigma_K\widehat{\bm v}_{j,k}
	=
	O(n_{K,+}^{-1}+n_{K,-}^{-1}),
	\]
	so again
	\begin{equation}\label{eq:hetero_s_vhat_small}
		s_K^\top \widehat{\bm v}_{j,k}\xrightarrow{p}0.
	\end{equation}
	Combining \eqref{eq:hetero_mubar_P_mubar},
	\eqref{eq:hetero_delta_vhat_small}, \eqref{eq:hetero_e_vhat_small},
	and \eqref{eq:hetero_s_vhat_small}, we obtain, for each fixed \(k\) and
	\(j\in\mathcal I_k\),
	\begin{equation}\label{eq:hetero_muKminusS_P_mukhat}
		(\bm\mu_K-s_K)^\top \widehat{\bm P}_{j,k}\widehat{\bm\mu}_k
		=
		a_{j,k}m_{j,k}^2+o_p(1),
	\end{equation}
	and likewise
	\begin{equation}\label{eq:hetero_muKplusS_P_mukhat}
		(\bm\mu_K+s_K)^\top \widehat{\bm P}_{j,k}\widehat{\bm\mu}_k
		=
		a_{j,k}m_{j,k}^2+o_p(1).
	\end{equation}
	
	We now prove \eqref{eq:Dplus_goal_hetero}. By \eqref{eq:dhat_expand_hetero},
	\begin{align*}
		D_{+,n}^{E}(\bm w)
		&=
		(\bm\mu_K-s_K)^\top \widehat{\bm d}^{E}(\bm w)\\
		&=
		\sum_{k=1}^K \frac{w_k}{\sigma_k^2}
		\left[
		(\bm\mu_K-s_K)^\top \widehat{\bm\mu}_k
		-
		\sum_{j\in\mathcal I_k}
		b_{j,k}(\bm\mu_K-s_K)^\top \widehat{\bm P}_{j,k}\widehat{\bm\mu}_k
		\right].
	\end{align*}
	By \eqref{eq:hetero_mu_mu_limit}, \eqref{eq:hetero_mu_e_small}, and \eqref{eq:hetero_s_e_limit},
	\[
	\frac{1}{\sigma_k^2}(\bm\mu_K-s_K)^\top \widehat{\bm\mu}_k
	=
	\frac{\|\bar{\bm\mu}\|^2}{\sigma_k^2}
	+
	\frac{\alpha_K^2}{\sigma_K^2}\ind_{\{k=K\}}
	-
	\Delta_K\ind_{\{k=K\}}
	+
	o_p(1).
	\]
	Together with \eqref{eq:hetero_muKminusS_P_mukhat}, this yields
	\begin{align*}
		D_{+,n}^{E}(\bm w)
		&=
		\sum_{k=1}^K w_k
		\left[
		\frac{\|\bar{\bm\mu}\|^2-\sum_{j\in\mathcal I_k}b_{j,k}a_{j,k}m_{j,k}^2}{\sigma_k^2}
		+
		\frac{\alpha_K^2}{\sigma_K^2}\ind_{\{k=K\}}
		-
		\Delta_K\ind_{\{k=K\}}
		\right]
		+o_p(1).
	\end{align*}
	Hence
	\[
	D_{+,n}^{E}(\bm w)=(\bm u_p^{E})^\top \bm w-w_K\Delta_K+o_p(1).
	\]
	Similarly,
	\[
	D_{-,n}^{E}(\bm w)
	=
	(\bm u_p^{E})^\top\bm w+w_K\Delta_K+o_p(1),
	\]
	which proves \eqref{eq:Dminus_goal_hetero}.
	
	It remains to prove \eqref{eq:V_goal_hetero}. Write
	\[
	V_n^{E}(\bm w)
	=
	\sum_{k,k'=1}^K w_k w_{k'}
	T_{kk',n}^{E},
	\]
	where
	\[
	T_{kk',n}^{E}
	:=
	\widehat{\bm\mu}_k^\top
	\widehat{\bm\Sigma}_k^{-1}\bm\Sigma_K\widehat{\bm\Sigma}_{k'}^{-1}
	\widehat{\bm\mu}_{k'}.
	\]
	Thus it suffices to show that, for each fixed \(k,k'\),
	\begin{equation}\label{eq:Tkk_goal_hetero}
		T_{kk',n}^{E}
		=
		A_{kk',p}^{E}+o_p(1).
	\end{equation}
	
	Define
	\[
	\bm M_k
	:=
	\bm I_p-\sum_{j\in\mathcal I_k} b_{j,k}\widehat{\bm P}_{j,k},
	\qquad
	\bm N_K
	:=
	\bm I_p+\sum_{\ell\in\mathcal I_K}\lambda_{\ell,K}\bm P_{\ell,K}.
	\]
	Then
	\[
	\widehat{\bm\Sigma}_k^{-1}=\frac{1}{\sigma_k^2}\bm M_k,
	\qquad
	\bm\Sigma_K=\sigma_K^2\bm N_K,
	\]
	and therefore
	\begin{equation}\label{eq:hetero_sigma_expand}
		T_{kk',n}^{E}
		=
		\frac{\sigma_K^2}{\sigma_k^2\sigma_{k'}^2}
		\widehat{\bm\mu}_k^\top \bm M_k \bm N_K \bm M_{k'} \widehat{\bm\mu}_{k'}.
	\end{equation}
	
	Now decompose
	\[
	\widehat{\bm\mu}_k=\bar{\bm\mu}+\bm\delta_k+e_k,
	\qquad
	\widehat{\bm\mu}_{k'}=\bar{\bm\mu}+\bm\delta_{k'}+e_{k'}.
	\]
	Since each \(\widehat{\bm P}_{j,k}\) and \(\bm P_{\ell,K}\) is rank one, every mixed term
	involving at least one factor among \(\bm\delta_k,e_k,\bm\delta_{k'},e_{k'}\) factors through an
	inner product of one of these random vectors with either a sample eigenvector
	\(\widehat{\bm v}_{j,k}\), \(\widehat{\bm v}_{m,k'}\) or a deterministic spike direction
	\(\bm v_{\ell,K}\). By \eqref{eq:hetero_delta_vhat_small},
	\eqref{eq:hetero_e_vhat_small}, and the deterministic bounds
	\[
	\bm\delta_a^\top \bm v_{\ell,K}\xrightarrow{p}0,
	\qquad
	e_a^\top \bm v_{\ell,K}\xrightarrow{p}0
	\qquad (a\in\{k,k'\}),
	\]
	all such mixed terms are \(o_p(1)\), except for the pure identity contributions
	\(\bm\delta_k^\top\bm\delta_{k'}\) and \(e_k^\top e_{k'}\). Consequently,
	\begin{align}
		T_{kk',n}^{E}
		&=
		\frac{\sigma_K^2}{\sigma_k^2\sigma_{k'}^2}
		\bar{\bm\mu}^\top \bm M_k \bm N_K \bm M_{k'}\bar{\bm\mu}
		+
		\frac{\sigma_K^2}{\sigma_k^2\sigma_{k'}^2}\bm\delta_k^\top\bm\delta_{k'}
		+
		\frac{\sigma_K^2}{\sigma_k^2\sigma_{k'}^2}e_k^\top e_{k'}
		+
		o_p(1).
		\label{eq:hetero_reduce_main}
	\end{align}
	
	We next evaluate the deterministic principal term.
	First,
	\begin{align*}
		&\bar{\bm\mu}^\top \bm M_k \bm M_{k'}\bar{\bm\mu}
		\\
		=&\|\bar{\bm\mu}\|^2
		-
		\sum_{j\in\mathcal I_k} b_{j,k}\,
		\bar{\bm\mu}^\top \widehat{\bm P}_{j,k}\bar{\bm\mu}
		-
		\sum_{m\in\mathcal I_{k'}} b_{m,k'}\,
		\bar{\bm\mu}^\top \widehat{\bm P}_{m,k'}\bar{\bm\mu}
		+
		\sum_{j\in\mathcal I_k}\sum_{m\in\mathcal I_{k'}}
		b_{j,k}b_{m,k'}\,
		\bar{\bm\mu}^\top \widehat{\bm P}_{j,k}\widehat{\bm P}_{m,k'}\bar{\bm\mu}.
	\end{align*}
	By the first and the last displays of Lemma~\ref{lem:hetero_mixed_projectors},
	\begin{equation}\label{eq:hetero_MM_term}
		\bar{\bm\mu}^\top \bm M_k \bm M_{k'}\bar{\bm\mu}
		=
		\|\bar{\bm\mu}\|^2
		-
		\bm m_k^\top \bm D_k\bm m_k
		-
		\bm m_{k'}^\top \bm D_{k'}\bm m_{k'}
		+
		\bm m_k^\top \bm B_k\bm C_{kk'}\bm B_{k'}\bm m_{k'}
		+
		o_p(1).
	\end{equation}
	
	Next, expand the target-spike part:
	\begin{align*}
		&\bar{\bm\mu}^\top \bm M_k
		\left(\sum_{\ell\in\mathcal I_K}\lambda_{\ell,K}\bm P_{\ell,K}\right)
		\bm M_{k'}\bar{\bm\mu}\\
		&=
		\sum_{\ell\in\mathcal I_K}\lambda_{\ell,K}\,
		\bar{\bm\mu}^\top \bm P_{\ell,K}\bar{\bm\mu}
		-
		\sum_{j\in\mathcal I_k}\sum_{\ell\in\mathcal I_K}
		\lambda_{\ell,K}b_{j,k}\,
		\bar{\bm\mu}^\top \widehat{\bm P}_{j,k}\bm P_{\ell,K}\bar{\bm\mu}\\
		&\quad
		-
		\sum_{\ell\in\mathcal I_K}\sum_{m\in\mathcal I_{k'}}
		\lambda_{\ell,K}b_{m,k'}\,
		\bar{\bm\mu}^\top \bm P_{\ell,K}\widehat{\bm P}_{m,k'}\bar{\bm\mu}\\
		&\quad
		+
		\sum_{j\in\mathcal I_k}\sum_{\ell\in\mathcal I_K}\sum_{m\in\mathcal I_{k'}}
		\lambda_{\ell,K}b_{j,k}b_{m,k'}\,
		\bar{\bm\mu}^\top
		\widehat{\bm P}_{j,k}\bm P_{\ell,K}\widehat{\bm P}_{m,k'}\bar{\bm\mu}.
	\end{align*}
	By the second, third, and fourth displays of
	Lemma~\ref{lem:hetero_mixed_projectors}, this becomes
	\begin{equation}\label{eq:hetero_targetspike_term}
		\bar{\bm\mu}^\top \bm M_k
		\left(\sum_{\ell\in\mathcal I_K}\lambda_{\ell,K}\bm P_{\ell,K}\right)
		\bm M_{k'}\bar{\bm\mu}
		=
		\bm q_k^\top \bm\Lambda_K \bm q_{k'}+o_p(1).
	\end{equation}
	Combining \eqref{eq:hetero_MM_term} and \eqref{eq:hetero_targetspike_term},
	we obtain
	\begin{equation}\label{eq:hetero_principal_term}
		\bar{\bm\mu}^\top \bm M_k \bm N_K \bm M_{k'}\bar{\bm\mu}
		=
		\Omega_{kk',p}^{E}+o_p(1).
	\end{equation}
	
	Finally, by Assumption~\ref{ass2},
	\[
	\bm\delta_k^\top\bm\delta_{k'}
	=
	\alpha_k^2\ind_{\{k=k'\}}+o_p(1),
	\]
	and by \eqref{eq:hetero_e_e_limit},
	\[
	e_k^\top e_{k'}
	=
	\tau_k\sigma_k^2\,\ind_{\{k=k'\}}+o_p(1).
	\]
	Substituting these together with \eqref{eq:hetero_principal_term} into
	\eqref{eq:hetero_reduce_main}, we conclude that
	\[
	T_{kk',n}^{E}
	=
	\frac{\sigma_K^2}{\sigma_k^2\sigma_{k'}^2}\Omega_{kk',p}^{E}
	+
	\left(
	\frac{\sigma_K^2\alpha_k^2}{\sigma_k^4}
	+
	\tau_k\frac{\sigma_K^2}{\sigma_k^2}
	\right)\ind_{\{k=k'\}}
	+
	o_p(1).
	\]
	By the definition of \(A_{kk',p}^{E}\), this proves
	\eqref{eq:Tkk_goal_hetero}, and hence \eqref{eq:V_goal_hetero}.
	
	Therefore,
	\[
	\frac{D_{+,n}^{E}(\bm w)}{2\sqrt{V_n^{E}(\bm w)}}
	=
	\frac{(\bm u_p^{E})^\top \bm w-w_K\Delta_K}{2\sqrt{\bm w^\top \bm A_p^{E}\bm w}}
	+o_p(1),
	\qquad
	\frac{D_{-,n}^{E}(\bm w)}{2\sqrt{V_n^{E}(\bm w)}}
	=
	\frac{(\bm u_p^{E})^\top \bm w+w_K\Delta_K}{2\sqrt{\bm w^\top \bm A_p^{E}\bm w}}
	+o_p(1),
	\]
	because \(\bm w\neq\bm0\) and \(\bm A_p^{E}\) is positive definite by
	Lemma~\ref{lem:ApH_positive}. Substituting the above into
	\eqref{eq:err_margin_rep_hetero} and using the continuity of \(\Phi\), we obtain
	\eqref{eq:err_limit_uncorrected_hetero}.
	This completes the proof.
	\subsection{Proof of Proposition \ref{prop:optimal_intercept_hetero}}
	For a fixed scalar adjustment \(t\), the proof of Theorem~\ref{thm:err_hetero} gives the modified margins
	\[
	D_{+,n}^{E,(t)}(\bm w)=D_{+,n}^{E}(\bm w)+2t,
	\qquad
	D_{-,n}^{E,(t)}(\bm w)=D_{-,n}^{E}(\bm w)-2t.
	\]
	Therefore
	\[
	\Err_t^{E}(\bm w)
	=
	\frac12\Phi\left(
	-\frac{D_{-,n}^{E}(\bm w)-2t}{2\sqrt{V_n^{E}(\bm w)}}
	\right)
	+
	\frac12\Phi\left(
	-\frac{D_{+,n}^{E}(\bm w)+2t}{2\sqrt{V_n^{E}(\bm w)}}
	\right).
	\]
	On the event \(D_{+,n}^{E}(\bm w)+D_{-,n}^{E}(\bm w)>0\), differentiating the last display with respect to \(t\) shows that its unique minimizer is
	\[
	\widehat t_{p,E}^\ast(\bm w)=\frac14\{D_{-,n}^{E}(\bm w)-D_{+,n}^{E}(\bm w)\}.
	\]
	Since \((\bm u_p^{E})^\top\bm w>0\), \eqref{eq:Dplus_goal_hetero} and \eqref{eq:Dminus_goal_hetero} imply
	\[
	D_{+,n}^{E}(\bm w)+D_{-,n}^{E}(\bm w)
	=2(\bm u_p^{E})^\top\bm w+o_p(1)>0
	\]
	with probability tending to one. Moreover,
	\[
	\widehat t_{p,E}^\ast(\bm w)
	=\frac12w_K\Delta_K+o_p(1)
	=\frac12w_K\widehat\Delta_K+o_p(1),
	\]
	because \(\widehat\Delta_K\to\Delta_K\). Substituting this minimizer into the margin representation and using \eqref{eq:Dplus_goal_hetero}, \eqref{eq:Dminus_goal_hetero}, and \eqref{eq:V_goal_hetero} yields
	\[
	\Err_{\widehat t_{p,E}^\ast(\bm w)}^{E}(\bm w)
	-
	\Phi\left(
	-\frac{(\bm u_p^{E})^\top\bm w}{2\sqrt{\bm w^\top\bm A_p^{E}\bm w}}
	\right)
	\to0
	\]
	in probability. This completes the proof.
	
	\subsection{Proof of Corollary \ref{cor:oracle_weight_hetero}}
	By Proposition~\ref{prop:optimal_intercept_hetero}, the deterministic equivalent of the target-domain
	Gaussian-calibrated error under the optimal intercept is
	\[
	\Phi\left(
	-\frac{(\bm u_p^{E})^\top \bm w}{2\sqrt{\bm w^\top \bm A_p^{E}\bm w}}
	\right).
	\]
	Since \(\Phi\) is strictly increasing, minimizing this quantity is equivalent to maximizing
	\[
	\frac{(\bm u_p^{E})^\top \bm w}{\sqrt{\bm w^\top \bm A_p^{E}\bm w}}
	\]
	over all \(\bm w\neq \bm 0\). This criterion is invariant under positive rescaling of
	\(\bm w\). Therefore the oracle weight is identifiable only up to a positive scalar multiple.
	
	By Lemma~\ref{lem:ApH_positive}, \(\bm A_p^{E}\) is positive definite. Therefore,
	the Cauchy--Schwarz inequality under the \(\bm A_p^{E}\)-inner product gives
	\[
	(\bm u_p^{E})^\top \bm w
	=
	\bigl((\bm A_p^{E})^{-1/2}\bm u_p^{E}\bigr)^\top
	\bigl((\bm A_p^{E})^{1/2}\bm w\bigr)
	\le
	\sqrt{(\bm u_p^{E})^\top(\bm A_p^{E})^{-1}\bm u_p^{E}}\,
	\sqrt{\bm w^\top \bm A_p^{E}\bm w}.
	\]
	Hence
	\[
	\frac{(\bm u_p^{E})^\top \bm w}{\sqrt{\bm w^\top \bm A_p^{E}\bm w}}
	\le
	\sqrt{(\bm u_p^{E})^\top(\bm A_p^{E})^{-1}\bm u_p^{E}}.
	\]
	Equality holds if and only if
	\[
	(\bm A_p^{E})^{1/2}\bm w
	\propto
	(\bm A_p^{E})^{-1/2}\bm u_p^{E},
	\]
	or equivalently,
	\[
	\bm w\propto (\bm A_p^{E})^{-1}\bm u_p^{E}.
	\]
	Therefore the deterministic equivalent is minimized by
	\[
	\bm w_{p,E}^\ast \propto (\bm A_p^{E})^{-1}\bm u_p^{E}.
	\]
	
	Under the normalization \(\bm w^\top \bm A_p^{E}\bm w=1\), the unique maximizer is
	\[
	\bm w_{p,E}^\ast
	=
	\frac{(\bm A_p^{E})^{-1}\bm u_p^{E}}
	{\sqrt{(\bm u_p^{E})^\top(\bm A_p^{E})^{-1}\bm u_p^{E}}}.
	\]
	This completes the proof.
	\subsection{Proof of Theorem \ref{thm:emp_weight_hetero}}
	For convenience, for each \(k=1,\dots,K\) and \(\ell\in\mathcal I_K\), write
	\[
	q_{k,\ell}:=[\bm q_k]_\ell
	=
	m_{\ell,K}
	-
	\sum_{j\in\mathcal I_k}
	b_{j,k}a_{j,k}\rho_{\ell j}^{(K,k)}m_{j,k}.
	\]
	Then
	\[
	\bm q_k^\top \bm\Lambda_K \bm q_{k'}
	=
	\sum_{\ell\in\mathcal I_K}\lambda_{\ell,K}q_{k,\ell}q_{k',\ell},
	\]
	and in particular
	\[
	q_{K,\ell}=(1-b_{\ell,K}a_{\ell,K})m_{\ell,K}.
	\]
	
	We first prove the consistency of the basic scalar quantities.
	From the calculations in the proof of Theorem~\ref{thm:err_hetero}, for any fixed
	\(a\neq b\),
	\[
	\widehat{\bm\mu}_a^\top\widehat{\bm\mu}_b
	=
	\|\bar{\bm\mu}\|^2+o_p(1),
	\]
	while for each fixed \(k\),
	\[
	\|\widehat{\bm\mu}_k\|^2
	=
	\|\bar{\bm\mu}\|^2+\alpha_k^2+
	\sigma_k^2\left(\frac{p}{n_{k,+}}+\frac{p}{n_{k,-}}\right)+o_p(1)
	=
	\|\bar{\bm\mu}\|^2+\alpha_k^2+\sigma_k^2\tau_k+o_p(1).
	\]
	Hence, since \(K\) is fixed,
	\[
	\widehat{\mathcal A}_p(\bm I_p)
	=
	\|\bar{\bm\mu}\|^2+o_p(1),
	\]
	and
	\[
	\widehat{\beta}_k^{E}
	=
	\alpha_k^2+o_p(1),
	\qquad k=1,\dots,K.
	\]
	
	Next fix \(k\in\{1,\dots,K\}\). Since
	\[
	\widehat{\bm M}_k
	=
	\bm I_p-\sum_{j\in\mathcal I_k}b_{j,k}\widehat{\bm P}_{j,k},
	\]
	for any fixed \(a\neq b\),
	\[
	\widehat{\bm\mu}_a^\top \widehat{\bm M}_k \widehat{\bm\mu}_b
	=
	\widehat{\bm\mu}_a^\top \widehat{\bm\mu}_b
	-
	\sum_{j\in\mathcal I_k}
	b_{j,k}\widehat{\bm\mu}_a^\top \widehat{\bm P}_{j,k}\widehat{\bm\mu}_b.
	\]
	As in the proof of Theorem~\ref{thm:err_hetero}, write
	\[
	\widehat{\bm\mu}_a=\bar{\bm\mu}+\bm\delta_a+e_a,
	\qquad
	\widehat{\bm\mu}_b=\bar{\bm\mu}+\bm\delta_b+e_b.
	\]
	By Lemmas~\ref{lem:hetero_polarization_overlap}
	and~\ref{lem:hetero_mean_noise_vhat}, together with the independence of
	\(\bm\delta_a,\bm\delta_b\) from the training data, all terms involving at least one factor
	among \(\bm\delta_a,\bm\delta_b\) are \(o_p(1)\). For the sample-mean noise terms, if
	\(a=k\) then Lemma~\ref{lem:hetero_mean_noise_vhat} gives
	\[
	e_a^\top \widehat{\bm v}_{j,k}\xrightarrow{p}0.
	\]
	If \(a\neq k\), then \(e_a\) is independent of \(\widehat{\bm v}_{j,k}\), and
	\[
	\mathbb E\!\left[(e_a^\top \widehat{\bm v}_{j,k})^2\mid \widehat{\bm v}_{j,k}\right]
	=
	\left(\frac{1}{n_{a,+}}+\frac{1}{n_{a,-}}\right)\widehat{\bm v}_{j,k}^\top \bm\Sigma_a \widehat{\bm v}_{j,k}
	=
	O(n_{a,+}^{-1}+n_{a,-}^{-1}),
	\]
	so \(e_a^\top \widehat{\bm v}_{j,k}\xrightarrow{p}0\). The same argument applies to \(e_b\).
	Hence every term involving at least one factor among \(e_a,e_b\) is also \(o_p(1)\).
	Therefore, for each fixed
	\(j\in\mathcal I_k\),
	\[
	\widehat{\bm\mu}_a^\top \widehat{\bm P}_{j,k}\widehat{\bm\mu}_b
	=
	\bar{\bm\mu}^\top \widehat{\bm P}_{j,k}\bar{\bm\mu}
	+o_p(1)
	=
	a_{j,k}m_{j,k}^2+o_p(1).
	\]
	Hence
	\[
	\widehat{\bm\mu}_a^\top \widehat{\bm M}_k \widehat{\bm\mu}_b
	=
	\|\bar{\bm\mu}\|^2
	-
	\sum_{j\in\mathcal I_k}b_{j,k}a_{j,k}m_{j,k}^2
	+o_p(1)
	=
	\|\bar{\bm\mu}\|^2-\bm m_k^\top\bm D_k\bm m_k+o_p(1).
	\]
	Averaging over \(a\neq b\) gives
	\[
	\widehat{\mathcal A}_p(\widehat{\bm M}_k)
	=
	\|\bar{\bm\mu}\|^2-\bm m_k^\top\bm D_k\bm m_k+o_p(1).
	\]
	Consequently,
	\[
	\widehat u_{k,p}^{E}
	=
	\frac{\|\bar{\bm\mu}\|^2-\bm m_k^\top\bm D_k\bm m_k}{\sigma_k^2}
	+
	\frac{\alpha_K^2}{\sigma_K^2}\ind_{\{k=K\}}
	+o_p(1)
	=
	u_{k,p}^{E}+o_p(1).
	\]
	Thus
	\[
	\widehat{\bm u}_p^{E}-\bm u_p^{E}\to \bm 0
	\qquad\text{in probability}.
	\]
	
	We now turn to \(\widehat{\Omega}_{kk',p}^{E}\).
	Fix \(k,k'\in\{1,\dots,K\}\). Expanding
	\[
	\widehat{\bm M}_k\widehat{\bm M}_{k'}
	=
	\left(
	\bm I_p-\sum_{j\in\mathcal I_k}b_{j,k}\widehat{\bm P}_{j,k}
	\right)
	\left(
	\bm I_p-\sum_{m\in\mathcal I_{k'}}b_{m,k'}\widehat{\bm P}_{m,k'}
	\right),
	\]
	we obtain, for any fixed \(a\neq b\),
	\begin{align*}
		\widehat{\bm\mu}_a^\top \widehat{\bm M}_k\widehat{\bm M}_{k'} \widehat{\bm\mu}_b
		&=
		\widehat{\bm\mu}_a^\top\widehat{\bm\mu}_b
		-
		\sum_{j\in\mathcal I_k}
		b_{j,k}\widehat{\bm\mu}_a^\top\widehat{\bm P}_{j,k}\widehat{\bm\mu}_b
		-
		\sum_{m\in\mathcal I_{k'}}
		b_{m,k'}\widehat{\bm\mu}_a^\top\widehat{\bm P}_{m,k'}\widehat{\bm\mu}_b\\
		&\quad
		+
		\sum_{j\in\mathcal I_k}\sum_{m\in\mathcal I_{k'}}
		b_{j,k}b_{m,k'}
		\widehat{\bm\mu}_a^\top\widehat{\bm P}_{j,k}\widehat{\bm P}_{m,k'}\widehat{\bm\mu}_b.
	\end{align*}
	By the same reduction as above and the last display of
	Lemma~\ref{lem:hetero_mixed_projectors},
	\[
	\widehat{\bm\mu}_a^\top\widehat{\bm P}_{j,k}\widehat{\bm P}_{m,k'}\widehat{\bm\mu}_b
	=
	\bar{\bm\mu}^\top\widehat{\bm P}_{j,k}\widehat{\bm P}_{m,k'}\bar{\bm\mu}
	+o_p(1).
	\]
	Hence
	\begin{align*}
		\widehat{\bm\mu}_a^\top \widehat{\bm M}_k\widehat{\bm M}_{k'} \widehat{\bm\mu}_b
		&=
		\|\bar{\bm\mu}\|^2
		-
		\bm m_k^\top\bm D_k\bm m_k
		-
		\bm m_{k'}^\top\bm D_{k'}\bm m_{k'}\\
		&\quad
		+
		\bm m_k^\top\bm B_k\bm C_{kk'}\bm B_{k'}\bm m_{k'}
		+o_p(1).
	\end{align*}
	Averaging over \(a\neq b\) yields
	\begin{equation}\label{eq:AkMkMkprime}
		\widehat{\mathcal A}_p(\widehat{\bm M}_k\widehat{\bm M}_{k'})
		=
		\|\bar{\bm\mu}\|^2
		-
		\bm m_k^\top\bm D_k\bm m_k
		-
		\bm m_{k'}^\top\bm D_{k'}\bm m_{k'}
		+
		\bm m_k^\top\bm B_k\bm C_{kk'}\bm B_{k'}\bm m_{k'}
		+o_p(1).
	\end{equation}
	
	Next we show that, for each fixed \(\ell\in\mathcal I_K\),
	\begin{equation}\label{eq:Theta_consistency_goal}
		\widehat{\Theta}_{kk',\ell}
		=
		q_{k,\ell}q_{k',\ell}+o_p(1).
	\end{equation}
	Since \(\widehat a_{\ell,K}\to a_{\ell,K}\) and \(a_{\ell,K}\neq0\), replacing \(a_{\ell,K}\) by \(\widehat a_{\ell,K}\) in the normalizing factors does not change the limit.
	First consider the case \(k\neq K\) and \(k'\neq K\). By definition,
	\[
	\widehat{\Theta}_{kk',\ell}
	=
	\frac{1}{\widehat a_{\ell,K}}\,
	\widehat{\mathcal A}_p(\widehat{\bm M}_k\widehat{\bm P}_{\ell,K}\widehat{\bm M}_{k'}).
	\]
	Expanding the product and arguing exactly as above, using the second, third, and fourth
	displays of Lemma~\ref{lem:hetero_mixed_projectors}, we obtain
	\begin{align*}
		\frac{1}{\widehat a_{\ell,K}}
		\widehat{\bm\mu}_a^\top
		\widehat{\bm M}_k\widehat{\bm P}_{\ell,K}\widehat{\bm M}_{k'}
		\widehat{\bm\mu}_b
		&=
		m_{\ell,K}^2
		-
		\sum_{j\in\mathcal I_k}
		b_{j,k}a_{j,k}m_{j,k}\rho_{j\ell}^{(k,K)}m_{\ell,K}\\
		&\quad
		-
		\sum_{m\in\mathcal I_{k'}}
		b_{m,k'}a_{m,k'}m_{\ell,K}\rho_{\ell m}^{(K,k')}m_{m,k'}\\
		&\quad
		+
		\sum_{j\in\mathcal I_k}\sum_{m\in\mathcal I_{k'}}
		b_{j,k}b_{m,k'}a_{j,k}a_{m,k'}
		m_{j,k}\rho_{j\ell}^{(k,K)}\rho_{\ell m}^{(K,k')}m_{m,k'}\\
		&\quad +o_p(1).
	\end{align*}
	Since \(\widehat a_{\ell,K}\to a_{\ell,K}\) and
	\(\rho_{j\ell}^{(k,K)}=\rho_{\ell j}^{(K,k)}\), the right-hand side equals
	\(q_{k,\ell}q_{k',\ell}+o_p(1)\). Averaging over \(a\neq b\) gives
	\[
	\widehat{\Theta}_{kk',\ell}
	=
	q_{k,\ell}q_{k',\ell}+o_p(1).
	\]
	
	Next consider the case \(k=K\), \(k'\neq K\). By definition,
	\[
	\widehat{\Theta}_{Kk',\ell}
	=
	\frac{1-b_{\ell,K}\widehat a_{\ell,K}}{\widehat a_{\ell,K}}
	\widehat{\mathcal A}_p(\widehat{\bm P}_{\ell,K}\widehat{\bm M}_{k'}).
	\]
	Again by Lemma~\ref{lem:hetero_mixed_projectors},
	\[
	\frac{1}{\widehat a_{\ell,K}}
	\widehat{\bm\mu}_a^\top
	\widehat{\bm P}_{\ell,K}\widehat{\bm M}_{k'}
	\widehat{\bm\mu}_b
	=
	m_{\ell,K}^2
	-
	\sum_{m\in\mathcal I_{k'}}
	b_{m,k'}a_{m,k'}m_{\ell,K}\rho_{\ell m}^{(K,k')}m_{m,k'}
	+o_p(1)
	=
	m_{\ell,K}q_{k',\ell}+o_p(1).
	\]
	Hence
	\[
	\widehat{\Theta}_{Kk',\ell}
	=
	(1-b_{\ell,K}\widehat a_{\ell,K})m_{\ell,K}q_{k',\ell}+o_p(1)
	=
	(1-b_{\ell,K}a_{\ell,K})m_{\ell,K}q_{k',\ell}+o_p(1)
	=
	q_{K,\ell}q_{k',\ell}+o_p(1).
	\]
	The case \(k\neq K\), \(k'=K\) is identical and gives
	\[
	\widehat{\Theta}_{kK,\ell}
	=
	q_{k,\ell}q_{K,\ell}+o_p(1).
	\]
	
	Finally, when \(k=k'=K\),
	\[
	\widehat{\Theta}_{KK,\ell}
	=
	\frac{(1-b_{\ell,K}\widehat a_{\ell,K})^2}{\widehat a_{\ell,K}}
	\widehat{\mathcal A}_p(\widehat{\bm P}_{\ell,K}),
	\]
	and by the first display of Lemma~\ref{lem:hetero_mixed_projectors},
	\[
	\frac{1}{\widehat a_{\ell,K}}
	\widehat{\bm\mu}_a^\top\widehat{\bm P}_{\ell,K}\widehat{\bm\mu}_b
	=
	m_{\ell,K}^2+o_p(1).
	\]
	Thus
	\[
	\widehat{\Theta}_{KK,\ell}
	=
	(1-b_{\ell,K}\widehat a_{\ell,K})^2m_{\ell,K}^2+o_p(1)
	=
	(1-b_{\ell,K}a_{\ell,K})^2m_{\ell,K}^2+o_p(1)
	=
	q_{K,\ell}^2+o_p(1).
	\]
	This proves \eqref{eq:Theta_consistency_goal} in all cases.
	
	Combining \eqref{eq:AkMkMkprime} and \eqref{eq:Theta_consistency_goal}, we conclude that
	\begin{align*}
		\widehat{\Omega}_{kk',p}^{E}
		&=
		\widehat{\mathcal A}_p(\widehat{\bm M}_k\widehat{\bm M}_{k'})
		+
		\sum_{\ell\in\mathcal I_K}\lambda_{\ell,K}\widehat{\Theta}_{kk',\ell}\\
		&=
		\|\bar{\bm\mu}\|^2
		-
		\bm m_k^\top\bm D_k\bm m_k
		-
		\bm m_{k'}^\top\bm D_{k'}\bm m_{k'}
		+
		\bm m_k^\top\bm B_k\bm C_{kk'}\bm B_{k'}\bm m_{k'}\\
		&\quad
		+
		\sum_{\ell\in\mathcal I_K}\lambda_{\ell,K}q_{k,\ell}q_{k',\ell}
		+o_p(1)\\
		&=
		\Omega_{kk',p}^{E}+o_p(1).
	\end{align*}
	Since \(\widehat\tau_k\to\tau_k\), together with \(\widehat{\beta}_k^{E}=\alpha_k^2+o_p(1)\), this yields
	\[
	\widehat A_{kk',p}^{E}
	=
	A_{kk',p}^{E}+o_p(1).
	\]
	Since \(K\) is fixed, entrywise convergence implies
	\[
	\widehat{\bm A}_p^{E}-\bm A_p^{E}\to \bm 0
	\qquad\text{in probability}.
	\]
	
	It remains to prove the consistency of the empirical optimal weight.
	By Lemma~\ref{lem:ApH_positive}, there exists a constant \(c_E>0\) such that
	\[
	\lambda_{\min}(\bm A_p^{E})\ge c_E
	\]
	for all sufficiently large \(p\). Since
	\(\|\widehat{\bm A}_p^{E}-\bm A_p^{E}\|\to0\) in probability, Weyl's inequality gives
	\[
	\mathbb P\left(\lambda_{\min}(\widehat{\bm A}_p^{E})\ge c_E/2\right)\to1.
	\]
	Thus \(\widehat{\bm A}_p^{E}\) is positive definite with probability tending to one.
	Moreover, on this event,
	\[
	(\widehat{\bm A}_p^{E})^{-1}-(\bm A_p^{E})^{-1}
	=
	(\widehat{\bm A}_p^{E})^{-1}
	(\bm A_p^{E}-\widehat{\bm A}_p^{E})
	(\bm A_p^{E})^{-1},
	\]
	and therefore
	\[
	(\widehat{\bm A}_p^{E})^{-1}-(\bm A_p^{E})^{-1}\to \bm 0
	\qquad\text{in probability}.
	\]
	Hence
	\[
	(\widehat{\bm A}_p^{E})^{-1}\widehat{\bm u}_p^{E}
	-
	(\bm A_p^{E})^{-1}\bm u_p^{E}
	\to \bm 0
	\qquad\text{in probability},
	\]
	and also
	\[
	(\widehat{\bm u}_p^{E})^\top(\widehat{\bm A}_p^{E})^{-1}\widehat{\bm u}_p^{E}
	-
	(\bm u_p^{E})^\top(\bm A_p^{E})^{-1}\bm u_p^{E}
	\to 0
	\qquad\text{in probability}.
	\]
	Since \(\|\bm u_p^{E}\|\ge c_u\) and \(\lambda_{\max}(\bm A_p^{E})=O(1)\), there exists
	a constant \(c'>0\) such that
	\[
	(\bm u_p^{E})^\top(\bm A_p^{E})^{-1}\bm u_p^{E}\ge c'.
	\]
	Thus the denominator in \eqref{eq:emp_weight_hetero} is bounded away from zero with
	probability tending to one. Therefore, by the continuous mapping theorem,
	\[
	\widehat{\bm w}_{p,E}^\ast
	=
	\frac{(\widehat{\bm A}_p^{E})^{-1}\widehat{\bm u}_p^{E}}
	{\sqrt{(\widehat{\bm u}_p^{E})^\top(\widehat{\bm A}_p^{E})^{-1}\widehat{\bm u}_p^{E}}}
	\to
	\frac{(\bm A_p^{E})^{-1}\bm u_p^{E}}
	{\sqrt{(\bm u_p^{E})^\top(\bm A_p^{E})^{-1}\bm u_p^{E}}}
	=
	\bm w_{p,E}^\ast
	\]
	in probability. This completes the proof.
	
\end{document}